\documentclass[onecolumn, unpublished, a4paper]{quantumarticle}
\pdfoutput=1 
\usepackage{amsmath}
\usepackage{amssymb}
\usepackage[dvipsnames]{xcolor}
\usepackage[xcolor]{changes}
\usepackage[colorlinks=true,
            linkcolor=blue,
            citecolor=blue,
            urlcolor=blue]{hyperref}

\usepackage{mleftright}
\mleftright
\usepackage{quantikz}
\usepackage[T1]{fontenc}
\usepackage{bm}
\usepackage{mathrsfs}
\newcommand{\mathbfcal}[1]{\bm{\mathcal{#1}}}
\usepackage[noabbrev]{cleveref}

\newcommand{\tr}{\textnormal{tr}}
\usepackage{thm-restate}
\usepackage{comment}

\makeatletter
\@ifundefined{newcounteralias}{}{%
  \renewcommand\thmt@autorefsetup{%
    \@xa\def\csname\thmt@envname autorefname\@xa\endcsname
    \@xa{\thmt@thmname}%
  }%
}
\makeatother

\newtheorem{theorem}{Theorem}

\newtheorem{corollary}[theorem]{Corollary}

\newtheorem{definition}[theorem]{Definition}

\newtheorem{lemma}[theorem]{Lemma}

\newtheorem{remark}[theorem]{Remark}

\newenvironment{proof}[1][Proof]{\noindent\textbf{#1.} }{\ \rule{0.5em}{0.5em}}

\newcommand{\ketbra}[2]{|#1\rangle\!\langle#2|}

\usepackage{tikz}
\usepackage{subcaption}
\usetikzlibrary{
positioning,
calc,
arrows,
matrix,
decorations.pathreplacing,
decorations.markings,
backgrounds
}
\usepackage{tikz-3dplot}
\usepackage{braket}
\usepackage{graphicx}
\usepackage{mdwlist}
\usepackage{color}
\usepackage{bbold}

\usetikzlibrary{decorations.pathreplacing,decorations.markings}
\usetikzlibrary{quantikz2}

\tikzset{
    >=stealth',
    punkt/.style={
           rectangle,
           rounded corners,
           draw=black, very thick,
           text width=6.5em,
           minimum height=2em,
           text centered},
    pil/.style={
           ->,
           thick,
           shorten <=2pt,
           shorten >=2pt,},
  on each segment/.style={
    decorate,
    decoration={
      show path construction,
      moveto code={},
      lineto code={
        \path [#1]
        (\tikzinputsegmentfirst) -- (\tikzinputsegmentlast);
      },
      curveto code={
        \path [#1] (\tikzinputsegmentfirst)
        .. controls
        (\tikzinputsegmentsupporta) and (\tikzinputsegmentsupportb)
        ..
        (\tikzinputsegmentlast);
      },
      closepath code={
        \path [#1]
        (\tikzinputsegmentfirst) -- (\tikzinputsegmentlast);
      },
    },
  },
  mid arrow/.style={postaction={decorate,decoration={
        markings,
        mark=at position .5 with {\arrow[#1]{stealth'}}
      }}}
}

\tikzset{%
    mycontrol/.style={draw,circle,fill=black,inner sep=0.5mm},
    myoffcontrol/.style={draw,circle,fill=white,inner sep=0.5mm},
    mySQG/.style={draw=black,rectangle,fill=white,minimum size=0.6cm,inner sep=0.2mm},
    mytarget/.style={inner sep=-0.2mm},
    mymeasurement/.style={circle,draw,scale=1,inner sep=0.2mm},
    mymeasureselect/.style={shape=rounded rectangle, rounded rectangle left arc=none,draw,fill=white,scale=1,inner sep=0.2mm,minimum size=0.6cm},
    myMQG/.style n args={1}{
    draw,fill=white,fit=#1,minimum width=1.2cm,inner ysep=0.3cm
    },
    myMQM/.style n args={1}{
    draw,fill=white,shape=rounded rectangle, rounded rectangle left arc=none,fit=#1,minimum size=1.2cm,inner ysep=0.3cm},
    every edge/.append={thick},
}
\tikzset{Bell/.pic={[pic actions]
    \path 
    (0,0)           coordinate (-i1)
    ++(0,-1.5cm)      coordinate (-i2)
    ++(0,-1.5cm)      coordinate (-i3)
    (-i1)++(4cm,0)   coordinate (-m1)
    (-i2)++(4cm,0)   coordinate (-m2)
    (-i3)++(4cm,0)   coordinate (-m3)
    (-i1) edge (-m1)
    (-i2) edge (-m2)
    (-i3) edge (-m3)
    (-i2) edge[in=180,out=180] (-i3)
    node[myMQM=(-m1)(-m2),anchor=west] (-M) {}
    node[mymeasureselect,anchor=west] at (-m3) (-cM) {}
    ;
}
}
\tikzset{Bell0/.pic={[pic actions]
    \path 
    (0,0)           coordinate (-i1)
    ++(0,-1.5cm)      coordinate (-i2)
    ++(0,-1.5cm)      coordinate (-i3)
    (-i1)++(4cm,0)   coordinate (-m1)
    (-i2)++(4cm,0)   coordinate (-m2)
    (-i3)++(4cm,0)   coordinate (-m3)
    (-i1) edge (-m1)
    (-i2) edge (-m2)
    (-i3) edge (-m3)
    (-i2) edge[in=180,out=180] (-i3)
    node[mymeasureselect,anchor=west] at (-m1) (-cM1) {}
    node[mymeasureselect,anchor=west] at (-m2) (-cM2) {}
    node[mymeasureselect,anchor=west] at (-m3) (-cM3) {}
    ;
}
}

\tikzset{Bell1/.pic={[pic actions]
    \path 
    (0,0)           coordinate (-i1)
    ++(0,-1.5cm)      coordinate (-i2)
    ++(0,-1.5cm)      coordinate (-i3)
    (-i1)++(4cm,0)   coordinate (-m1)
    (-i2)++(4cm,0)   coordinate (-m2)
    (-i3)++(4cm,0)   coordinate (-m3)
    (-i1) edge (-m1)
    (-i2) edge (-m2)
    (-i3) edge (-m3)
    (-i2) edge[in=180,out=180] (-i3)
    (-m1) edge[in=0,out=0] (-m2)
    node[mymeasureselect,anchor=west] at (-m3) (-cM3) {}
    ;
}
}
\tikzset{Bell1b/.pic={[pic actions]
    \path 
    (0,0)           coordinate (-i1)
    ++(0,-1.5cm)      coordinate (-i2)
    ++(0,-1.5cm)      coordinate (-i3)
    (-i1)++(4cm,0)   coordinate (-m1)
    (-i2)++(4cm,0)   coordinate (-m2)
    (-i3)++(8cm,0)   coordinate (-m3)
    (-i1) edge (-m1)
    (-i2) edge (-m2)
    (-i3) edge (-m3)
    (-i2) edge[in=180,out=180] (-i3)
    (-m1) edge[in=0,out=0] (-m2)
    node[mymeasureselect,anchor=west] at (-m3) (-cM3) {}
    ;
}
}
\tikzset{Bell1c/.pic={[pic actions]
    \path 
    (0,0)           coordinate (-i1)
    ++(0,-1.5cm)      coordinate (-i2)
    ++(0,-1.5cm)      coordinate (-i3)
    (-i1)++(4cm,0)   coordinate (-m1)
    (-i2)++(4cm,0)   coordinate (-m2)
    (-i3)++(4cm,0)   coordinate (-m3)
    (-i1) edge (-m1)
    (-i2) edge (-m2)
    (-i3) edge (-m3)
    (-i2) edge[in=180,out=180] (-i3)
    (-m1) edge[in=0,out=0] (-m2)
    node[mymeasureselect,anchor=west,minimum width=1.2cm] at (-m3) (-cM3) {}
    ;
}
}

\tikzset{Bell1d/.pic={[pic actions]
    \path 
    (0,0)             coordinate (-i3)
    (-i3)++(2cm,0)   coordinate (-m3)
    (-i3) edge (-m3)
    node[mymeasureselect,anchor=west,minimum width=1.2cm] at (-m3) (-cM3) {}
    ;
}
}
\tikzset{AlternatingBell/.pic={[pic actions]
    \path 
    (0,0)           coordinate (-i1)
    ++(0,-1.25cm)      coordinate (-i2)
    ++(0,-1.25cm)      coordinate (-i3)
    ++(0,-1.25cm)      coordinate (-i4)
    ++(0,-1.25cm)      coordinate (-i5)
    (-i1)++(4cm,0)   coordinate (-m1)
    (-i2)++(4cm,0)   coordinate (-m2)
    (-i3)++(4cm,0)   coordinate (-m3)
    (-i4)++(4cm,0)   coordinate (-m4)
    (-i5)++(4cm,0)   coordinate (-m5)
    (-i1) edge (-m1)
    (-i2) edge (-m2)
    (-i3) edge (-m3)
    (-i4) edge (-m4)
    (-i5) edge (-m5)
    (-i2) edge[in=180,out=180] (-i4)
    (-i3) edge[in=180,out=180] (-i5)
    node[myMQM=(-m1)(-m2)(-m3),anchor=west] (-M) {}
    node[mymeasureselect,anchor=west] at (-m4) (-cM0) {}
    node[mymeasureselect,anchor=west] at (-m5) (-cM1) {}
    ;
}
}

\tikzset{AlternatingBellwunitary/.pic={[pic actions]
    \path 
    (0,0)           coordinate (-i1)
    ++(0,-1.25cm)      coordinate (-i2)
    ++(0,-1.25cm)      coordinate (-i3)
    ++(0,-1.25cm)      coordinate (-i4)
    ++(0,-1.25cm)      coordinate (-i5)
    (-i1)++(4cm,0)   coordinate (-m1)
    (-i2)++(4cm,0)   coordinate (-m2)
    (-i3)++(4cm,0)   coordinate (-m3)
    (-i4)++(7cm,0)   coordinate (-m4)
    (-i5)++(7cm,0)   coordinate (-m5)
    (-i1) edge (-m1)
    (-i2) edge (-m2)
    (-i3) edge (-m3)
    (-i4) edge (-m4)
    (-i5) edge (-m5)
    (-i2) edge[in=180,out=180] (-i4)
    (-i3) edge[in=180,out=180] (-i5)
    node[myMQM=(-m1)(-m2)(-m3),anchor=west] (-M) {} 
    node[] at (-m4) (-cM0) {} 
    node[] at (-m5) (-cM1) {}
    ;
}
}

\pgfdeclarelayer{front}
\pgfdeclarelayer{bg1}
\pgfdeclarelayer{bg2}
\pgfdeclarelayer{bg3}
\pgfdeclarelayer{bg4}
\pgfdeclarelayer{bg5}
\pgfsetlayers{background,bg5,bg4,bg3,bg2,bg1,main,front}

\title{Equivalence of non-local computation tasks beyond Clifford operations}

\author[1]{Andreas Bluhm}
\email{andreas.bluhm@univ-grenoble-alpes.fr}
\orcid{0000-0003-4796-7633}

\author[1]{Simon H\"{o}fer}
\email{simon.hofer@univ-grenoble-alpes.fr}
\orcid{}

\author[2,3]{Alex May}
\email{amay@perimeterinstitute.ca}
\orcid{0000-0002-4030-5410}

\author[4]{Florian Speelman}
\email{f.speelman@uva.nl}
\orcid{0000-0003-3792-9908}

\author[5]{Philip Verduyn Lunel}
\email{philip.verduyn-lunel@lip6.fr}
\orcid{0000-0001-5419-6027}

\affiliation[1]{Univ. Grenoble Alpes, CNRS, Grenoble INP, LIG}
\affiliation[2]{Perimeter Institute for Theoretical Physics}
\affiliation[3]{Institute for Quantum Computing, Waterloo, Ontario}
\affiliation[4]{QuSoft, University of Amsterdam}
\affiliation[5]{Sorbonne Universit\'e, CNRS, LIP6, Paris}

\begin{document} 

\abstract{Non-local quantum computation (NLQC) studies how two collaborating players can implement channels on distributed systems using a single simultaneous round of quantum communication and shared entanglement. 
NLQC has applications in diverse areas, ranging from quantum position-verification to quantum gravity. 
Recently, it has been realized that the relationships among families of NLQC tasks are highly structured: many seemingly distinct tasks are related by reductions, wherein implementations of one task can be used to efficiently implement a second task. 
This is analogous to the notion of reduction in complexity theory, and reveals the relative hardness of NLQC tasks. 
In this work we continue the study of reductions among NLQC tasks. 
We focus on NLQC examples of the greatest interest in quantum position-verification; in particular examples involving large classical inputs and fixed-size quantum inputs, since these constitute the most feasible protocols for position-verification schemes.
Within this setting, we find many new relationships among NLQC tasks. 
For instance, protocols for the simplest example of redirecting a quantum system based on a classical control imply protocols for controlled single qubit measurements in arbitrary bases, the controlled application of any Clifford unitary, and even the controlled application of any unitary of the form $U=C_1DC_0$ with $D$ an arbitrary diagonal unitary and $C_0, C_1$ Clifford circuits. 
This implies that many feasible position-verification schemes have the same asymptotic scaling for their entanglement cost, and hence a similar level of security. 
Our techniques rely on ideas from gate teleportation and measurement based quantum computation, among other areas, bringing several new strategies into NLQC which may be of independent interest. 
}

\maketitle

\tableofcontents

\section{Introduction}\label{sec:introduction}

Suppose two parties, Alice and Bob, hold quantum systems $A$ and $B$ at distant locations and wish to implement a joint quantum channel $\mathcal{N}_{AB\rightarrow A'B'}$ on their systems (\cref{fig:local}).
This would be easy if they would be able to freely exchange messages, but what if they are allowed only a single simultaneous round of communication, as in \cref{fig:non-localcomputation}?
This task is known as non-local quantum computation (NLQC), and it turns out that, surprisingly, any quantum channel can be implemented in this form \cite{buhrman2014position,beigi2011simplified}.
However, the known general constructions require that Alice and Bob pre-share an exponentially large entangled state; we are interested in understanding how much entanglement is necessary to implement given channels as NLQC.

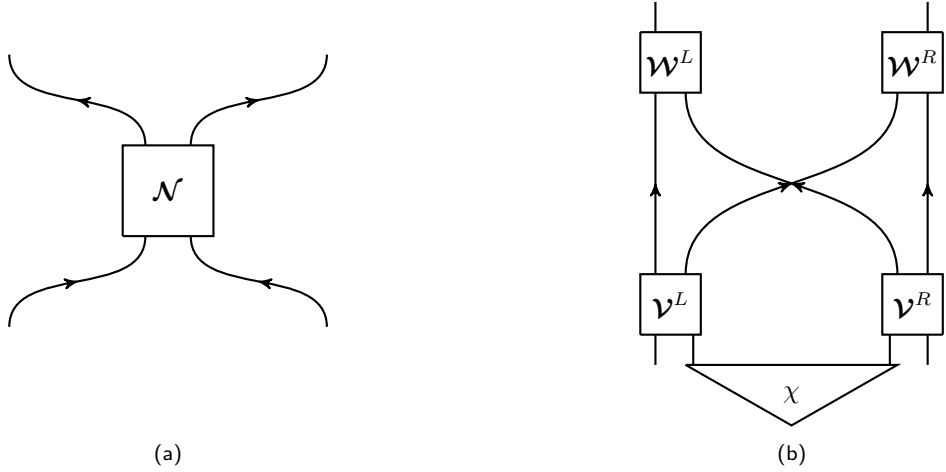
\begin{figure}
    \centering
    \begin{subfigure}{0.45\textwidth}
    \centering
    \begin{tikzpicture}[scale=0.6]
                
    \draw[thick] (-1,-1) -- (-1,1) -- (1,1) -- (1,-1) -- (-1,-1);
                
    \draw[thick,mid arrow] (-3.5,-3) to [out=90,in=-90] (-0.5,-1);
    \draw[thick,mid arrow] (3.5,-3) to [out=90,in=-90] (0.5,-1);
                
    \draw[thick,mid arrow] (0.5,1) to [out=90,in=-90] (3.5,3);
    \draw[thick,mid arrow] (-0.5,1) to [out=90,in=-90] (-3.5,3);
                
    \node at (0,0) {$\mathbfcal{N}$};
                
    \node at (0,-5) {$ $};
                
    \end{tikzpicture}
    \caption{}
    \label{fig:local}
    \end{subfigure}
    \hfill
    \begin{subfigure}{0.45\textwidth}
    \centering
    \begin{tikzpicture}[scale=0.4]
                
    \draw[thick] (-5,-5) -- (-5,-3) -- (-3,-3) -- (-3,-5) -- (-5,-5);
    \node at (-4,-4) {$\mathbfcal{V}^L$};
                
    \draw[thick] (5,-5) -- (5,-3) -- (3,-3) -- (3,-5) -- (5,-5);
    \node at (4,-4) {$\mathbfcal{V}^R$};
                
    \draw[thick] (5,5) -- (5,3) -- (3,3) -- (3,5) -- (5,5);
    \node at (4,4) {$\mathbfcal{W}^R$};
                
    \draw[thick] (-5,5) -- (-5,3) -- (-3,3) -- (-3,5) -- (-5,5);
    \node at (-4,4) {$\mathbfcal{W}^L$};
                
    \draw[thick,mid arrow] (-4.5,-3) -- (-4.5,3);
                
    \draw[thick,mid arrow] (4.5,-3) -- (4.5,3);
                
    \draw[thick,mid arrow] (-3.5,-3) to [out=90,in=-90] (3.5,3);
                
    \draw[thick,mid arrow] (3.5,-3) to [out=90,in=-90] (-3.5,3);
                
    \draw[thick] (-3.5,-6) -- (3.5,-6) -- (0,-8) -- (-3.5,-6);
    \draw[thick] (-3.25,-6) -- (-3.25,-5);
    \draw[thick] (3.25,-6) -- (3.25,-5);
    \node at (0,-7) {$\chi$};
                
    \draw[thick] (-4.5,-6) -- (-4.5,-5);
    \draw[thick] (4.5,-6) -- (4.5,-5);
                
    \draw[thick] (4.5,5) -- (4.5,6);
    \draw[thick] (-4.5,5) -- (-4.5,6);
                
    \end{tikzpicture}
    \caption{}
    \label{fig:non-localcomputation}
    \end{subfigure}
    \caption{(a) Local implementation of a channel $\mathbfcal{N}$. (b) Non-local quantum computation. $\mathbfcal{V}^L$, $\mathbfcal{V}^R$, $\mathbfcal{W}^L$, and $\mathbfcal{W}^R$ are channels and $\chi$ is a resource state. Figure reproduced from \cite{bluhm2025complexity}.}
    \label{fig:non-localandlocal}
\end{figure}

Historically, ideas very similar to NLQC were first considered in the context of understanding which measurements on distributed systems could be made instantaneously \cite{vaidman2003}. 
Since then, NLQC has appeared in many settings, including in
quantum position-verification (QPV) \cite{kent2011quantum, buhrman2014position}, quantum gravity \cite{may2019quantum,may2020holographic,may2022connected,may2022complexity}, information-theoretic classical and quantum cryptography \cite{allerstorfer2024relating,asadi2025conditional,girish2026new,girish2026comparing}, uncloneable secret sharing \cite{ananth2024unclonable}, communication complexity \cite{girish2026magic}, and other settings \cite{yu2012fast,apel2024security}. 
In all of these settings, understanding the entanglement cost of NLQC is central to the application. 
In this work, we focus on those NLQC examples which are of the greatest interest in the context of quantum position-verification, specifically settings where a large classical input is combined with a small quantum input. 

In the simplest QPV setting, two verifiers interact remotely with a prover to try to establish the physical location of the prover. 
If the prover is acting honestly, they implement a desired channel locally, as in \cref{fig:local}. 
If the prover cheats, however, they will consist of a coalition of agents, none at the honest location, who together attempt to replace the local channel with an NLQC implementing the same transformation. 
A useful QPV scheme then involves a choice of channel which is easy to implement locally (by an honest prover) but hard to implement nonlocally (by a cheating prover). 
To quantify how hard an NLQC is in this context, one is typically interested in how much entanglement must be used in the NLQC; since entanglement is difficult to prepare over long distances and then store, we expect schemes where a cheating player needs large entanglement to be secure. 

In this work, we explore the hardness of NLQC tasks and the corresponding security of QPV schemes by following an approach begun in \cite{bluhm2025complexity}, which involves relating the hardness of NLQC families to one another. 
We make use of a notion of reduction among NLQC families, which lets us compare their relative hardness. 
This elucidates a rich structure among NLQC tasks, and among bipartite channels more generally. 

To better motivate this approach to NLQC, it is useful to consider the best studied NLQC family, known as $f$-routing. 
The $f$-routing task involves classical inputs $x\in\{0,1\}^n$ and $y\in\{0,1\}^n$ given to Alice and Bob respectively, along with a constant-size quantum system $Q$ given to Alice. 
The goal is to bring $Q$ to Bob if $f(x,y)=1$, or keep $Q$ with Alice if $f(x,y)=0$. 
Despite most of the inputs being classical, all known NLQC implementations of this task require entanglement that grows with $n$.
This is advantageous from the perspective of QPV, since it means that an honest player only needs to implement mostly classical operations along with an $O(1)$ size quantum operation, while a dishonest player seems to need growing quantum entanglement to attack the scheme. 
Indeed, this can even be proven under certain restrictions on the protocol~\cite{bluhm2022single}, proven when assuming the protocol is perfectly correct~\cite{asadi2024rank}, or proven when taking the quantum resource to be quantum gates rather than entanglement~\cite{asadi2025linear}.
Another approach was pursued in~\cite{unruh2014quantum}, which considers a random oracle model and lower bounds the needed oracle calls. 

Despite these features, $f$-routing still faces certain security limitations.
For instance, while it was initially hoped that in the worst case $f$-routing would require $2^{O(n)}$ EPR pairs of entanglement, in fact $2^{O(\sqrt{n\log n})}$ suffices \cite{allerstorfer2024relating}. 
It is not known if this is the worst-case cost of $f$-routing, or if further improvements are possible.
Further, the status of lower bounds is quite weak, with the best bounds being linear and only holding under extra assumptions, as noted above. 
As well, there are upper bounds on entanglement based on various measures of the complexity of $f$, including the formula size of $f$ and the span program size \cite{buhrman2013garden,cree2023code}. 
These upper bounds provide efficient attacks for certain classes of functions. 
Further, these upper bounds mean that proving super-linear lower bounds (as would be desirable for improved security) is unlikely in the near term, since proving (for instance) even a quadratic lower bound on span programs is a difficult open problem in complexity theory. 
Given this status it is natural to explore other possible protocols, which may avoid these issues. 
For instance, perhaps one can give a quantum system to Bob as well, call it $Q'$, and take the task to be implementing $U_{QQ'}^{f(x,y)}$ (implement $U$ classically controlled on the value of $f(x,y)$). Or, many other variations are possible; see e.g. \cite{lau2011insecurity,chakraborty2015practical} for earlier considerations of such protocols. We can ask if perhaps some variation which is not too much harder to implement locally will evade the limitations of $f$-routing. 

To address this, we use the notion of reduction among NLQC families, and gain more understanding of which families are all related to $f$-routing. Our main contribution in this work is to show that in fact many new examples are no harder as NLQCs, and hence no more secure, than $f$-routing. This includes showing that all $f$-controlled single-qubit measurement protocols can be reduced to $f$-routing. Even implementing an $f$-controlled unitary acting on $O(1)$ qubits of the form $U=C_1DC_0$ with $D$ an arbitrary diagonal unitary and $C_0, C_1$ Clifford can be reduced to $O(1)$ copies of $f$-routing.
 
The $f$-routing task amounts to a classical controlled SWAP operation. 
Surprisingly, although SWAP is Clifford, we find that many classically controlled non-Clifford operations can be reduced to the SWAP.
Indeed, earlier upper bounds suggest non-Clifford operations are much harder than Clifford operations to perform as NLQCs \cite{speelman2015instantaneous}.

Along the way to proving our reductions, we establish several new techniques for implementing NLQCs which may be of independent interest. 
In particular, we exploit gate-teleportation like gadgets in \cref{sec:angleaddition}, and ideas from measurement based quantum computation in \cref{sec:controlleddiagonal}. 
We hope that these new reduction techniques may find further use in upper bound strategies for NLQC. 

\subsection{Previous work}\label{sec:previouswork}

We focus on previous work on NLQC tied to how different examples of NLQC are related. 
A broader review of NLQC can be found in \cite{may2026entanglement}. 

A similar notion of reduction as discussed here first appeared in \cite{allerstorfer2024relating}, where examples of NLQC were related to primitives studied in information-theoretic cryptography. 
There, $f$-routing was related to a primitive known as conditional disclosure of secrets (CDS) \cite{gertner1998protecting,gay2015communication,applebaum2021placing}, and another NLQC family known as coherent function evaluation was related to private simultaneous message passing \cite{feige1994minimal}. 
These relationships were not true reductions in the sense we define here. 
In particular, it was shown that the existence of an $f$-routing protocol implies the existence of a similarly efficient CDS protocol for the same choice of Boolean function, but that CDS protocols imply $f$-routing protocols only when using the purification of the CDS scheme's resource state. 
In contrast, our notion of reduction requires use of the same resource state in both protocols. 
However, this relationship was already strong enough to imply several new results for $f$-routing and CDS, and motivated the study of relationships among NLQC examples outside the context of information-theoretic cryptography. 

The notion of reduction among NLQC tasks was first formalized in \cite{bluhm2025complexity}.
There, a family of NLQC tasks $\mathcal{F}=\{F_n\}$ is said to be reducible to another family $\mathcal{G}=\{G_n\}$ if $F_n$ can be implemented using a small number of copies of the resource state for $G_n$.
They also consider a second, stronger notion of reduction. 
In an \emph{oracle reduction}, implementations of $G_n$ are accessed as a black box and used to implement $F_n$, perhaps along with some additional small shared resources.\footnote{We review notions of reduction in detail in the main text.}
We also say that $G_n$ implies $F_n$ when $F_n$ can be reduced to $G_n$.
The key tasks related in that work were $f$-measure and $f$-route. The
$f$-route task we reviewed above. 
The $f$-measure task involves a single qubit quantum input on Alice's side in the state $H^{f(x,y)}\ket{b}$, $x\in\{0,1\}^n$ given to Alice, and $y\in\{0,1\}^n$ given to Bob. 
The goal is for $b$ to be produced by both Alice and Bob at the end of the protocol. 
$f$-measure is another well studied NLQC example because of its favourable properties as a QPV protocol \cite{buhrman2014position,asadi2024rank,asadi2025linear,bluhm2022single,escola2025quantum,allerstorfer2025making, lau2011insecurity}, and it (or variations thereof) is being studied experimentally \cite{kavuri2025device,fan2026relativistic,kanneworff2025towards}.
The work \cite{bluhm2025complexity} showed that $f$-route implies $f$-measure using an oracle reduction. 
We state this as the following theorem.

\begin{theorem} \label{thm:frouteTOfmeasure}
    \textbf{(From \cite{bluhm2025complexity}) $f$-route $\Rightarrow$ $f$-measure:} An $\epsilon$-correct $f$-route protocol for function $f$ using resource system $\Psi$ implies the existence of an $\epsilon$-correct $f$-measure protocol for the same function using resource system $\Psi^{\otimes 2}$.
\end{theorem}
More precisely, inspecting the proof of this theorem gives an oracle implication from one instance of $f$-route and one instance of $\neg f$-route to $f$-measure; we define oracle reductions and implications in \cref{sec:oraclereductions}. 

As a partial converse to this result, \cite{bluhm2025complexity} showed that access to a few copies of \emph{purifications} of the $f$-measure resource state suffices to implement $f$-route.
This does not quite satisfy the formal notion of a reduction, as we define it here.
One consequence of our work here is an equivalence between $f$-route and $f$-measure under this stricter notion of reduction (which does not allow purifications), and in fact under the even stronger notion of oracle reductions. 

\subsection{Our results}\label{sec:ourresults}

The most basic object of study in this paper is a $2\rightarrow 2$ task, which involves two input systems $A,B$ and two output systems $A',B'$. 
Systems $A,A'$ are held by Alice and $B,B'$ by Bob. 
We are interested in comparing the relative hardness of such tasks, and are specifically interested in how hard they are to implement in the form of an NLQC. 

To understand this, we define the notion of a reduction between $2\rightarrow 2$ tasks. 
Reductions come in two forms, resource state reductions and oracle reductions. 
In a resource state reduction, task 1 reduces to task 2 if task 1 can be implemented using a small number of copies of any shared resource state that allows task 2 to be implemented as an NLQC, plus some small additional resources (see \cref{def:reduction} for a formal statement). 
In an oracle reduction, task 1 reduces to task 2 if task 1 can be implemented using a few copies of any channel that implements task 2, plus some small additional resources. 
Resource state reductions are clearly a weaker notion of reduction than oracle reductions. 
However, if we are interested mainly in the shared entanglement resources needed to implement an NLQC, establishing a resource state reduction suffices to compare the hardness of the tasks. 

Meanwhile, oracle reductions imply resource state reductions, but also say more. 
An oracle reduction says that the two tasks are related in how hard they are to implement when viewed as abstract interactions. 
Since all of the reductions we prove in this article are in fact oracle reductions, we can interpret them both as statements about the relative hardness of these tasks as NLQCs, but also their relative hardness as interactions.

\begin{figure}
\centering
\makebox[\textwidth][c]{
\begin{tabular}{|p{3.5cm}|p{6.5cm}|p{5.5cm}|}
     \hline 
     Task & Definition & Implied by \\
     \hline 
     $f$-route & Redirect a quantum system (\cref{def:frouting}) & Implied by $f$-measure$(I,H)$ (\cref{cor:routereducedtomeasureH})\\
     \hline
     $f$-measure$(I,H)$ &  Measure in either computational or Hadamard basis (special case of \cref{def:f-measure}) & Implied by $f$-route plus $\neg f$-route (\cref{thm:frouteTOfmeasure}) \\
     \hline 
     $f$-measure$(U_1,V_1)$, single qubit $U_1,V_1$ & Measure in either basis $\{U_1\ket{0},U_1\ket{1}\}$ or basis $\{V_1\ket{0}, V_1\ket{1}\}$ (special case of \cref{def:f-measure}) & Implied by any other non-trivial single qubit $f$-measure$(U_2,V_2)$ (\cref{thm:fmeasureEquivalence}) \\
     \hline 
     $f$-measure(Bell) & Measure in either the computational basis or the Bell basis (\cref{def:fBell}) & Implied by $f$-measure$(I,H)$ (\cref{lem:fBB84tofBell})\\
     \hline 
     $f$-measure(Clifford) & Measure in either the basis $\{C_0\ket{i}\}_i$ or $\{C_1\ket{i}\}_i$ for (multi-qubit) $C_0,C_1$ Clifford (special case of \cref{def:f-measure}) & Implied by $f$-measure(Bell) (\cref{thm:cliffordmeasure}) \\
     \hline 
     $f$-unitary(Clifford) & Apply one of two Clifford unitaries (special case of \cref{def:f-unitary}) & Implied by $f$-measure(Bell) (\cref{thm:fcliffordunitary})\\
     \hline 
     $f$-unitary$(I,C_1DC_0)$ & Apply either the identity, or $U=C_1DC_0$ with $D$ diagonal, $C_1,C_0$ Clifford (special case of \cref{def:f-unitary}) & Oracle reduced to $f$-measure$(I,H)$ via $f$-measure$(I,R_X(\theta))$ (\cref{thm:beyondClifford}) \\
     \hline 
\end{tabular}
}
\caption{Summary of the $2\rightarrow 2$ tasks we study and the implications we prove. All implications are oracle implications.}\label{fig:reductiontable}
\end{figure}

The set of oracle reductions we establish among $2\rightarrow 2$ tasks is shown in \cref{fig:reductiontable}. 
To understand them, a useful starting point is $f$-routing, which we introduced above. 
Our new reductions significantly extend the set of tasks which can be reduced to $f$-routing, and consequently extend the class of tasks which can be implemented efficiently as NLQCs. 
In this paper we focus on tasks which involve implementing a small quantum operation controlled off of the value of a classical Boolean function $f:\{0,1\}^n\times \{0,1\}^n\rightarrow \{0,1\}$. 
We give Alice $x\in\{0,1\}^n$ and Bob $y\in \{0,1\}^n$. 
As discussed above, these are the examples of the most interest for QPV. 

As a starting point, consider the task which requires that we measure a single qubit in one of two bases, which we take to be $\{U\ket{i}\}_i$ and $\{V\ket{i}\}_i$, with the measurement basis fixed by $f(x,y)$. 
The outputs from the task should be the measurement outcome copied to both Alice and Bob's lab. 
We call this task $f$-measure$(U,V)$. 
Note that similar tasks have been studied since QPV was first defined and studied, see e.g.\ \cite{kent2011quantum, olivo2020breaking, lau2011insecurity}. 
We show that for any fixed $U,V$, there is an $O(1)$ oracle reduction to the case where $U=I,V=H$, with $H$ the Hadamard gate. 
Earlier results (see \cref{thm:frouteTOfmeasure}) reduce the Hadamard case to $f$-routing. 
Conversely, we show that $f$-measure$(I,H)$ can be oracle reduced to $f$-measure$(U,V)$ for all single qubit $U,V$ except when $U=V$.
This is proven as the following theorem. 

\begin{restatable}{theorem}{fmeasureEquivalence}\label{thm:fmeasureEquivalence}
    Consider two non-trivial single qubit $f$-measure protocols\footnote{We define what non-trivial means here in the main text. As an example, $f$-measure$(U,V)$ with $U=V$ is trivial, as is $U=I,V=R_Z(\theta)$, since $R_Z(\theta)$ leaves the measurement basis elements $\{\ketbra{0}{0}, \ketbra{1}{1}\}$ unchanged under conjugation.} \\ $f$-measure$(U_1,V_1)$ and $f$-measure$(U_2,V_2)$.
    Then they are equivalent under $O(1)$ oracle reductions. 
\end{restatable}
This theorem shows that all such protocols are not more secure as QPV protocols than $f$-measure or $f$-route, and furthermore that we can study their security by considering the simplest such example, $f$-measure$(I,H)$. 

Moving beyond single qubit measurements, we also show that the classically controlled application of (multi-qubit) Clifford unitaries, or measurements in bases related by Cliffords, can also be reduced to $f$-measure$(I,H)$. 
To do this, an intermediate step is to show that $f$-measure$(I,B)$ where $B$ is the unitary mapping from the computational to the Bell basis can be reduced to $f$-measure$(I,H)$. 
We prove this as \cref{lem:fBB84tofBell} in the main text. 
With this in hand, we are then able to apply controlled Cliffords by using gate-teleportation strategies. 
This leads to the following theorem. 

\begin{restatable}{theorem}{CliffordMeasure}\label{thm:cliffordmeasure}
    \textbf{$f$-Bell $\Rightarrow$ $f$-measure(Clifford):} There is an $O(n_Q)$ oracle implication from $f$-Bell to $f$-measure$(C_0,C_1)$, for any pair of Cliffords $C_0,C_1$ acting on $n_Q$ qubits.
\end{restatable}

An $f$-controlled Clifford is a plausible QPV candidate for realistic implementation: the honest player needs only to implement classical computations, plus a small number of Clifford operations. 
Clifford operations are the easiest to perform in an error-corrected way. 
Thus the above theorem shows that a large class of plausible QPV candidates reduce to the simplest scheme, $f$-measure$(I,H)$. 

Since $f$-measure has some limitations (for instance subexponential attacks for all functions), we can ask if this can be remedied, at the expense of making the quantum operation more general. 
We give a partial answer in the negative with the following theorem. 

\begin{restatable}{theorem}{beyondClifford}\label{thm:beyondClifford}
    Let $U$ be any $n_Q$ qubit unitary of the form $U=C_1 DC_0$, where $D$ is a diagonal unitary and $C_0,C_1$ are Cliffords. Then, there is an $O(2^{n_Q})$ oracle implication from $f$-measure$(I,H)$ to $f$-unitary$(I,U)$. 
\end{restatable}

This result is surprising in that the unitary $D$ can be high $T$-depth, so that $U=C_1 DC_0$ is in some sense highly non-Clifford, but nonetheless implementing it can be oracle reduced to $f$-measure$(I,H)$, which involves only (classically controlled) Clifford operations. In particular, for applications in position verification, where we typically want $n_Q = O(1)$, we see that the scaling with respect to the classical input size $n$ is no larger for unitaries of the form $f$-unitary($I,C_1 DC_0$) than for $f$-measure$(I,H)$.  

We leave open whether all classically controlled $n_Q$-qubit unitaries can be $O(F(n_Q))$ reduced to $f$-measure for some function $F=F(n_Q)$ (not dependent on $n$, the classical input size), or if perhaps two layers of diagonal gates suffice to obtain a higher level of security in QPV schemes. 

\vspace{0.2cm}
\noindent \textbf{Acknowledgements:} 
ChatGPT was used to write code that helped find earlier versions of the circuits used in \cref{lemma:anglesplitting}.
Research at Perimeter Institute is supported in part by the Government of Canada through the Department of Innovation, Science and Economic Development Canada and by the Province of Ontario through the Ministry of Colleges and Universities.
AB was supported by the ANR project PraQPV, grant number ANR-24-CE47-3023. FS was supported by the Dutch Ministry of Economic Affairs and Climate Policy (EZK), as part of the Quantum Delta NL program, and the project Divide and Quantum `D\&Q' NWA.1389.20.241 of the program `NWA-ORC', which is partly funded by the Dutch Research Council (NWO). PVL is supported by France 2030 under the French National Research Agency award number ANR-22-PETQ-0007.
SH is supported by the program QuanTEdu-France n°ANR-22-CMAS-0001 France 2030 and by the MSCA Cofund QuanG (Grant Number: 101081458), funded by the European Union \scalebox{0.05}{\includegraphics{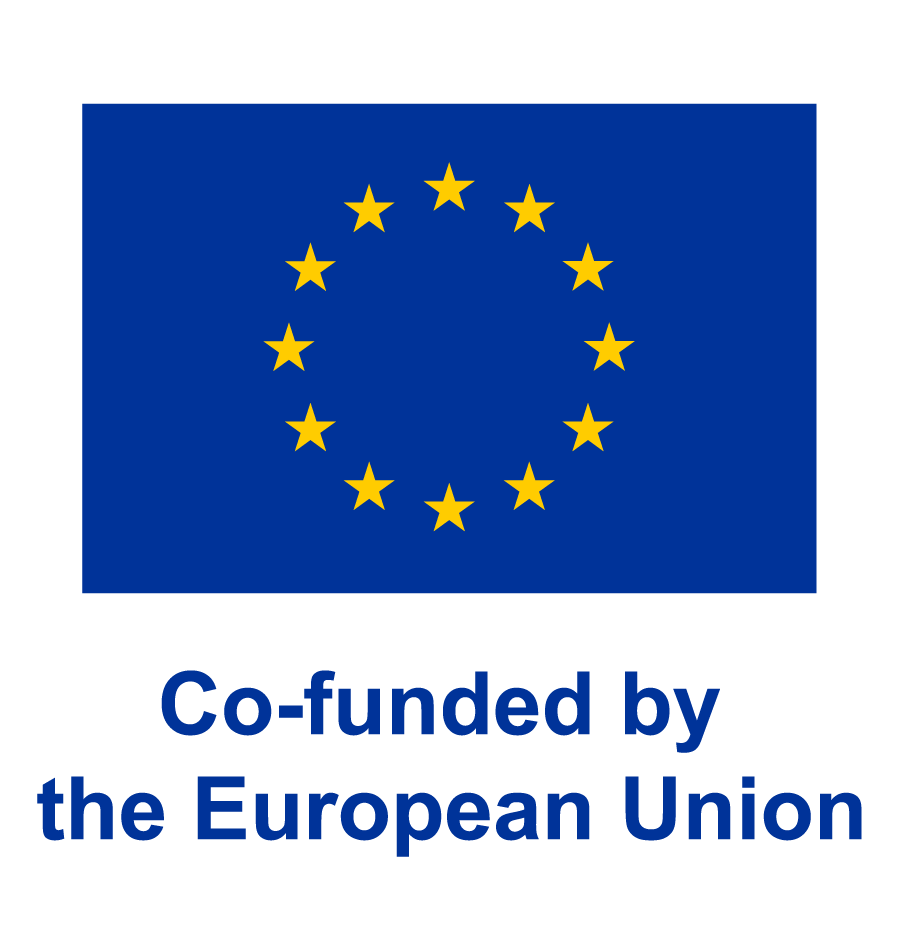}}.

\section{Preliminaries}\label{sec:preliminaries}

\subsection{Tools from quantum information theory}\label{sec:QItools}

We let $d_A$ denote the dimension of Hilbert space $\mathcal{H}_A$, and $n_A=\log d_A$. 
We always take $\log$ to denote the base 2 logarithm. 
We label the maximally entangled state of two qubits by $\ket{\Psi_{00}}$, and the maximally entangled state of any two $d$ dimensional systems by $\ket{\Psi^+}$. 
We quantify the distance between quantum states with the one-norm distance, 
\begin{align}
    \Vert\rho-\sigma\Vert_1 = \tr|\rho-\sigma|.
\end{align}
Note that $D(\rho,\sigma)=\frac{1}{2}\Vert\rho-\sigma\Vert_1$ is the trace distance. 

We label the identity channel on Hilbert space $A$ by $\mathcal{I}_A$. 
We quantify the distance between quantum channels using the diamond norm distance. 
\begin{definition} Let $\mathbfcal{N}_{B\rightarrow C}, \mathbfcal{M}_{B\rightarrow C}: \mathcal{L}(\mathcal{H}_B)\rightarrow \mathcal{L}(\mathcal{H}_C)$ be quantum channels. 
The \textbf{diamond norm distance} is defined by 
\begin{align}
    \Vert\mathbfcal{N}_{B\rightarrow C}-\mathbfcal{M}_{B\rightarrow C}\Vert_\diamond = \sup_{d} \max_{\Psi_{A_dB}}\Vert\mathbfcal{N}_{B\rightarrow C}(\Psi_{A_dB}) - \mathbfcal{M}_{B\rightarrow C}(\Psi_{A_dB})\Vert_1
\end{align}
where $\mathcal{H}_{A_d}$ is a $d$ dimensional Hilbert space. 
\end{definition}

It will be useful to have the following statement regarding the diamond norm.

\begin{restatable}{lemma}{densitytoDN}\label{lemma:densitytoDN}
    Let $\mathcal{M}$, $\mathcal{N}$ be quantum channels acting on a $d$ dimensional Hilbert space. Suppose that for all density matrices $\rho$, 
    \begin{align}
        \Vert \mathcal{N}(\rho)-\mathcal{M}(\rho)\Vert_1\leq \epsilon.
    \end{align}
    Then we have that
    \begin{align}
        \Vert \mathcal{N}-\mathcal{M}\Vert_\diamond \leq 2\epsilon d.
    \end{align}
\end{restatable}
This is proven in \cref{sec:lemmaproofs}.

Next, we give some circuit identities \cite{Zhou_2000} that will prove useful. 
\begin{itemize}
    \item \textbf{Identity 1:} Moving $Z$-rotations:
    \\
    \begin{center}
    \begin{tikzpicture}[on grid,node distance=10mm]
          \node at (0,0)      (i1) {$\ket{\psi}$};
          \node[below=of i1]  (i2) {$\ket{0}$};
          \path (i1)  ++(3cm,0)     node (o1) {}
                      ++(0,-10mm)   node (o2) {};
          \path   
                  (i1)    ++(1cm,0) node[mycontrol]     (c1) {}
                          ++(0,-1cm) node[mytarget]     (t2) {$\bigoplus$}
                  (c1) edge (t2)
                  (i1)    ++(2cm,0) node[mySQG]         (g1)  {$R_z(\theta)$}
                  ;
      \begin{pgfonlayer}{background}
          \path       (i1) edge (o1)
                      (i2) edge (o2);
      \end{pgfonlayer}
        \path (o1) ++(1cm,-0.5cm) node (e1) {$=$};
          \path (o1) ++(2cm,0cm)  node  (Ri1)   {$\ket{\psi}$};
          \node[below=of Ri1]           (Ri2)   {$\ket{0}$};
          \path (Ri1) ++(3cm,0)      node     (Ro1) {}
                      ++(0,-10mm)    node     (Ro2) {};
          \path   
                  (Ri1)   ++(1cm,0) node[mycontrol]     (Rc1) {}
                          ++(0,-1cm) node[mytarget]     (Rt2) {$\bigoplus$}
                  (Rc1) edge (Rt2)
                  (Ri2)    ++(2cm,0) node[mySQG]         (Rg1)  {$R_z(\theta)$}
                  ;
      \begin{pgfonlayer}{background}
          \path       (Ri1) edge (Ro1)
                      (Ri2) edge (Ro2);
      \end{pgfonlayer}
\end{tikzpicture}
    \end{center}

    \item \textbf{Identity 2:} Measurements and controls: \\
    \begin{center}
    \begin{tikzpicture}[on grid,node distance=10mm]
          \node at (0,0)      (i1) {};
          \node[below=of i1]  (i2) {};
          \path (i1)  ++(2cm,0)     node (o1) {}
                      ++(0,-10mm)   node[mymeasureselect] (m2) {$a$};
          \path   
                  (i1)    ++(1cm,0)  node[mySQG]       (t1) {$U$}
                          ++(0,-1cm) node[mycontrol]     (c2) {}
                  (t1) edge (c2)
                  ;
      \begin{scope}[on background layer]
          \path       (i1) edge (o1)
                      (i2) edge (m2);
      \end{scope}
        \path (o1) ++(1cm,-0.5cm) node (e1) {$=$};
          \path (o1) ++(2cm,0cm)  node  (Ri1)   {};
          \node[below=of Ri1]           (Ri2)   {};
          \path (Ri1) ++(2cm,0)      node     (Ro1) {}
                      ++(0,-1cm)    node[mymeasureselect]     (Rm2) {$a$};
          \path   
                  (Ri1)    ++(1cm,0) node[mySQG]         (Rg1)  {$U^a$}
                  ;
      \begin{scope}[on background layer]
          \path       (Ri1) edge (Ro1)
                      (Ri2) edge (Rm2);
      \end{scope}
\end{tikzpicture}      
    \end{center}

Note that the symbol at the end of the second wire represents a post-selection onto $\ket{a}$.
\item \textbf{Identity 3:} We also highlight a useful special case of identity 2 above,\\
\begin{center}
\begin{tikzpicture}[on grid,node distance=10mm]
          \node at (0,0)      (i1) {};
          \node[below=of i1]  (i2) {};
          \path (i1)  ++(3cm,0)     node (o1) {}
                      ++(0,-10mm)   node[mymeasureselect] (m2) {$a$};
          \path   
                (i1)    ++(1cm,0)  node[mycontrol]       (c1) {}
                        ++(0,-1cm) node[mytarget]        (t2) {$\bigoplus$}
                (c1) edge (t2)
                (i2)    ++(2cm,0) node[mySQG] (g2) {$H$}
                ;
      \begin{scope}[on background layer]
          \path       (i1) edge (o1)
                      (i2) edge (m2);
      \end{scope}
        \path (o1) ++(1cm,-0.5cm) node (e1) {$=$};
          \path (o1) ++(2cm,0cm)  node  (Ri1)   {};
          \node[below=of Ri1]           (Ri2)   {};
          \path (Ri1) ++(3cm,0)      node     (Ro1) {}
                      ++(0,-1cm)    node[mymeasureselect]     (Rm2) {$a$};
          \path   
                (Ri1)    ++(2cm,0)  node[mySQG]       (Rt1) {$Z$}
                        ++(0,-1cm) node[mycontrol]        (Rc2) {}
                (Rt1) edge (Rc2)
                (Ri2)    ++(1cm,0) node[mySQG] (Rg2) {$H$}
                ;
      \begin{scope}[on background layer]
          \path       (Ri1) edge (Ro1)
                      (Ri2) edge (Rm2);
      \end{scope}
        \path (Ro1) ++(1cm,-0.5cm) node (e1) {$=$};
          \path (Ro1) ++(2cm,0cm)  node  (R2i1)   {};
          \node[below=of R2i1]           (R2i2)   {};
          \path (R2i1) ++(2cm,0)      node     (R2o1) {}
                      ++(0,-1cm)    node[mymeasureselect]     (R2m2) {$a$};
          \path   
                (R2i1)    ++(1cm,0) node[mySQG]         (R2g1)  {$Z^a$}
                (R2i2)    ++(1cm,0) node[mySQG] (R2g2) {$H$}
                ;
      \begin{scope}[on background layer]
          \path       (R2i1) edge (R2o1)
                      (R2i2) edge (R2m2);
      \end{scope}
\end{tikzpicture}
\end{center}

This follows from the previous identity by using that $H_B\operatorname{CNOT}_{A\rightarrow B}H_B=CZ_{AB} = CZ_{BA}$. 
\end{itemize}

We will make use of Hoeffding's inequality \cite{hoeffding1963probability}, which we state below. 
\begin{theorem}\label{thm:hoeffding} Let $X_1, \dots ,X_n$ be independent random variables with ranges $a_i\leq X_i \leq b_i$. 
Consider the sum
\begin{align}
    S_n= X_1 + \dots  +X_n.
\end{align}
Then we have that
\begin{align}
    \text{Pr}\left[|S_n-\mathbb{E}[S_n]|\geq t \right] \leq 2\exp\left(- \frac{2t^2}{\sum_{i}(b_i-a_i)^2}\right).
\end{align}
\end{theorem}

\subsection{Non-local quantum computation}\label{sec:NLQC}

We begin by recalling a definition from \cite{bluhm2025complexity}, which establishes what an NLQC is trying to accomplish. 
\begin{definition}
    A \textbf{$2\rightarrow 2$ quantum task} is defined by a pair of input systems $A$, $B$, a pair of output systems $A'$, $B'$, and a set of input/output state pairs $\mathcal{S}=\{ (\rho_{RAB},\sigma_{RA'B'})\}$. We require that there exists a quantum channel $\mathbfcal{M}_{AB \to A'B'}$ such that
    \begin{align}
        \forall \,\, (\rho_{RAB},\sigma_{RA'B'})\in \mathcal{S}, \,\,\,\left(\rho_{RAB},\mathbfcal{M}_{AB\rightarrow A'B'}(\rho_{RAB})\right)\in \mathcal{S}.
    \end{align}
\end{definition}
Note that this definition is more general than requiring a certain quantum channel be implemented. 
The above instead requires just that a channel be implemented that satisfies a given condition, as expressed by the list of allowed input/output pairs. 
The definition does not require a unique output density matrix for a given input density matrix, since input states can appear more than once in the list of allowed input/output pairs. 

We further define \emph{families} of $2\rightarrow 2$ quantum tasks, which are collections of $2\rightarrow 2$ tasks parameterized by a natural number $n$.
The parameter $n$ will correspond to an input size with the exact relation specified in the definition of each family of tasks. 
We will label families of tasks with capital letters $F$, $G$, etc., where these denote sets of tasks, so that $F=\{F_n\}_n$, $G=\{G_n\}_n$, etc. 
As an example, the $f$-route example mentioned in the introduction with a choice of Boolean function family $\{f_n\}_n$ defines a family of $2\rightarrow 2$ tasks, with each member of the family labelled by an element of $\{f_n\}_n$. 

\begin{definition}
    A \textbf{non-local quantum computation (NLQC)} is a channel in the form
    \begin{align}
        \mathcal{N}_{AB\rightarrow A'B'} = (\mathcal{W}^L_{K_aM_a\rightarrow A'}\otimes\mathcal{W}^R_{K_bM_b\rightarrow B'})\circ (\mathcal{V}^L_{AL\rightarrow K_{a}M_b}\otimes \mathcal{V}^R_{RB\rightarrow M_aK_b})\circ (I_{AB}\otimes\mathcal{S}_{\emptyset\rightarrow LR})
    \end{align}
    We refer to $\mathcal{S}_{\emptyset\rightarrow LR}$ as the resource preparation channel and its output, which we label $\Psi_{LR}$, as the resource state. 
    The $\mathcal{V}^L, \mathcal{V}^R$ channels are called the (left and right) first round operations, and the $\mathcal{W}^L, \mathcal{W}^R$ channels are called the (left and right) second round operations.
    
    We say an NLQC is an $\epsilon$-correct implementation of a $2\rightarrow 2$ task $F_n$ if the channel $\mathcal{N}_{AB\rightarrow A'B'}$ implemented as an NLQC is $\epsilon$-close in diamond norm to at least one channel $\mathcal{M}_{AB\rightarrow A'B'}$ relating the input and output states in the definition of the $2\rightarrow 2$ task.
\end{definition}

We will often be interested in settings where the $A$ and $B$ input consist of mostly classical data. 
In this case, it is useful to adopt some further language. 
We let Alice's classical input be $x\in X=\{0,1\}^n$ and Bob's classical input be $y\in Y=\{0,1\}^n$. 
Meanwhile, Alice's quantum input is labelled $Q$ and Bob's quantum input $Q'$; their quantum outputs (if any) are labelled $O$ and $O'$ respectively.
When considering the classical inputs explicitly, we say the task is a classically controlled $2\rightarrow 2$ task.
Every classically controlled $2\rightarrow 2$ task then amounts to a family of channels $\{\mathcal{M}^{x,y}\}_{x,y}$, where the players should implement $\mathcal{M}^{x,y}$ when given inputs $x,y$. 
Meanwhile, the players' actions in the form of an NLQC have the following form. 
\begin{definition}
    \label{def:classicalNLQC}
    A \textbf{classically controlled NLQC} consists of a choice of four sets of quantum channels, $\{\mathcal{V}^{L,x}\}_x, \{\mathcal{V}^{R,y}\}_y, \{\mathcal{W}^{L,x,y}\}_{x,y}, \{\mathcal{W}^{R,x,y}\}_{x,y}$, which define a family of NLQCs according to
    \begin{align}
        \mathcal{N}^{x,y}_{QQ'\rightarrow OO'} =(\mathcal{W}^{L,x,y}_{K_aM_a\rightarrow O}\otimes\mathcal{W}^{R,x,y}_{K_bM_b\rightarrow O'})\circ (\mathcal{V}^{L,x}_{QL\rightarrow K_{a}M_b}\otimes \mathcal{V}^{R,y}_{Q'R\rightarrow M_aK_b})\circ (I_{QQ'}\otimes\mathcal{S}_{\emptyset\rightarrow LR}).
    \end{align}
    We say that a classically controlled $2\rightarrow 2$ NLQC completes a $2\rightarrow 2$ classically controlled task $\epsilon$-correctly if, for all inputs $(x,y)$, we have that $\Vert \mathcal{N}^{x,y}-\mathcal{M}^{x,y} \Vert_\diamond \leq \epsilon$. 
\end{definition}

Usually, we will be interested in classically controlled $2\rightarrow 2$ tasks where the goal is to implement one of just two possible channels, call them $\mathcal{M}^{0}$ and $\mathcal{M}^1$. 
The map from the inputs $(x,y)$ to the choice of channel then is specified by a Boolean function, $f:\{0,1\}^n\times \{0,1\}^n\rightarrow \{0,1\}$. 
This leads to the next definition. 

\begin{definition}
    A \textbf{$f$-controlled} $2 \rightarrow 2$ task is specified by a Boolean function $f:\{0,1\}^n\times \{0,1\}^n\rightarrow \{0,1\}$ and two input/output sets of state pairs $\mathcal{S}^0=\{ (\rho_{RQQ'}^0,\sigma_{ROO'}^0)\}$, $\mathcal{S}^1=\{ (\rho_{RQQ'}^1,\sigma_{ROO'}^1)\}$. 
    We say that a classically controlled NLQC $\{\mathcal{N}^{x,y}\}_{x,y}$ implements the task if 
    \begin{align}
        \forall (x,y) \in f^{-1}(0), &\quad \sigma_{ROO'}^0=\mathcal{N}^{x,y}(\rho^0_{RQQ'}), \nonumber \\
        \forall (x,y) \in f^{-1}(1), &\quad \sigma_{ROO'}^1=\mathcal{N}^{x,y}(\rho^1_{RQQ'}).
    \end{align}
\end{definition}
In words, this definition is saying that if the classical input satisfies $f(x,y)=0$, we require condition $\mathcal{S}^0$ on the quantum inputs and outputs, while if $f(x,y)=1$ we require condition $\mathcal{S}^1$.

\subsection{NLQC reductions}\label{sec:oraclereductions}

In this section we give our definitions of reductions between $2\rightarrow 2$ tasks. 
We introduce two notions of reduction, beginning with a resource state reduction. 

\begin{definition}\label{def:reduction}
    Let $F=\{F_n\}_n$ and $G=\{G_n\}_n$ be families of $2\rightarrow 2$ tasks, and $\alpha$, $\beta$, $\delta$ all be functions of $(n,\lambda,\epsilon)$, with $\delta\rightarrow 0$ as $\lambda\rightarrow \infty, \epsilon\rightarrow 0$.
    Then we say there is an $(\alpha, \beta,\delta)$-\textbf{reduction} from $G$ to $F$ if, for any resource state $\ket{\Psi^{n,\epsilon}}_{LR}$ which can be used to implement $F_n$ at least $\epsilon$-correctly as an NLQC, it is possible to implement $G_n$ $\delta$-correctly using $\alpha$ copies of $\ket{\Psi^{n,\epsilon}}_{LR}$ along with $\beta$ additional qubits of shared resource state.
\end{definition}
When there is a reduction from family $F$ to family $G$, we will also say there is an \textbf{implication} from $G$ to $F$, and write $G \Rightarrow F$.
When $\alpha, \beta$ are both $O(g(n))$, we say we have a \emph{$O(g(n))$ reduction}. 
In practice, we will construct $O(1)$ reductions, by which we mean that $\alpha,\beta$ have no $n$ dependence. 
Note that $\alpha,\beta$ may still depend on $\lambda,\epsilon$. 
The parameter $\lambda$ should be interpreted as parameterizing a family of protocols, which implement $F$ well as we increase $\lambda$. 
For example, when using port-teleportation as a step in a reduction, the port-teleportation is only approximate but becomes exact as the number of ports $N\rightarrow \infty$. 
In this case the parameter $\lambda$ could be taken to be $N$, so that the reduction becomes exact as $\lambda \rightarrow \infty$. 

We should comment on the requirement that $\delta\rightarrow 0$ as $\lambda\rightarrow \infty, \epsilon\rightarrow 0$, which appears in both notions of reduction. 
This requirement ensures that our notion of reduction is non-trivial. 
For instance, consider $f$-routing: without use of any entanglement, it is always possible to bring the input quantum system to the correct side with probability $1/2$, and hence trivially achieve some $\epsilon=\epsilon_0$ correctness parameter. 
Thus, without the $\delta\rightarrow 0$ requirement, any NLQC implies $f$-routing under an $(\alpha,\beta,\delta=\epsilon_0)$ reduction. 
Our definition excludes trivial constructions like this from being considered as reductions. 

We also comment on the requirement that $\delta\rightarrow 0$ in particular in the double limit, $\lambda\rightarrow \infty$ and $\epsilon\rightarrow 0$. 
This is requiring that as the implementations of $F$ become perfect, and we are allowed to increase how many of them we use ($\alpha$ can grow with $\lambda$) or how much additional resources we use ($\beta$ can grow with $\lambda$), the implementation of $G$ becomes perfect. 
In the work \cite{bluhm2025complexity}, there was no parameter $\lambda$ included in the definition of a reduction. 
We have introduced this parameter because this work introduces new kinds of reduction (for instance reductions using port-teleportation), which do not have $\delta\rightarrow 0$ as $\epsilon\rightarrow 0$ alone, but where the resources used in the reduction can be increased in some way to obtain $\delta\rightarrow 0$. 
We have made this more general definition to capture these interesting constructions. 

The notion of reduction given as \cref{def:reduction} focuses on the key resource we want to study in NLQC, which is the shared entangled state. 
In practice, the reductions we find will also satisfy a stricter notion, which we give next.

\begin{definition}\label{def:oraclereduction}
    Let $F=\{F_n\}_n$ and $G=\{G_n\}_n$ be families of $2\rightarrow 2$ tasks, and $\alpha$, $\beta$, $\delta$ all be functions of $(n,\lambda,\epsilon)$, with $\delta\rightarrow 0$ as $\lambda\rightarrow \infty, \epsilon\rightarrow 0$.
    Then we say that there is an $(\alpha,\beta,\delta)$-\textbf{oracle reduction} from $G$ to $F$ if $G$ can be implemented $\delta$-correctly by using $\alpha$ parallel implementations of $F$ along with $\beta$ additional qubits of resource system and communication. 
\end{definition}

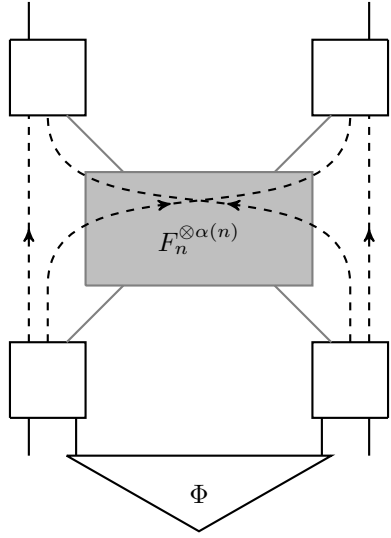
\begin{figure}
    \centering
    \begin{tikzpicture}[scale=0.5]

    \draw[thick, fill=gray,opacity=0.5] (-3,-1.5) -- (3,-1.5) -- (3,1.5) -- (-3,1.5) -- (-3,-1.5);
    
    \draw[thick] (-5,-5) -- (-5,-3) -- (-3,-3) -- (-3,-5) -- (-5,-5);
    
    \draw[thick] (5,-5) -- (5,-3) -- (3,-3) -- (3,-5) -- (5,-5);
    
    \draw[thick] (5,5) -- (5,3) -- (3,3) -- (3,5) -- (5,5);
    
    \draw[thick] (-5,5) -- (-5,3) -- (-3,3) -- (-3,5) -- (-5,5);
    
    \draw[thick,mid arrow,dashed] (-4.5,-3) -- (-4.5,3);
    
    \draw[thick,mid arrow,dashed] (4.5,-3) -- (4.5,3);
    
    \draw[thick,mid arrow, dashed] (-4,-3) -- (-4,-1.5) to [out=90,in=-90] (4,3);
    
    \draw[thick,mid arrow, dashed] (4,-3) -- (4,-1.5) to [out=90,in=-90] (-4,3);

    \draw[thick,gray] (-3.5,-3) -- (-2,-1.5);
    \draw[thick,gray] (3.5,-3) -- (2,-1.5);
    \draw[thick,gray] (3.5,3) -- (2,1.5);
    \draw[thick,gray] (-3.5,3) -- (-2,1.5);
    
    \draw[thick] (-3.5,-6) -- (3.5,-6) -- (0,-8) -- (-3.5,-6);
    \draw[thick] (-3.25,-6) -- (-3.25,-5);
    \draw[thick] (3.25,-6) -- (3.25,-5);
    \node at (0,-7) {$\Phi$};

    \node at (0,-0.25) {$F_n^{\otimes \alpha(n)}$};
    
    \draw[thick] (-4.5,-6) -- (-4.5,-5);
    \draw[thick] (4.5,-6) -- (4.5,-5);
    
    \draw[thick] (4.5,5) -- (4.5,6);
    \draw[thick] (-4.5,5) -- (-4.5,6);
    
    \end{tikzpicture}
    \caption{An oracle reduction from an NLQC $G$ to $F$. The protocol for $G_n$ uses $\alpha$ instances of $F_n$, plus a $\beta$ qubit resource system $\Phi$.}
    \label{fig:oraclereduction}
\end{figure}

The structure of an implementation of $G$ using $F$ as an oracle is shown in \cref{fig:oraclereduction}. 
We will call the local operations before and after use of the oracles \emph{pre-processing} and \emph{post-processing} operations, respectively.

A convenient fact about oracle reductions is that they have nice error-propagation properties. 
This follows from the next remark. 
\begin{remark}\label{remark:diamondadditive}
    Consider two sets of quantum channels $\{\mathcal{N}^i\}$ and $\{\mathcal{M}^i\}$, with the property that $\Vert \mathcal{N}^i-\mathcal{M}^i \Vert_\diamond \leq \epsilon_i$ for $\epsilon_i\in[0,1]$. 
    Then 
    \begin{align}
        \Vert \mathcal{N}^k\circ \dots \circ \mathcal{N}^1 - \mathcal{M}^k\circ \dots \circ \mathcal{M}^1\Vert_\diamond \leq \sum_i \epsilon_i.
    \end{align}
\end{remark}
This remark follows by applying the triangle inequality to the diamond norm. 
A consequence is that if an implementation of $G_n$ using $\alpha$ copies of a perfect $F_n$ oracle works exactly, then an $\epsilon$-correct implementation of $F_n$ gives an $\alpha \epsilon$-correct implementation of $G_n$ immediately by the above. 
We use this repeatedly in the rest of the article.

In this paper, we are primarily interested in how the difficulty of $f$-controlled $2\rightarrow 2$ tasks relates to the size of the classical input and choice of Boolean function $f$. 
We show in appendix \ref{sec:PTtrick} that this difficulty does not strongly depend on the configuration of the input quantum systems for the task in the following sense.

\begin{remark}\label{remark:redistributequantuminputs}
Consider an $f$-controlled $2\rightarrow 2$ task with quantum input $Q$ on the left and quantum input $Q'$ on the right. 
Consider a new task with the same inputs and outputs, but which has quantum inputs $A$ on the left and $B$ on the right, where $\mathcal{H}_A\otimes \mathcal{H}_B=\mathcal{H}_Q\otimes \mathcal{H}_{Q'}$. 
Then there is always an $(\alpha,\beta,\delta)$-oracle equivalence between these tasks, where $\alpha,\beta$ are of order $2^{O(m)}/\delta^2$ where $m=\max\{n_Q,n_{Q'}\}$.
\end{remark}
Note that the reduction used in this remark is an example of one where $\delta\rightarrow 0$ as $\lambda\rightarrow \infty$, but not as $\epsilon\rightarrow 0$ alone.

A consequence of this reduction is that when the quantum inputs are $O(1)$ size (which is the case we are usually interested in), where the quantum inputs are given does not significantly affect the hardness of the task. 
For this reason, we typically assume for simplicity that the quantum inputs are all given on the left.

\section{Some \texorpdfstring{$2\rightarrow 2$}{TEXT} tasks}

First we define $f$-route, introduced in \cite{kent2011quantum, buhrman2013garden}. 
$f$-routing is a special case of the general $f$-controlled $2\rightarrow 2$ task, where we choose the constraints to be such that the map $\mathcal{M}^{x,y}_{Q\rightarrow OO'}$ satisfies $\tr_{O'}\circ \mathcal{M}^{x,y}_{Q\rightarrow OO'}=\mathcal{I}_{Q\rightarrow O}$ if $f(x,y)=0$, and the condition $\tr_{O}\circ \mathcal{M}^{x,y}_{Q\rightarrow OO'}=\mathcal{I}_{Q\rightarrow O'}$ if $f(x,y)=1$.
We define this more carefully next. 

\begin{definition}\label{def:frouting}
    A \textbf{$f$-route} task is defined by a choice of Boolean function $f:\{ 0,1\}^{2n}\rightarrow \{0,1\}$, and a $d_Q$ dimensional Hilbert space $\mathcal{H}_Q$.
    Inputs $x\in \{0,1\}^{n}$ and system $Q$ are given to Alice, and input $y\in \{0,1\}^{n}$ is given to Bob.
    The goal is to bring $Q$ to Alice if $f(x,y)=0$, and to Bob if $f(x,y)=1$. 
    We say the protocol is $\epsilon$-correct if it implements a family of channels $\mathcal{M}^{x,y}_{Q\rightarrow OO'}$ satisfying
    \begin{align}
        \forall (x,y)\in f^{-1}(0), \quad \Vert \tr_{O'}\circ \mathcal{M}^{x,y}_{Q\rightarrow OO'}-\mathcal{I}_{Q\rightarrow O}\Vert_\diamond\leq \epsilon, \nonumber \\
        \forall (x,y)\in f^{-1}(1), \quad \Vert \tr_{O}\circ \mathcal{M}^{x,y}_{Q\rightarrow OO'}-\mathcal{I}_{Q\rightarrow O'}\Vert_\diamond \leq \epsilon.
    \end{align}
    In words, Alice can recover the input to $Q$ if $f(x,y)=0$ and Bob can recover $Q$ if $f(x,y)=1$. 
\end{definition}

We introduce another important special case of an $f$-controlled $2\rightarrow 2$ task next. 
First though, we introduce the following notion. 
\begin{definition}
    Let $M=\{\Lambda^i\}_{i}$ be a POVM on $QQ'$. We define the \textbf{measure-and-copy channel} as the channel $\mathcal{C}_{QQ'\rightarrow OO'}$
    \begin{equation}
        \mathcal{C}_M(\rho):=\sum_{i}\tr (\rho \Lambda^i) \ketbra{i}{i}_{O}\otimes \ketbra{i}{i}_{O'}
    \end{equation}
    that measures the input state $\rho$ with the POVM $\{\Lambda^i\}_{i}$ and then copies the classical outcome to two output systems $OO'$.
\end{definition}
We will be interested in implementing measure-and-copy channels as $f$-controlled $2\rightarrow 2$ tasks, where the value of $f$ fixes which POVM should be used. 

\begin{definition}\label{def:f-measure}
    The $f$\textbf{-measure}$(M_0,M_1)$ task is defined by a choice  of Boolean function $f:\{ 0,1\}^{2n}\rightarrow \{0,1\}$, and two POVMs, $M_0=\{\Lambda^i_1\}_i$, $M_1=\{\Lambda^i_2\}_i$.
    We say that $f$-measure$(M_0, M_1)$ is completed $\epsilon$-correctly if we execute a family of channels $\mathcal{N}^{x,y}$ such that 
    \begin{align}
        \forall (x,y), \,\,\,\Vert \mathcal{C}_{M_{f(x,y)}}(\rho)-\mathcal{N}^{x,y}(\rho)\Vert_\diamond \leq \epsilon
    \end{align}
    where $\mathcal{C}_{M_{f(x,y)}}$ is the measure and copy channel for measurement $M_{f(x,y)}$. 
\end{definition}
A well studied example is $f$-measure$(I, H)$, where $Q$ is a single qubit, $Q'$ is empty, and where $M_0$, $M_1$ are measurements in the computational and Hadamard bases respectively. 

Another type of $f$-controlled $2\rightarrow 2$ task considers the controlled application of a pair of unitaries. 
\begin{definition}\label{def:f-unitary}
    The $f$\textbf{-unitary}$(U^0,U^1)$ task is defined by a choice of Boolean function $f:\{ 0,1\}^{2n}\rightarrow \{0,1\}$, and two unitaries $U^0_{QQ'\rightarrow OO'}, U^1_{QQ'\rightarrow OO'}$.
    Let $\mathcal{U}^b(\cdot)\equiv U^b(\cdot) (U^b)^\dagger$.
    Then we say that the $f$\textbf{-unitary}$(U^0,U^1)$ task is completed $\epsilon$-correctly by a family of channels $\{\mathcal{N}^{x,y}\}_{x,y}$ if
    \begin{align}
        \forall (x,y),\,\,\, \Vert \mathcal{U}^{f(x,y)} - \mathcal{N}^{x,y} \Vert_\diamond \leq \epsilon.
    \end{align}
\end{definition}

Motivated by \cref{remark:redistributequantuminputs}, we assume $QQ'$ are both held by Alice. 
This allows her to apply arbitrary isometries before implementing any NLQC, which means that various $f$-unitary and $f$-measure examples are equivalent. 
For instance, considering $f$-unitary$(U^0,U^1)$, Alice can first choose to apply any unitary $V$ then run the $f$-unitary$(U^0,U^1)$ protocol. 
These combined actions then implement $f$-unitary$(U^0V,U^1V)$. 
This reasoning leads to the following remarks. 

\begin{remark}
    $f$-unitary$(U^0,U^1)$ and $f$-unitary$(U^0V,U^1V)$ are equivalent under $O(1)$ oracle reductions for any choice of unitary $V$. 
\end{remark}

\begin{remark}\label{remark:fmeasurepreprocessing}
    $f$-measure$(M^0,M^1)$ and $f$-measure$(M_V^0,M_V^1)$  are equivalent under $O(1)$ oracle reductions, where $M_V^b=\{V^\dagger \Lambda^i_b V\}_i$, for any choice of unitary $V$. 
\end{remark}

A special case of an $f$-controlled measurement that will be of special interest is the controlled Bell measurement. 
We define $\mathcal{B}^0$ to be the measure and copy channel which measures two input qubits in the computational basis, and $\mathcal{B}^1$ to be the measure and copy channel which measures in the Bell basis. 

\begin{definition}\label{def:fBell}
    A \textbf{$f$-measure(Bell)} task is defined by a choice of Boolean function $f:\{ 0,1\}^{2n}\rightarrow \{0,1\}$, and a two qubit Hilbert space $\mathcal{H}_Q$.
    Inputs $x\in \{0,1\}^{n}$ and system $Q$ are given to Alice, and input $y\in \{0,1\}^{n}$ is given to Bob.
    The $f$-Bell task is completed $\epsilon$-correctly on input $(x,y)$ if channel $\mathcal{N}^{x,y}$ executed on input $x,y$ satisfies
    \begin{align}
        \Vert \mathcal{B}^{f(x,y)} - \mathcal{N}^{x,y} \Vert_\diamond \leq \epsilon \enspace.
    \end{align}
    We say the task is implemented $\epsilon$-correctly if the above holds for all inputs $(x,y)$. 
\end{definition}

\section{Single qubit \texorpdfstring{$f$}{TEXT}-measure protocols}\label{sec:singlequbitfmeasure}

In this section we study $f$-measure protocols, restricted to projective measurements of single qubits. 
Our main result is to show that all of these are equivalent under $O(1)$ oracle reductions.\footnote{Actually, there is a trivial exception: $f$-measure$(M^0,M^1)$ with $M^0=M^1$ is a trivial task, and does not imply any other $f$-measure task.} 

\subsection{Angle characterization} \label{sec:anglecharacterization}

To show the equivalence of single qubit projective rank 1 measurement $f$-measure protocols, we begin by noting that we can always describe projective rank 1 measurements by starting with the computational basis measurement and then conjugating by a unitary, 
\begin{align}
    \Pi^i_0&=U^i\ketbra{0}{0}(U^i)^\dagger, \nonumber \\
    \Pi^i_1&=U^i\ketbra{1}{1}(U^i)^\dagger.
\end{align}
Because of this we also use the notation $f$-measure$(U^0, U^1)$ when dealing with projective measurements. 
Recall from \cref{remark:fmeasurepreprocessing} that by applying a unitary $W$ to the input, the POVM elements of measurement $M^0, M^1$, are transformed according to
\begin{align}
    \Lambda_i^j\rightarrow W^\dagger \Lambda_i^j W.
\end{align}
Thus when describing the measurement in terms of a unitary, we have
\begin{align}
    \Pi_i^j = U^j\ketbra{i}{i}(U^j)^\dagger \rightarrow W^\dagger U^j\ketbra{i}{i}(U^j)^\dagger W
\end{align}
so that pre-processing by $W$ transforms $f$-measure$(U^0, U^1)$ to \\ $f$-measure$(W^\dagger U^0, W^\dagger U^1)$, i.e.\ the $W$ appears daggered and on the left. 

Our first step in showing the equivalence of all such tasks is to show that for each $f$-measure$(U^0, U^1)$ task, there is a choice of angle $\theta$ such that it is equivalent under $O(1)$ oracle reductions to $f$-measure$(I, R_X(\theta))$, where we recall that $R_X(\theta)$ is a rotation around the $x$ axis of the Bloch sphere, 
\begin{align}
    R_X(\theta)=\exp\left(-i X \theta/2 \right).
\end{align}
The rotations $R_Z(\phi)$ and $R_Y(\varphi)$ are defined similarly. 

\begin{lemma}
    \label{lemma:anglecharacterization}
    For any $\theta\in [0,4\pi]$ the following are all equivalent under oracle reductions using single oracle calls:
    \begin{enumerate}
        \item $f$-measure$(I,R_X(\theta))$
        \item $f$-measure$(H,R_Z(\theta)H)$
        \item $f$-measure$(I,R_Y(\theta))$
    \end{enumerate}
    Moreover, for any two single qubit unitaries $U_0,U_1$ there exists an angle $\theta'\in [0,\pi/2]$ such that 
    \begin{enumerate}
        \item $f$-measure$(U_0,U_1)$
        \item $f$-measure$(I,R_X(\theta'))$
    \end{enumerate}
    are equivalent under oracle reductions using a single oracle call.
\end{lemma}

\begin{proof}
We first notice that having Alice apply $U_0$ to the inputs and then run the $f$-measure$(U_0,U_1)$ protocol  implements $f$-measure$(I,U_0^\dagger U_1)$.
Thus $f$-measure$(U_0,U_1)$ implies $f$-measure$(I,U_0^\dagger U_1)$.
 
We let $V=U_0^\dagger U_1$. 
Recall that any single qubit unitary can be decomposed into a $ZXZ$ sequence of rotations \cite{barenco1995elementary}, so that there exist angles $\phi, \theta, \varphi, \xi$ such that
\begin{align}
    V=e^{i\xi}R_Z(\phi)R_X(\theta)R_Z(\varphi).
\end{align}
Here $e^{i\xi}$ is a global phase which we can ignore. 
An implementation of $f$-measure$(I,V)$ then performs measurements in either the computational basis or computational basis acted on by $V$, 
\begin{align}
    M^0 &= \{\ketbra{0}{0}, \ketbra{1}{1}\}, \nonumber \\
    M^1 &= \{V\ketbra{0}{0}V^\dagger , V\ketbra{1}{1}V^\dagger\}.
\end{align}
Notice that part of $V$ acts trivially on the computational basis states, 
\begin{align}
    V\ketbra{i}{i}V^\dagger &= R_Z(\phi)R_X(\theta)R_Z(\varphi)\ketbra{i}{i}R_Z(-\varphi)R_X(-\theta)R_Z(-\phi) \nonumber \\
    &= R_Z(\phi)R_X(\theta)\ketbra{i}{i}R_X(-\theta)R_Z(-\phi).
\end{align}
Thus $f$-measure$(I,V)$ is equivalent to $f$-measure$(I,R_Z(\phi)R_X(\theta))$. 
But now we have Alice apply $R_Z(\phi)$ as pre-processing, which combined with a use of $f$-measure$(I,R_Z(\phi)R_X(\theta))$ gives $f$-measure$(R_Z(-\phi),R_X(\theta))$. 
But conjugation by $R_Z(-\phi)$ acts trivially on the projectors onto the computational basis states, so this is just $f$-measure$(I,R_X(\theta))$. 
Thus $f$-measure$(I,V)$ implies $f$-measure$(I,R_X(\theta))$ where $\theta$ is the angle appearing in the $ZXZ$ decomposition of $V$. 
Combined with our earlier claim that $f$-measure$(U_0,U_1)$ implies $f$-measure$(I,U_0^\dagger U_1)$, this proves that $f$-measure$(U_0,U_1)$ implies \\ $f$-measure$(I,R_X(\theta))$ for an appropriate choice of $\theta$. 

In fact, since we applied reversible (unitary) pre-processing, each step in this argument can be reversed, so that it is also true that $f$-measure$(I,R_X(\theta))$ implies $f$-measure$(U_0,U_1)$, where we choose $\theta$ to be the $X$ rotation angle in the $ZXZ$ decomposition of $U_0^\dagger U_1$. 

To show equivalence with $f$-measure$(H,R_Z(\theta)H)$ we use pre-processing by $H$ and then an application of $f$-measure$(I,R_X(\theta))$, which implements $f$-measure$(H,HR_X(\theta))$, then note that
\begin{align}
    HR_X(\theta) = H (HR_Z(\theta)H)= R_Z(\theta)H.
\end{align}
Applying $H$ again as pre-processing and acting with $f$-measure$(H,HR_X(\theta))$ also returns $f$-measure$(I,R_X(\theta))$, so these are equivalent.

Next we show the equivalence of $f$-measure$(I,R_X(\theta)$ and $f$-measure$(I,R_Y(\theta))$.
Beginning with an $f$-measure$(I,R_X(\theta))$ protocol, consider adding a pre-processing step of applying $R_Z(\phi)$ to the input qubit. 
This leads to an \\
$f$-measure$(I,R_Z^\dagger(\phi)R_X(\theta)R_Z(\phi))$ protocol.
Choosing $\phi=\pi/2$, we use that
\begin{equation}
    R_Y(\theta)=R_Z(-\pi/2)R_X(\theta)R_Z(\pi/2)
\end{equation}
so that we have an $f$-measure$(I,R_Y(\theta))$ protocol. 
An $f$-measure$(I,R_Y(\theta))$ protocol can similarly be transformed into an $f$-measure$(I,R_X(\theta))$ by applying $R_Z(-\pi/2)$ as preprocessing. 

Finally, we show that we can always choose the angle $\theta$ to be in the interval $[0,\pi/2]$. 
In general, the rotation matrix $R_X(\theta)$ is $4\pi$-periodic, however, it has the property that 
\begin{equation}
    \forall \theta\in \mathbb{R}, \quad R_X(\theta+\pi)=iR_X(\theta)X.
\end{equation}
Conjugation with the $X$ operator simply acts on the set of projection operators of the computational basis measurement as a bijection between them. 
This permutation of the outcome labels can always be corrected by classical post-processing of the measurement outcomes, since the function value $f(x,y)$ and therefore whether the labels were permuted is known in the second round by both players. 
Therefore the $f$-measure$(I,R_X(\theta))$ and $f$-measure$(I,R_X(\theta+\pi k))$ protocols are equivalent by local post-processing for any $\theta\in \mathbb{R}$ and $k\in\mathbb{Z}$. 
Because of this we can choose $\theta\in[-\pi/2,\pi/2]$. 
Applying pre-processing with $Z$ to an $f$-measure$(I,R_X(\theta))$ protocol yields $f$-measure$(Z,ZR_X(\theta))$ which is equivalent to $f$-measure$(I,ZR_X(\theta)Z=R_X(-\theta))$ so that we can choose $\theta\in[0,\pi/2]$.
\end{proof}

\subsection{Addition of measurement angles}\label{sec:angleaddition}

So far, we have shown that each single qubit $f$-measure$(U_0,U_1)$ task is equivalent, under reductions using a single oracle call, to $f$-measure$(I,R_X(\theta))$ for some choice of angle $\theta$. 
In this section we begin to explore how this angle can be changed. 
We will show that one oracle use of $f$-measure$(I,R_X(\theta_1))$ and one oracle use of $f$-measure$(I,R_X(\theta_2))$ can be combined to implement $f$-measure$(I,R_X(\theta_1+\theta_2))$.

\begin{lemma}\label{lemma:addingangles}
    The $f$-measure$(I,R_X(\theta_1+\theta_2))$ task can be oracle reduced to one implementation of $f$-measure$(I,R_X(\theta_1))$ and one implementation of $f$-measure$(I,R_X(\theta_2))$. 
\end{lemma}

\begin{proof}
By \cref{lemma:anglecharacterization}, it suffices to show that $f$-measure$(H,R_Z(\theta_1+\theta_2)H)$ can be implemented using one $f$-measure$(H,R_Z(\theta_1)H)$ oracle and one \\ $f$-measure$(H,R_Z(\theta_2)H)$ oracle.

As pre-processing, Alice prepares two ancilla qubits $\ket{0}_A\ket{0}_B$, and applies two CNOT gates from her input $Q$ into $A$ and $B$. 
Then, she inputs $A$ and $B$ into separate $f$-measure$(H, R_Z(\theta_1)H)$ and $f$-measure$(H, R_Z(\theta_2)H)$ oracles, and measures input $Q$ in the Hadamard basis. 
Conditioned on getting measurement outcomes $c,a,b$, this has the effect of implementing:

\begin{center}
\begin{tikzpicture}[on grid,node distance=10mm]
          \node at (0,0)      (i1) {$\ket{\psi}$};
          \node[below=of i1]  (i2) {$\ket{0}$};
          \node[below=of i2]  (i3) {$\ket{0}$};
          \path (i1)  ++(7cm,0)     node[mymeasureselect] (m1) {$c$}
                      ++(0,-10mm)   node[mymeasureselect] (m2) {$a$}
                      ++(0,-10mm)   node[mymeasureselect] (m3) {$b$};
          \path   
                (i1)    ++(1cm,0)  node[mycontrol]       (c1) {}
                        ++(0,-1cm) node[mytarget]        (t2) {$\bigoplus$}
                (c1) edge (t2)
                (i1)    ++(2cm,0)  node[mycontrol]       (c1b) {}
                        ++(0,-2cm) node[mytarget]        (t3) {$\bigoplus$}
                (c1b) edge (t3)
                (i2)    ++(4cm,0) node[mySQG] (r2) {$R_Z^{f(x,y)}(-\theta_1)$}
                (i3)    ++(4cm,0) node[mySQG] (r3) {$R_Z^{f(x,y)}(-\theta_2)$}
                (i1)    ++(6cm,0) node[mySQG] (h1) {$H$}
                (i2)    ++(6cm,0) node[mySQG] (h2) {$H$}
                (i3)    ++(6cm,0) node[mySQG] (h3) {$H$}
                ;
      \begin{scope}[on background layer]
          \path       (i1) edge (m1)
                      (i2) edge (m2)
                      (i3) edge (m3);
      \end{scope}
\end{tikzpicture}
\end{center}

\noindent The bit $c$ is known to Alice in the first round, while $a,b$ are the measurement outcomes from the two $f$-measure oracles so are obtained by Alice and Bob in the second round. 

We can now manipulate this circuit using the identities introduced in \cref{sec:preliminaries} to understand its effect. 
We first use identity 1 to move the $(R_Z(\theta_i))^{f(x,y)}$ rotations to the first qubit, 

\begin{center}
\begin{tikzpicture}[on grid,node distance=10mm]
          \node at (0,0)      (i1) {$\ket{\psi}$};
          \node[below=of i1]  (i2) {$\ket{0}$};
          \node[below=of i2]  (i3) {$\ket{0}$};
          \path (i1)  ++(10cm,0)     node[mymeasureselect] (m1) {$c$}
                      ++(0,-10mm)   node[mymeasureselect] (m2) {$a$}
                      ++(0,-10mm)   node[mymeasureselect] (m3) {$b$};
          \path   
                (i1)    ++(1cm,0)  node[mycontrol]       (c1) {}
                        ++(0,-1cm) node[mytarget]        (t2) {$\bigoplus$}
                (c1) edge (t2)
                (i1)    ++(2cm,0)  node[mycontrol]       (c1b) {}
                        ++(0,-2cm) node[mytarget]        (t3) {$\bigoplus$}
                (c1b) edge (t3)
                (i1)    ++(4cm,0) node[mySQG] (r1) {$R_Z^{f(x,y)}(-\theta_1)$}
                (i1)    ++(7cm,0) node[mySQG] (r1b) {$R_Z^{f(x,y)}(-\theta_2)$}
                (i1)    ++(9cm,0) node[mySQG] (h1) {$H$}
                (i2)    ++(9cm,0) node[mySQG] (h2) {$H$}
                (i3)    ++(9cm,0) node[mySQG] (h3) {$H$}
                ;
      \begin{scope}[on background layer]
          \path       (i1) edge (m1)
                      (i2) edge (m2)
                      (i3) edge (m3);
      \end{scope}
\end{tikzpicture}
\end{center}

\noindent We have also freely re-ordered the controls and the $Z$ rotations since they all commute. 
Next, we use circuit identity 3 to understand the state of the first qubit given the measurement outcomes of the second and third qubits,

\begin{center}
\begin{tikzpicture}[on grid,node distance=10mm]
          \node at (0,0)      (i1) {$\ket{\psi}$};
          \path (i1)  ++(7cm,0)     node[mymeasureselect] (m1) {$c$};
          \path 
                (i1)    ++(1cm,0) node[mySQG] (g1) {$Z^{a\oplus b}$}
                (i1)    ++(3.5cm,0) node[mySQG] (r1) {$R_Z^{f(x,y)}(-(\theta_1+\theta_2))$}
                (i1)    ++(6cm,0) node[mySQG] (h1) {$H$}

                ;
      \begin{scope}[on background layer]
          \path       (i1) edge (m1);
      \end{scope}
\end{tikzpicture}
\end{center}

\noindent Then we commute the $Z$ gate and the $Z$-rotation, and then conjugate the $Z$ gate to the other side of the Hadamard, 

\begin{center}
\begin{tikzpicture}[on grid,node distance=10mm]
          \node at (0,0)      (i1) {$\ket{\psi}$};
          \path (i1)  ++(7cm,0)     node[mymeasureselect] (m1) {$c$};
          \path 
                (i1)    ++(2.25cm,0) node[mySQG] (r1) {$R_Z^{f(x,y)}(-(\theta_1+\theta_2))$}
                (i1)    ++(4.5cm,0) node[mySQG] (h1) {$H$}
                (i1)    ++(5.75cm,0) node[mySQG] (g1) {$X^{a\oplus b}$}
                ;
      \begin{scope}[on background layer]
          \path       (i1) edge (m1);
      \end{scope}
\end{tikzpicture}
\end{center}

\noindent Finally we can absorb the $X$ gate into the measurement outcome, 

\begin{center}
\begin{tikzpicture}[on grid,node distance=10mm]
          \node at (0,0)      (i1) {$\ket{\psi}$};
          \path (i1)  ++(6cm,0)     node[mymeasureselect] (m1) {$a\oplus b\oplus c$};
          \path 
                (i1)    ++(2.25cm,0) node[mySQG] (r1) {$R_Z^{f(x,y)}(-(\theta_1+\theta_2))$}
                (i1)    ++(4.5cm,0) node[mySQG] (h1) {$H$}
                ;
      \begin{scope}[on background layer]
          \path       (i1) edge (m1);
      \end{scope}
\end{tikzpicture}
\end{center}

We see that if Alice and Bob return $c'=c\oplus a\oplus b$, they are giving the same outcome as if they had executed $f$-measure$(H, R_Z(\theta_1+\theta_2)H)$. 
Thus this (perfectly) implements $f$-measure$(H, R_Z(\theta_1+\theta_2)H)$ using one $f$-measure$(H, R_Z(\theta_1)H)$ oracle and one $f$-measure$(H, R_Z(\theta_2)H)$ oracle. 
\end{proof}

\subsection{Splitting the measurement angle}\label{sec:anglesplitting}

Another way to build single qubit $f$-measure protocols is by reducing the rotation angle. 
This can be done with the help of additional oracles to $f$-measure$(I,H)$ protocols.

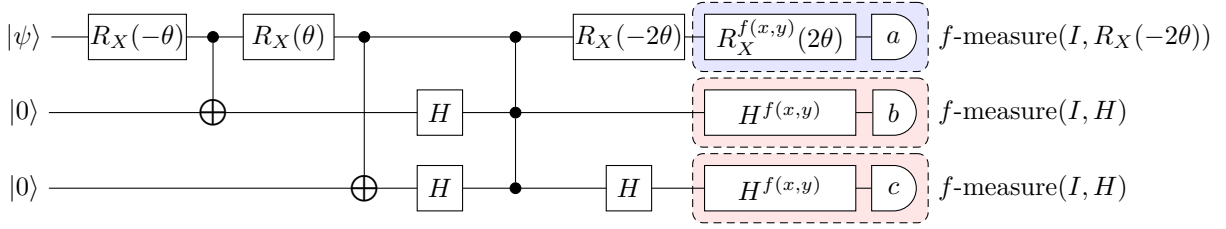
\begin{figure}
\centering
\makebox[\textwidth][c]{
\begin{tikzpicture}[on grid,node distance=10mm]
          \node at (0,0)      (i1) {$\ket{\psi}$};
          \node[below=of i1]  (i2) {$\ket{0}$};
          \node[below=of i2]  (i3) {$\ket{0}$};
          \path (i1)  ++(11.5cm,0)     node[mymeasureselect] (m1) {$a$}
                      ++(0,-10mm)   node[mymeasureselect] (m2) {$b$}
                      ++(0,-10mm)   node[mymeasureselect] (m3) {$c$};
          \path 
                (i1)    ++(1.5cm,0) node[mySQG,minimum width=1.2cm] (r1) {$R_X(-\theta)$}
                (i1)    ++(2.5cm,0)  node[mycontrol]       (c1) {}
                        ++(0,-1cm) node[mytarget]        (t2) {$\bigoplus$}
                (c1) edge (t2)
                (i1)    ++(3.5cm,0) node[mySQG,minimum width=1.2cm] (r1b) {$R_X(\theta)$}
                (i1)    ++(4.5cm,0)  node[mycontrol]       (c1b) {}
                        ++(0,-2cm) node[mytarget]        (t3) {$\bigoplus$}
                (c1b) edge (t3)
                (i2)    ++(5.5cm,0) node[mySQG] (h2) {$H$}
                (i3)    ++(5.5cm,0) node[mySQG] (h3) {$H$}
                (i1)    ++(6.5cm,0)  node[mycontrol]       (c1c) {}
                        ++(0,-1cm) node[mycontrol]       (c2b) {}
                        ++(0,-1cm) node[mycontrol]       (c3) {}
                (c1c) edge (c2b)
                (c2b) edge (c3)
                (i1)    ++(8cm,0) node[mySQG] (r1c) {$R_X(-2\theta)$}
                (i3)    ++(8cm,0) node[mySQG] (h3b) {$H$}
                (i1)    ++(10cm,0) node[mySQG,minimum width=2cm] (r1d) {$R^{f(x,y)}_X(2\theta)$}
                (i2)    ++(10cm,0) node[mySQG,minimum width=2cm] (h2b) {$H^{f(x,y)}$}
                (i3)    ++(10cm,0) node[mySQG,minimum width=2cm] (h3c) {$H^{f(x,y)}$}
                ;
      
    \begin{scope}[on background layer]
        \node[draw=black,rounded corners,densely dashed,fill=blue,fill opacity=0.1,fit=(r1d)(m1),inner sep=1.5mm](fit1) {};
        \node[right=0mm of fit1.east,anchor=west] (fitl1) {$f$-measure$\left(I,R_X(-2\theta)\right)$};
        \node[draw=black,rounded corners,densely dashed,fill=red,fill opacity=0.1,fit=(h2b)(m2),inner sep=1.5mm](fit2) {};
        \node[right=0mm of fit2.east,anchor=west] (fitl2) {$f$-measure$(I,H)$};
        \node[draw=black,rounded corners,densely dashed,fill=red,fill opacity=0.1,fit=(h3c)(m3),inner sep=1.5mm](fit3) {};
        \node[right=0mm of fit3.east,anchor=west] (fitl3) {$f$-measure$(I,H)$};
          \path       (i1) edge (m1)
                      (i2) edge (m2)
                      (i3) edge (m3);
    \end{scope}
\end{tikzpicture}
}
\caption{Circuit that implements $f$-measure$(I,R_X(\theta))$ with local pre-processing and oracle use of $f$-measure$(I,R_X(2\theta))$ and two copies of $f$-measure$(I,H)$. }
\label{fig:splittingangle}
\end{figure}

\begin{lemma}
    \label{lemma:anglesplitting}
    For any $\theta$, the $f$-measure$(I, R_X(\theta))$ protocol can be implemented exactly with oracle access to a single $f$-measure$(I, R_X(2\theta))$ oracle, and two $f$-measure$(I,H)$ oracles. 
\end{lemma}

\begin{proof} We first describe the protocol then prove its correctness.

\vspace{0.2cm}
\noindent \textbf{Protocol:} Alice and Bob's protocol is illustrated in \cref{fig:splittingangle}.
The coloured boxes represent the effect of the three oracle calls used in the protocol; the remaining circuit constitutes pre-processing applied by Alice. 
Call the input qubit $A$ and the two ancilla $B$ and $C$. 
The coloured boxes represent three instances of $f$-measure oracles, one for $f$-measure$(I, R_X(-2\theta))$\footnote{Recall that since $ZR_X(-\phi)Z=R_X(\phi)$, $f$-measure$(I, R_X(-2\theta))$ and $f$-measure$(I, R_X(2\theta))$ are equivalent under an $O(1)$ oracle reduction, so we can freely use $f$-measure$(I, R_X(-2\theta))$ here.}, and two for $f$-measure$(I, H)$.
Let the measurement outcomes from the three oracles be $a$, $b$, and $c$, respectively. 
We claim that there is a deterministic function of $a,b,c$ and $f(x,y)$ that, when output, simulates the outcome of measuring the input $A$ in the computational basis when $f(x,y)=0$, and in the computational basis rotated by $R_X(\theta)$ when $f(x,y)=1$. 
Concretely, the function is
\begin{align}
    m = (1-f)[(1-b) \cdot c + b \cdot a] + f\cdot [b\oplus a\cdot c].
\end{align}
Outputting this bit on both sides and throwing away all remaining systems implements the channel
\begin{align}
    \mathcal{N}_{A\rightarrow OO'}(\rho)=\sum_m \tr(\Pi_m^{f(x,y)}\rho) \ketbra{m}{m}_O\otimes \ketbra{m}{m}_{O'}
\end{align}
with 
\begin{align}
    \Pi^0_0 &= \ketbra{0}{0},\nonumber \\
    \Pi^0_1 &= \ketbra{1}{1}, \nonumber \\
    \Pi^1_0 &= R_X(\theta)\ketbra{0}{0}R_X^\dagger(\theta) \nonumber \\
    \Pi^1_1 &= R_X(\theta)\ketbra{1}{1}R_X^\dagger(\theta),
\end{align}
so that this implements the correct measure and copy channel. 

\vspace{0.2cm}
\noindent \textbf{Correctness:} We analyze the action of the protocol separately in $f(x,y)=0$ and $f(x,y)=1$ cases. 
Begin with the $f=0$ case.
We view the second wire as controlling a $CZ$ gate on the first and third wires. 
Then apply circuit identity 2 to obtain

\begin{center}
\begin{tikzpicture}[on grid,node distance=10mm]
          \node at (0,0)      (i1) {$\ket{\psi}$};
          \node[below=of i1]  (i2) {$\ket{0}$};
          \node[below=of i2]  (i3) {$\ket{0}$};
          \path (i1)  ++(10.5cm,0)     node[mymeasureselect] (m1) {$a$}
                      ++(0,-20mm)      node[mymeasureselect] (m3) {$c$}
                (i2)  ++(4.75cm,0)      node[mymeasureselect] (m2) {$b$};
          \path 
                (i1)    ++(1.5cm,0) node[mySQG,minimum width=1.2cm] (r1) {$R_X(-\theta)$}
                (i1)    ++(2.5cm,0)  node[mycontrol]       (c1) {}
                        ++(0,-1cm) node[mytarget]        (t2) {$\bigoplus$}
                (c1) edge (t2)
                (i1)    ++(3.5cm,0) node[mySQG,minimum width=1.2cm] (r1b) {$R_X(\theta)$}
                (i2)    ++(3.5cm,0) node[mySQG] (h2) {$H$}
                (i1)    ++(5.5cm,0)  node[mycontrol]       (c1b) {}
                        ++(0,-2cm) node[mytarget]        (t3) {$\bigoplus$}
                (c1b) edge (t3)
                (i3)    ++(6.5cm,0) node[mySQG] (h3) {$H$}
                (i1)    ++(7.5cm,0)  node[mycontrol]       (c1c) {}
                        ++(0,-2cm) node[mycontrol]       (c3) {}
                (c1c) edge (c3)
                (i1)    ++(9cm,0) node[mySQG] (r1c) {$R_X(-2\theta)$}
                (i3)    ++(9cm,0) node[mySQG] (h3b) {$H$}
                ;
        \node[draw=black,rounded corners,densely dashed,fit=(c1c)(c3),inner sep=1.5mm](fit1) {};
        \node[above=1mm of fit1.north,anchor=south] (fitl1) {$b$};
      \begin{scope}[on background layer]
          \path       (i1) edge (m1)
                      (i2) edge (m2)
                      (i3) edge (m3);
      \end{scope}
\end{tikzpicture}
\end{center}

\noindent The gate inside the dashed box is applied if $b=1$. 
Now apply circuit identity 3 to the first and second wires, obtaining

\begin{center}
\begin{tikzpicture}[on grid,node distance=10mm]
          \node at (0,0)      (i1) {$\ket{\psi}$};
          \node[below=of i1]  (i2) {$\ket{0}$};
          \node[below=of i2]  (i3) {$\ket{0}$};
          \path (i1)  ++(10.5cm,0)     node[mymeasureselect] (m1) {$a$}
                      ++(0,-20mm)      node[mymeasureselect] (m3) {$c$}
                (i2)  ++(1.5cm,0)      node[mymeasureselect] (m2) {$b$};
          \path 
                (i1)    ++(1.5cm,0) node[mySQG,minimum width=1.2cm] (r1) {$R_X(-\theta)$}
                (i1)    ++(3cm,0)  node[mySQG]       (z1) {$Z^b$}
                (i1)    ++(4.5cm,0) node[mySQG,minimum width=1.2cm] (r1b) {$R_X(\theta)$}
                (i1)    ++(5.5cm,0)  node[mycontrol]       (c1b) {}
                        ++(0,-2cm) node[mytarget]        (t3) {$\bigoplus$}
                (c1b) edge (t3)
                (i3)    ++(6.5cm,0) node[mySQG] (h3) {$H$}
                (i1)    ++(7.5cm,0)  node[mycontrol]       (c1c) {}
                        ++(0,-2cm) node[mycontrol]       (c3) {}
                (c1c) edge (c3)
                (i1)    ++(9cm,0) node[mySQG] (r1c) {$R_X(-2\theta)$}
                (i3)    ++(9cm,0) node[mySQG] (h3b) {$H$}
                ;
        \node[draw=black,rounded corners,densely dashed,fit=(c1c)(c3),inner sep=1.5mm](fit1) {};
        \node[above=1mm of fit1.north,anchor=south] (fitl1) {$b$};
      \begin{scope}[on background layer]
          \path       (i1) edge (m1)
                      (i2) edge (m2)
                      (i3) edge (m3);
      \end{scope}
\end{tikzpicture}
\end{center}

\noindent At this point, the second wire is decoupled from the others, so we can remove it and simplify this to

\begin{center}
\begin{tikzpicture}[on grid,node distance=10mm]
          \node at (0,0)      (i1) {$\ket{\psi}$};
          \node[below=of i1]  (i3) {$\ket{0}$};
          \path (i1)  ++(10.5cm,0)     node[mymeasureselect] (m1) {$a$}
                      ++(0,-10mm)      node[mymeasureselect] (m3) {$c$};
          \path 
                (i1)    ++(1.5cm,0) node[mySQG,minimum width=1.2cm] (r1) {$R_X(-\theta)$}
                (i1)    ++(3cm,0)  node[mySQG]       (z1) {$Z^b$}
                (i1)    ++(4.5cm,0) node[mySQG,minimum width=1.2cm] (r1b) {$R_X(\theta)$}
                (i1)    ++(5.5cm,0)  node[mycontrol]       (c1b) {}
                        ++(0,-1cm) node[mytarget]        (t3) {$\bigoplus$}
                (c1b) edge (t3)
                (i3)    ++(6.5cm,0) node[mySQG] (h3) {$H$}
                (i1)    ++(7.5cm,0)  node[mycontrol]       (c1c) {}
                        ++(0,-1cm) node[mycontrol]       (c3) {}
                (c1c) edge (c3)
                (i1)    ++(9cm,0) node[mySQG] (r1c) {$R_X(-2\theta)$}
                (i3)    ++(9cm,0) node[mySQG] (h3b) {$H$}
                ;
        \node[draw=black,rounded corners,densely dashed,fit=(c1c)(c3),inner sep=1.5mm](fit1) {};
        \node[above=1mm of fit1.north,anchor=south] (fitl1) {$b$};
      \begin{scope}[on background layer]
          \path       (i1) edge (m1)
                      (i3) edge (m3);
      \end{scope}
\end{tikzpicture}
\end{center}

\noindent Now we further need to consider cases $b=0,1$. 

For $b=0$, notice that the $R_X(-\theta)$ and $R_X(\theta)$ gates cancel, and the $CZ$ gate is not applied, which leaves the $H$ gates to cancel. 
Thus the full circuit amounts to a CNOT with $1$ as the control and $3$ as the target, so that the output $c$ will amount to a measurement of the input $A$ in the computational basis. 
Since also when $f=0,b=0$ we have that $m=c$, we find that $p_m=p_c=\tr(\Pi^0_m\rho)$ as needed. 

Now consider the $b=1$ case. 
Then the $CZ$ gate is applied, and combined with the $H$ gates it gives a second CNOT. 
The two CNOT gates then cancel. 
Then use 
\begin{align}
    R_X(\theta)ZR_X(-\theta)=R_X(2\theta)Z
\end{align}
which can be verified directly. 
The $R_X(2\theta)$ cancels with the $R_X(-2\theta)$ applied just before the top wire is measured. 
Since also in this case $m=a$, we get that $p_m=p_a=\tr(\Pi^0_m\rho)$, as needed. 
This handles all possibilities when $f=0$. 

We now consider the $f=1$ case. 
Taking $f=1$, the circuit in \cref{fig:splittingangle} simplifies to

\begin{center}
\begin{tikzpicture}[on grid,node distance=10mm]
          \node at (0,0)      (i1) {$\ket{\psi}$};
          \node[below=of i1]  (i2) {$\ket{0}$};
          \node[below=of i2]  (i3) {$\ket{0}$};
          \path (i1)  ++(8.5cm,0)     node[mymeasureselect] (m1) {$a$}
                      ++(0,-10mm)   node[mymeasureselect] (m2) {$b$}
                      ++(0,-10mm)   node[mymeasureselect] (m3) {$c$};
          \path 
                (i1)    ++(1.5cm,0) node[mySQG,minimum width=1.2cm] (r1) {$R_X(-\theta)$}
                (i1)    ++(2.5cm,0)  node[mycontrol]       (c1) {}
                        ++(0,-1cm) node[mytarget]        (t2) {$\bigoplus$}
                (c1) edge (t2)
                (i1)    ++(3.5cm,0) node[mySQG,minimum width=1.2cm] (r1b) {$R_X(\theta)$}
                (i1)    ++(4.5cm,0)  node[mycontrol]       (c1b) {}
                        ++(0,-2cm) node[mytarget]        (t3) {$\bigoplus$}
                (c1b) edge (t3)
                (i2)    ++(5.5cm,0) node[mySQG] (h2) {$H$}
                (i3)    ++(5.5cm,0) node[mySQG] (h3) {$H$}
                (i1)    ++(6.5cm,0)  node[mycontrol]       (c1c) {}
                        ++(0,-1cm) node[mycontrol]       (c2b) {}
                        ++(0,-1cm) node[mycontrol]       (c3) {}
                (c1c) edge (c2b)
                (c2b) edge (c3)
                (i2)    ++(7.5cm,0) node[mySQG] (h2b) {$H$}
                ;
      \begin{scope}[on background layer]
          \path       (i1) edge (m1)
                      (i2) edge (m2)
                      (i3) edge (m3);
      \end{scope}
\end{tikzpicture}
\end{center}

\noindent Then, view the last wire as controlling a $CZ$ gate on the first two wires and apply circuit identity 2 to yield

\begin{center}
\begin{tikzpicture}[on grid,node distance=10mm]
          \node at (0,0)      (i1) {$\ket{\psi}$};
          \node[below=of i1]  (i2) {$\ket{0}$};
          \node[below=of i2]  (i3) {$\ket{0}$};
          \path (i1)  ++(8.5cm,0)     node[mymeasureselect] (m1) {$a$}
                      ++(0,-10mm)   node[mymeasureselect] (m2) {$b$}
                      ++(-2cm,-10mm)   node[mymeasureselect] (m3) {$c$};
          \path 
                (i1)    ++(1.5cm,0) node[mySQG,minimum width=1.2cm] (r1) {$R_X(-\theta)$}
                (i1)    ++(2.5cm,0)  node[mycontrol]       (c1) {}
                        ++(0,-1cm) node[mytarget]        (t2) {$\bigoplus$}
                (c1) edge (t2)
                (i1)    ++(3.5cm,0) node[mySQG,minimum width=1.2cm] (r1b) {$R_X(\theta)$}
                (i1)    ++(4.5cm,0)  node[mycontrol]       (c1b) {}
                        ++(0,-2cm) node[mytarget]        (t3) {$\bigoplus$}
                (c1b) edge (t3)
                (i2)    ++(5.5cm,0) node[mySQG] (h2) {$H$}
                (i3)    ++(5.5cm,0) node[mySQG] (h3) {$H$}
                (i1)    ++(6.5cm,0)  node[mycontrol]       (c1c) {}
                        ++(0,-1cm) node[mycontrol]       (c2b) {}
                (c1c) edge (c2b)
                (i2)    ++(7.5cm,0) node[mySQG] (h2b) {$H$}
                ;
        \node[draw=black,rounded corners,densely dashed,fit=(c1c)(c2b),inner sep=1.5mm](fit1) {};
        \node[above=1mm of fit1.north,anchor=south] (fitl1) {$c$};

      \begin{scope}[on background layer]
          \path       (i1) edge (m1)
                      (i2) edge (m2)
                      (i3) edge (m3);
      \end{scope}
\end{tikzpicture}
\end{center}

\noindent Then use circuit identity 3 to obtain

\begin{center}
\begin{tikzpicture}[on grid,node distance=10mm]
          \node at (0,0)      (i1) {$\ket{\psi}$};
          \node[below=of i1]  (i2) {$\ket{0}$};
          \node[below=of i2]  (i3) {$\ket{0}$};
          \path (i1)  ++(9cm,0)       node[mymeasureselect] (m1) {$a$}
                      ++(0,-10mm)     node[mymeasureselect] (m2) {$b$}
                (i1)  ++(2.5cm,-20mm) node[mymeasureselect] (m3) {$c$};
          \path 
                (i1)    ++(1.5cm,0) node[mySQG,minimum width=1.2cm] (r1) {$R_X(-\theta)$}
                (i1)    ++(2.5cm,0)  node[mycontrol]       (c1) {}
                        ++(0,-1cm) node[mytarget]        (t2) {$\bigoplus$}
                (c1) edge (t2)
                (i1)    ++(3.5cm,0) node[mySQG,minimum width=1.2cm] (r1b) {$R_X(\theta)$}
                (i1)    ++(5cm,0)  node[mySQG]       (g1) {$Z^c$}
                (i2)    ++(6cm,0) node[mySQG] (h2) {$H$}
                (i3)    ++(1.5cm,0) node[mySQG] (h3) {$H$}
                (i1)    ++(7cm,0)  node[mycontrol]       (c1c) {}
                        ++(0,-1cm) node[mycontrol]       (c2b) {}
                (c1c) edge (c2b)
                (i2)    ++(8cm,0) node[mySQG] (h2b) {$H$}
                ;
        \node[draw=black,rounded corners,densely dashed,fit=(c1c)(c2b),inner sep=1.5mm](fit1) {};
        \node[above=1mm of fit1.north,anchor=south] (fitl1) {$c$};

      \begin{scope}[on background layer]
          \path       (i1) edge (m1)
                      (i2) edge (m2)
                      (i3) edge (m3);
      \end{scope}
\end{tikzpicture}
\end{center}

\noindent Now remove the third wire, which has decoupled, and re-write the controlled $CZ$ as a CNOT, 

\begin{center}
\begin{tikzpicture}[on grid,node distance=10mm]
          \node at (0,0)      (i1) {$\ket{\psi}$};
          \node[below=of i1]  (i2) {$\ket{0}$};
          \path (i1)  ++(7cm,0)       node[mymeasureselect] (m1) {$a$}
                      ++(0,-10mm)     node[mymeasureselect] (m2) {$b$};
          \path 
                (i1)    ++(1.5cm,0) node[mySQG,minimum width=1.2cm] (r1) {$R_X(-\theta)$}
                (i1)    ++(2.5cm,0)  node[mycontrol]       (c1) {}
                        ++(0,-1cm) node[mytarget]        (t2) {$\bigoplus$}
                (c1) edge (t2)
                (i1)    ++(3.5cm,0) node[mySQG,minimum width=1.2cm] (r1b) {$R_X(\theta)$}
                (i1)    ++(5cm,0)  node[mySQG]       (g1) {$Z^c$}
                (i1)    ++(6cm,0)  node[mycontrol]       (c1c) {}
                        ++(0,-1cm) node[mytarget]       (t2b) {$\bigoplus$}
                (c1c) edge (t2b)
                ;
        \node[draw=black,rounded corners,densely dashed,fit=(c1c)(t2b),inner sep=1.5mm](fit1) {};
        \node[above=1mm of fit1.north,anchor=south] (fitl1) {$c$};

      \begin{scope}[on background layer]
          \path       (i1) edge (m1)
                      (i2) edge (m2);
      \end{scope}
\end{tikzpicture}
\end{center}

\noindent Now use circuit identity 2 again, this time with the first wire as the control, 

\begin{center}
\begin{tikzpicture}[on grid,node distance=10mm]
          \node at (0,0)      (i1) {$\ket{\psi}$};
          \node[below=of i1]  (i2) {$\ket{0}$};
          \path (i1)  ++(6cm,0)       node[mymeasureselect] (m1) {$a$}
                      ++(0,-10mm)     node[mymeasureselect] (m2) {$b$};
          \path 
                (i1)    ++(1.5cm,0) node[mySQG,minimum width=1.2cm] (r1) {$R_X(-\theta)$}
                (i1)    ++(2.5cm,0)  node[mycontrol]       (c1) {}
                        ++(0,-1cm) node[mytarget]        (t2) {$\bigoplus$}
                (c1) edge (t2)
                (i1)    ++(3.5cm,0) node[mySQG,minimum width=1.2cm] (r1b) {$R_X(\theta)$}
                (i1)    ++(5cm,0)  node[mySQG]       (g1) {$Z^c$}
                (i2)    ++(5cm,0)  node[mySQG]       (g2) {$X^{a\cdot c}$}
                ;
      \begin{scope}[on background layer]
          \path       (i1) edge (m1)
                      (i2) edge (m2);
      \end{scope}
\end{tikzpicture}
\end{center}

\noindent We can see that this amounts to measuring the first qubit in the $R_X(\theta)$ rotated basis, and recording the outcome in the second output. 
The output will be flipped if $a\cdot c=1$. 
Thus outputting $m=b\oplus(a\cdot c)$ gives a measurement of the first qubit in the $R_X(\theta)$ rotated basis, so that $p_m=\tr(\Pi^1_m\rho)$, as needed. 
\end{proof}

\subsection{From an arbitrary angle to \texorpdfstring{$f$}{TEXT}-measure\texorpdfstring{$(I,H)$}{TEXT}}\label{sec:QECC}

By \cref{lemma:anglecharacterization} any single qubit $f$-measure$(U,U')$ protocol is equivalent to an $f$-measure$(I,R_Y(\theta))$ protocol for some $\theta\in[0,\pi/2]$. 
Any such protocol for which $\theta\neq 0$ will be called non-trivial.
The following lemma then shows that any non-trivial single qubit $f$-measure$(U,U')$ protocol implies $f$-measure$(I,H)$ under $O(1)$ oracle reductions.

\begin{lemma}\label{lemma:backtofH}
    The $f$-measure$(I, H)$ task can be implemented $\epsilon$-correctly using $m=O(1)$ instances of a single qubit $f$-measure$(U,U')$ oracle, whenever  $f$-measure$(U,U')$ is non-trivial.
\end{lemma}

\begin{proof}
We first use that (by assumption) $f$-measure$(U,U')$ is equivalent under $O(1)$ reductions to $f$-measure$(I,R_Y(\theta))$ with $\theta\neq 0$. 
For simplicity of our notation we define $R_Y(\theta)=V$.
Define
\begin{align}
    p_{0|+}=& |\bra{0}V^\dagger\ket{+}|^2,\qquad \quad  p_{1|+}= |\bra{1}V^\dagger\ket{+}|^2, \nonumber \\
    p_{0|-}=& |\bra{0}V^\dagger\ket{-}|^2,\qquad \quad  p_{1|-}= |\bra{1}V^\dagger\ket{-}|^2.
\end{align}
Note that $p_{0|+}+p_{1|+}=1$, $p_{1|-}+p_{0|-}=1$, $p_{0|+}+p_{0|-}=1$, $p_{1|+}+p_{1|-}=1$.\footnote{Note that only three of these equations are independent, the fourth can be derived from the previous three.} 
These equations leave just one free value, which we can take to be $p_{0|+}$. 
We focus on the case where $p_{0|+}>p_{1|+}$, which also implies $p_{0|-}<p_{1|-}$, but the remaining cases (as long as there is no equality, which is guaranteed by non-triviality) can be handled similarly.

\vspace{0.2cm}
\noindent \textbf{Protocol:} 
Alice copies the quantum input $Q$ in the $\ket{\pm}$ basis $m$ times,
\begin{align}
    \ket{+}&\rightarrow \ket{+}^{\otimes m},\nonumber \\
    \ket{-}&\rightarrow \ket{-}^{\otimes m}.
\end{align}
Label the output systems from this operation as $Q_1 \dots Q_m$. 
Each of the $m$ qubits $Q_i$ are input to one of the $f$-measure$(I,V)$ oracles, producing measurement outcomes $m_i$ which are revealed on both sides in the second round. 
Alice and Bob then determine their outputs as follows: 
\begin{itemize}
    \item If $f(x,y)=0$, they output the parity of the $m_i$, $\bigoplus_{i} m_i$. 
    \item If $f(x,y)=1$, then they output $0$ if they get at least $(p_{0|+}-\delta)m\equiv \lambda m$ outcomes that are $V\ket{0}$, and $1$ otherwise. $\delta>0$ is a parameter that can be chosen arbitrarily, with the restriction that $\lambda>1/2$ should still hold.
\end{itemize}

\vspace{0.2cm}
\noindent \textbf{Correctness:} We consider what occurs in the $f=0$, $f=1$ cases separately.

First suppose that $f=0$. 
Each qubit gets measured in the computational basis, so let us look at the input state in that basis, 
\begin{align}
    \ket{\psi} &= \alpha \ket{0}+\beta\ket{1},\nonumber \\
    &= \frac{1}{\sqrt{2}}\alpha (\ket{+}+\ket{-})+\frac{1}{\sqrt{2}}\beta (\ket{+}-\ket{-}).
\end{align}
Define $wt(x)=\sum_i x_i$ to be the Hamming weight of $x$, and say that $x$ is even parity if $\sum_i x_i=0 \,\, \text{mod}\,\, 2$, and odd otherwise. 
After the pre-processing step where we copy in the Hadamard basis, the input state becomes 
\begin{align}
    \ket{\psi} &\rightarrow \alpha \left(\frac{1}{\sqrt{2^{m-1}}} \sum_{x:x\,\text{even}} \ket{x} \right)+\beta\left(\frac{1}{\sqrt{2^{m-1}}} \sum_{x:x\,\text{odd}}(-1)\ket{x}\right).
\end{align}
Thus Alice and Bob get an even parity outcome with probability $|\alpha|^2$ and odd parity outcome with probability $|\beta|^2$.
Since $|\alpha|^2=\tr(\Pi^0_0 \ketbra{\psi}{\psi})$ and $|\beta|^2=\tr(\Pi^0_1 \ketbra{\psi}{\psi})$, this performs the measure and copy channel in the computational basis. 

Next consider the case when $f=1$. 
Now write the input state $\ket{\psi}$ in the $X$ basis, 
\begin{align}
    \ket{\psi} = a\ket{+}+b\ket{-}.
\end{align}
And after the copying step, this becomes
\begin{align}
    \ket{\Psi} = a\ket{+}^{\otimes m} + b\ket{-}^{\otimes m}.
\end{align}
Now, each of the $Q_i$ are input to an $f$-measure$(I,V)$ oracle, so are measured in the $V$ basis.
Let the $V$ basis states be $\ket{\psi_0}=V\ket{0}$, $\ket{\psi_1}=V\ket{1}$.
Alice and Bob output $0$ if more than $\lambda m$ of the outputs from these oracles are $0$, and $1$ otherwise, so they are measuring the projectors 
\begin{align}
    \Pi^1_0&=\sum_{k\geq m \lambda}\sum_{\sigma} S_{\sigma}\left(\ketbra{\psi_0}{\psi_0}^{\otimes k} \otimes \ketbra{\psi_1}{\psi_1}^{\otimes (m-k)}\right), \nonumber \\
    \Pi^1_1&=I-\Pi^1_0.
\end{align}
Here $\sigma$ denotes a permutation, and $S_\sigma$ permutes the $m$ tensor factors of the Hilbert space.
We sum over permutations that act distinctly, i.e.\ over the $\binom{m}{k}$ permutations that bring the first $k$ factors to each possible subset of size $k$ of the $m$ factors. 

Returning the outcome of this measurement means that Alice and Bob implement the channel
\begin{align}
    \mathcal{N}'(\rho)=\sum_{i} \tr(\Pi_i^1\mathcal{C}(\rho)) \ketbra{i}{i}\otimes \ketbra{i}{i}
\end{align}
where $\mathcal{C}$ is the copying operation in the $X$ basis.
We would like that this converges, as $m\rightarrow \infty$, to the channel
\begin{align}
    \mathcal{N}(\rho)=\sum_{i} \tr(\Pi_i^H\rho) \ketbra{i}{i}\otimes \ketbra{i}{i}
\end{align}
where $\Pi_i^H$ are the projectors for the Hadamard basis. 
We will show below that 
\begin{align}\label{eq:gammabounds}
    |\tr(\Pi_0^1\mathcal{C}(\rho)) - \tr(\Pi_0^H\rho)| \leq \gamma_m, \nonumber \\
    |\tr(\Pi_1^1\mathcal{C}(\rho)) - \tr(\Pi_1^H\rho)| \leq \gamma_m,
\end{align}
where the function $\gamma_m$ goes to zero as $m\rightarrow \infty$. 
From this, it follows that for all density matrices $\rho$, 
\begin{align}
    \Vert \mathcal{N}'(\rho)-\mathcal{N}(\rho)\Vert_1 \leq 2\gamma_m
\end{align}
and then applying \cref{lemma:densitytoDN}, we can see that this implies the diamond norm satisfies 
\begin{align}
    \Vert \mathcal{N}'-\mathcal{N}\Vert_\diamond \leq 4d_A\gamma_m.
\end{align}

Towards showing this, first note that it suffices to prove an upper bound on $|\tr(\Pi_0^1\mathcal{C}(\rho)) - \tr(\Pi_0^H\rho)|$, since
\begin{align}
    &|\tr(\Pi_1^1\mathcal{C}(\rho)) - \tr(\Pi_1^H\rho)|
    \\
    =&|\tr((\Pi_1^1-I)\mathcal{C}(\rho))+\tr{\mathcal{C}(\rho)} - \tr((\Pi_1^H-I)\rho)-\tr{\rho}|
    \\
    =&
    |\tr((\Pi_1^1-I)\mathcal{C}(\rho)) - \tr((\Pi_1^H-I)\rho)|
    \\
    =&
    |\tr(\Pi_0^1\mathcal{C}(\rho)) - \tr(\Pi_0^H\rho)|
\end{align}
Thus we focus on just $|\tr(\Pi_0^1\mathcal{C}(\rho)) - \tr(\Pi_0^H\rho)|$.

It remains then to show that $\tr(\Pi_i^1\mathcal{C}(\rho))$ converges to $\tr(\Pi_i^H\rho)$ and to determine $\gamma_m$. 
We look at $\tr(\Pi_0^1\mathcal{C}(\rho))$ first. 
For a pure state input we have that the copied state is $\ket{\Psi}=a\ket{+}^{\otimes m}+b\ket{-}^{\otimes {m}}$, then the probability of outcome 0 is
\begin{align}\label{eq:Pi01plus}
    \tr(\Pi_0^1\ketbra{\Psi}{\Psi}) &= |a|^2 \tr(\Pi_0^1 \ketbra{+}{+}^{\otimes m}) + |b|^2 \tr(\Pi_0^1 \ketbra{-}{-}^{\otimes m}) \nonumber \\
    &\qquad + ab^* \,\tr(\Pi^1_0 \ketbra{+}{-}^{\otimes m}) + a^*b \,\tr(\Pi^1_0 \ketbra{-}{+}^{\otimes m}).
\end{align}
Focusing on the first term, we have
\begin{align}
    \tr(\Pi_0^1 \ketbra{+}{+}^{\otimes m}) = \sum_{k\geq \lambda m} \binom{m}{k} p_{0|+}^kp_{1|+}^{m-k}.
\end{align}
Since $p_{0|+}>p_{1|+}$, this converges to $1$ as $m\rightarrow \infty$: intuitively this is because the above is a binomial distribution with peak at $mp_{0|+}$. 
Since we sum from somewhere just below this peak $m\lambda=m(p_{0|+}-\delta)$, the sum includes most of the distribution. 
Being more precise, the size of the remainder term can be bounded using Hoeffding's inequality.
We define the random variable $X_i\in\{0,1\}$, with $Pr[X_i=0]=p_{1|+}$, $Pr[X_i=1]=p_{0|+}$. 
Then define $S_m=\sum_{i=1}^m X_i$, and notice
\begin{align}
    Pr[S_m=\ell]=\binom{m}{\ell} p_{0|+}^\ell p_{1|+}^{m-\ell}
\end{align}
so that also
\begin{align}
    Pr[S_m\geq \lambda m]=\sum_{k\geq \lambda m} \binom{m}{k} p_{0|+}^kp_{1|+}^{m-k}
\end{align}
Now use Hoeffding's inequality
\begin{align}
    Pr[S_m\geq \lambda m] &= Pr[S_m \geq (p_{0|+}-\delta)m] \nonumber \\
    &= Pr[-S_m + p_{0|\mathord{+}}m\leq \delta m] \nonumber \\
    &= 1-Pr[-S_m + p_{0|+}m > \delta m] \nonumber \\
    &\geq 1-Pr[|S_m - p_{0|+}m| \geq \delta m] \nonumber \\
    &\geq 1-2\exp\left(-2\delta^2m \right)
\end{align}
so that 
\begin{align}
    \tr(\Pi_0^1 \ketbra{+}{+}^{\otimes m}) \geq 1- 2\exp\left(-2\delta^2m \right).
\end{align}

Now return to expression \eqref{eq:Pi01plus}, and consider the second term. 
This is 
\begin{align}
    \tr(\Pi_0^1 \ketbra{-}{-}^{\otimes m}) = \sum_{k\geq \lambda m} \binom{m}{k} p_{0|-}^kp_{1|-}^{m-k}
\end{align}
Now since $p_{0|-}<p_{1|-}$, this will converge to zero: now we are summing over a binomial distribution starting from $k=\lambda m$ with $\lambda>1/2$, but the peak is below $m/2$. 
Again we can upper bound this term with Hoeffding's inequality via a similar argument as with the first term, obtaining an upper bound of $\exp\left(-2m (p_{0|-}-\lambda)^2 \right)$.
Since $p_{0|-}<1/2, \lambda>1/2$ we know this exponent is non-zero. 

Finally we need to handle the third and fourth terms, which we claim also go to zero. 
Consider the third term. 
This is, without the $ab^*$ factor, 
\begin{align}\label{eq:34termbound}
    \left|\tr(\Pi^1_0 \ketbra{+}{-}^{\otimes m})\right| &= \left|\sum_{k\geq \lambda m} \binom{m}{k} \braket{\psi_0}{+}^k\braket{-}{\psi_0}^k \braket{\psi_1}{+}^{m-k}\braket{-}{\psi_1}^{m-k}\right| \nonumber \\
    &\leq \sum_{k\geq \lambda m} \binom{m}{k} p_{0|+}^{k/2}p_{0|-}^{k/2} p_{1|+}^{(m-k)
    /2}p_{1|-}^{(m-k)/2} \nonumber \\
    &\leq  \sum_{k\geq \lambda m} \binom{m}{k} \left(\frac{1}{2}\right)^m \nonumber \\
    &\leq \exp\left(-2m(\lambda-1/2)^2 \right) \,,
\end{align}
where the first inequality follows by bringing the absolute value inside the sum, the second follows because $\sqrt{x(1-x)}\leq 1/2$ when $0\leq x\leq 1$, and the third follows by applying a standard tail bound to the binomial distribution.
The fourth term satisfies the same upper bound. 

Combining our bounds on each term of equation \eqref{eq:Pi01plus}, we have that 
\begin{align}
    \tr(\Pi_0^1\ketbra{\Psi}{\Psi}) &\geq |a|^2(1-\exp(-2\delta^2 m)) - 2\exp(-2m(\lambda-1/2)^2) \nonumber \\
    &\geq |a|^2-\exp(-2\delta^2 m) - 2\exp(-2m(\lambda-1/2)^2) \,,
\end{align}
where we can drop the $|b|^2$ term because it is positive, and we used that $|a^*b|\leq 1$ and equation \eqref{eq:34termbound} to bound the sum of the third and fourth terms.
We use $|a|^2\leq 1$ in the last line. 
We can also upper bound the same quantity, 
\begin{align}
    \tr(\Pi_0^1\ketbra{\Psi}{\Psi})\leq |a|^2 + \exp(-2m(p_{0|-}-\lambda)^2) + 2\exp(-2m(\lambda-1/2)^2).
\end{align}
Thus choosing 
\begin{align}
    \gamma_m=\exp(-2\delta^2 m) + 2\exp(-2m(\lambda-1/2)^2) + \exp(-2m(p_{0|-}-\lambda)^2)
\end{align}
we obtain
\begin{align}
    ||a|^2-\tr(\Pi^1_0\mathcal{C}(\rho))|\leq \gamma_m \,,
\end{align}
but $|a|^2=\tr(\Pi_0^H\rho)$, so this is exactly the first expression of equation \eqref{eq:gammabounds}. 
If we wish to implement $f$-measure$(I,H)$ $\epsilon$-correctly, we need $\gamma_m\leq \epsilon$, so we need $m=\Omega(1/\log \epsilon)$.
Notice that because $\gamma_m$ does not depend on $n$, we can make the error as small as we like while choosing $m$ independent of $n$. 
Thus this procedure gives an $O(1)$ oracle reduction. 

The above handles the case where $p_{0|+}>p_{1|+}$, if instead $p_{0|+}<p_{1|+}$ we have Alice and Bob behave as follows. 
Their behaviour is unchanged with $f(x,y)=0$. 
If $f(x,y)=1$, they output $0$ if more than $(p_{1|+}-\delta)m$ outcomes are $V\ket{1}$, and $1$ otherwise.
\end{proof}

\subsection{All single qubit \texorpdfstring{$f$}{TEXT}-measure protocols are approximately equivalent}\label{sec:fmeasureequivalence}

Finally, we are ready to collect our results from this section to prove that all single qubit $f$-measure$(U,V)$ tasks are equivalent under $O(1)$ reductions. 

\fmeasureEquivalence*

\begin{proof}
By \cref{lemma:anglecharacterization}, $f$-measure$(U_1,V_1)$ is equivalent under $O(1)$ oracle reductions to $f$-measure$(I,R_X(\theta_1))$ for some angle $\theta_1$. Similarly, $f$-measure$(U_2,V_2)$ is equivalent under $O(1)$ oracle reductions to $f$-measure$(I,R_X(\theta_2))$, for another angle $\theta_2$. 

If $\theta_1=\theta_2$ we are done. 

If $\theta_1\neq \theta_2$ we proceed as follows. We will show that $f$-measure$(I,R_X(\theta_1))$ implies $f$-measure$(I,R_X(\theta_2))$; the reverse direction uses the same argument. The first step is to invoke \cref{lemma:backtofH}, which shows that $f$-measure$(I,R_X(\theta_1))$ implies $f$-measure$(I,H)$ using an $O(1)$ oracle construction. Then, using $f$-measure$(I,H)$ oracles, we can use \cref{lemma:anglesplitting} to construct $f$-measure$(I,R_X(\pi/2^{k+1}))$ oracles for any $k\geq0$, again using an $O(1)$ oracle construction. 
Then, we express angle $\theta_2$ as 
\begin{align}
    \theta_2 &= \pi \sum_{k=0}^t \frac{c_k}{2^{k+1}} + \epsilon, \nonumber \\
    &= \tilde{\theta}_2 + \epsilon.
\end{align}
with $c_k\in \{0,1\}$, $\epsilon<\pi/2^{t+1}$, and the last equality defining $\tilde{\theta}_2$. We then use the angle addition lemma (\cref{lemma:addingangles}) repeatedly to obtain $f$-measure$(I,R_X(\tilde{\theta}_2))$.

The total number of oracle calls used in this procedure is set by $\theta_2$, and $\epsilon$. 
More specifically, each time we halve an angle, we use 2 additional oracle calls to $f$-measure$(I,H)$. 
Thus to implement a single $\pi/2^{k+1}$ angle measurement oracle, we need $m_k=2k + 1$ Hadamard oracle calls, so the total number of oracle calls is then
\begin{align}
    m=\sum_{k=0}^{t} (2k+1)c_k \leq \sum_{k=0}^{t} (2k+1) = (t+1)^2.
\end{align}
To obtain $\epsilon$-accuracy, we need $\epsilon<\pi/2^{t+1}$, so we can take $t=\log(\pi/\epsilon) - 1$. 
Thus $m$ is bounded above by $[\log (\pi/\epsilon)]^2$. 
Thus if we take $m$ as above we can achieve any desired $\epsilon$. 
Since this does not depend on $n$, this achieves a $O(1)$ oracle reduction. 
\end{proof}

Note that \cref{thm:fmeasureEquivalence} holds when using $\epsilon$-approximate implementations of the $f$-measure$(U_1,V_1)$, in particular giving a $\delta$-approximate implementation of $f$-measure$(U_2,V_2)$.
However, we need to require $\epsilon$ is lower than some threshold error $\epsilon_0$, which will depend on $\delta,\theta_1,\theta_2$, since the required number of oracle calls depends on these parameters, and therefore by \cref{remark:diamondadditive} the additional error incurred also depends on these parameters.

Note that for any given rotation angle, a $f$-measure$(I,R_X(\theta))$ protocol is only non-trivial below a certain error threshold. 
This is because for small enough angles and allowing large enough errors, the task can be performed by simply measuring in the basis halfway in between the $f(x,y)=0$ and $f(x,y)=1$ bases and returning the output.

We finish by giving some consequences of \cref{thm:fmeasureEquivalence} for security of QPV schemes. 
First, our $O(1)$ reduction means that all single-qubit projective $f$-measure tasks inherit the upper bounds known previously for $f$-route, including sub-exponential attacks for all functions.

\begin{corollary}\label{cor:singlequbitconsequences}
All non-trivial single-qubit projective $f$-measure tasks have subexponential-size NLQC implementations. 
Consequently, when these NLQCs are viewed as attacks on QPV protocols, all corresponding single-qubit projective $f$-measure QPV protocols have subexponential entanglement attacks. 
\end{corollary}

\begin{proof}
By \cref{thm:frouteTOfmeasure}, $f$-route and $\neg f$-route imply $f$-measure$(I,H)$ with constant overhead. 
By \cref{thm:fmeasureEquivalence}, $f$-measure$(I,H)$ implies every non-trivial single-qubit projective $f$-measure task with constant overhead. 
Thus the subexponential NLQC implementation of $f$-route \cite{allerstorfer2024relating}, which also applies to $\neg f$, gives the claim by composition. 
\end{proof}

Further, single-qubit projective $f$-measure tasks inherit the entanglement upper bound for $f$-route based on span-programs \cite{cree2023code}.

\section{Controlled Clifford operations}\label{sec:controlledCliffords}

In this section we show that $f$-measure$(I,H)$ implies a number of multi-qubit tasks, for instance the $f$-controlled application of an arbitrary Clifford unitary. 
We first show that $f$-measure$(I,H)$ is equivalent to an $f$-controlled measurement in the Bell basis. 
Then, exploiting the ability to measure controllably in the Bell basis, we can use teleportation-like strategies to control where input states move, and thereby control which unitary operations act on them. 
Since in particular the SWAP gate is a Clifford, this also lets us implement $f$-route. 
By \cref{thm:frouteTOfmeasure}, $f$-route plus $\neg f$-route imply $f$-measure$(I,H)$ as oracles, which (essentially) shows the oracle equivalence of $f$-route and $f$-measure$(I,H)$.

\subsection{\texorpdfstring{$f$}{TEXT}-measure\texorpdfstring{$(I,H)$}{TEXT} is equivalent to \texorpdfstring{$f$}{TEXT}-Bell measure}

We begin by showing that $f$-measure$(I,H)$ implies $f$-Bell. 
Recall the definition of $f$-Bell, which was given as \cref{def:fBell}. 

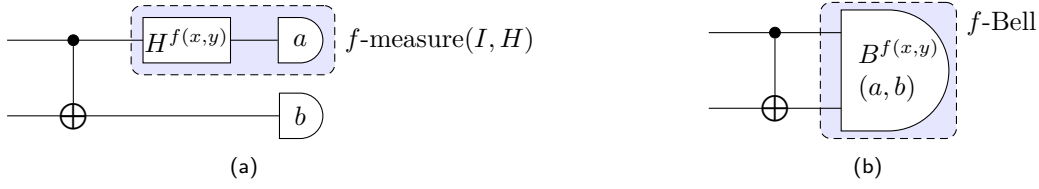
\begin{figure}
    \centering
    \begin{subfigure}{0.45\textwidth}
    \centering
    \begin{tikzpicture}[on grid,node distance=10mm]
          \node at (0,0)      (i1) {};
          \node[below=of i1]  (i2) {};
          \path (i1)  ++(4cm,0)     node[mymeasureselect] (m1) {$a$}
                      ++(0,-10mm)   node[mymeasureselect] (m2) {$b$};
          \path 
                (i1)    ++(1cm,0)  node[mycontrol]       (c1) {}
                        ++(0,-1cm) node[mytarget]        (t2) {$\bigoplus$}
                (c1) edge (t2)
                (i1)    ++(2.5cm,0) node[mySQG] (h1) {$H^{f(x,y)}$}
                ;
      \begin{scope}[on background layer]
        \node[draw=black,rounded corners,densely dashed,fill=blue,fill opacity=0.1,fit=(h1)(m1),inner sep=1.5mm](fit1) {};
        \node[right=0mm of fit1.east,anchor=west] (fitl1) {$f$-measure$\left(I,H\right)$};
          \path       (i1) edge (m1)
                      (i2) edge (m2);
      \end{scope}
\end{tikzpicture}
    \caption{}
    \label{fig:fHtofBell}
    \end{subfigure}
    \hfill
    \begin{subfigure}{0.45\textwidth}
    \centering
    \begin{tikzpicture}[on grid,node distance=10mm]
          \node at (0,0)      (i1) {};
          \node[below=of i1]  (i2) {};
          \path 
                (i1)    ++(1cm,0)  node[mycontrol]       (c1) {}
                        ++(0,-1cm) node[mytarget]        (t2) {$\bigoplus$}
                (c1) edge (t2)
                (i1)    ++(2.5cm,0)   coordinate (c1) 
                        ++(0,-1)    coordinate (c2)
                        node[myMQM=(c1)(c2)] (mqg) {}
                        node[align=left,text width=1cm] at ($(mqg.base)+(1mm,0mm)$) {$B^{f(x,y)}$ $(a,b)$}
                ;
      \begin{scope}[on background layer]
        \node[draw=black,rounded corners,densely dashed,fill=blue,fill opacity=0.1,fit=(c1)(c2),minimum size=1.8 cm,inner sep=1.5mm](fit1) {};
        \node[right=0mm of fit1.north east,anchor=north west] (fitl1) {$f$-Bell};
          \path       (i1) edge (c1)
                      (i2) edge (c2);
      \end{scope}
\end{tikzpicture}
    \caption{}
    \label{fig:fBelltofH}
    \end{subfigure}
    \caption{(a) Protocol to implement $f$-Bell using $f$-measure$(I,H)$. The blue box represents a use of an $f$-measure$(I,H)$ oracle. (b) Protocol to implement $f$-measure$(I,H)$ using $f$-Bell. The blue box represents a use of an $f$-Bell oracle.}
\end{figure}

\begin{lemma} \label{lem:fBB84tofBell}
    \textbf{$f$-measure$(I,H)$ $\Rightarrow$ $f$-Bell:} There is an $O(1)$ oracle implication from $f$-measure$(I,H)$ to $f$-Bell.
\end{lemma}

\begin{proof} 
We begin with a description of the $f$-Bell protocol and then check correctness below. 

\vspace{0.2cm}
\noindent \textbf{Protocol:} Alice and Bob's protocol is as follows.
First, Alice applies a CNOT to her input state. 
Then, Alice and Bob apply the $f$-measure$(I, H)$ protocol on the \textit{first} qubit and Alice always measures the second qubit in the computational basis, then broadcasts her measurement outcome. 
In the second round, Alice and Bob learn $f(x,y)$ and obtain the two measurement outcomes from the $f$-measure$(I, H)$ protocol and the computational basis measurement.
A circuit description of the protocol is given as \cref{fig:fHtofBell}. 

If $f(x,y)=0$, the $f$-measure$(I, H)$ protocol measures the first qubit in the computational basis and Alice measures the second qubit also in the computational basis, obtaining outcomes $(a,b)$. 
Alice and Bob consider the action of a CNOT,
\begin{align}\label{eq:CNOTpermutation}
    \ket{00} \xrightarrow{\operatorname{CNOT}} \ket{00} \nonumber \\
    \ket{01} \xrightarrow{\operatorname{CNOT}} \ket{01} \nonumber \\
    \ket{10} \xrightarrow{\operatorname{CNOT}} \ket{11} \nonumber \\
    \ket{11} \xrightarrow{\operatorname{CNOT}} \ket{10} 
\end{align}
and undo it, looking up the pre-image of $\ket{a,b}$ in this mapping and using that as their output. 

If $f(x,y)=1$, the $f$-measure$(I, H)$ protocol measures the first qubit in the Hadamard basis, obtaining one of the four outcomes $\ket{+}\ket{0}, \ket{+}\ket{1}, \ket{-}\ket{0}, \ket{-}\ket{1}$.
Labelling the 4 Bell states by $\ket{\Psi_{ij}}_{AB}=X_A^iZ_A^j\ket{\Psi_{00}}$, where $\ket{\Psi_{00}}=\frac{1}{\sqrt{2}}(\ket{00}+\ket{11})$, the pre-image of the four possible measurement outcomes under the CNOT are 
\begin{align}\label{eq:Bellpermutation}
    \ket{\Psi_{00}} &\xrightarrow{\operatorname{CNOT}} \ket{+}\ket{0} \nonumber \\
    \ket{\Psi_{01}} &\xrightarrow{\operatorname{CNOT}} \ket{-}\ket{0} \nonumber \\
    \ket{\Psi_{10}} &\xrightarrow{\operatorname{CNOT}} \ket{+} \ket{1} \nonumber \\
    \ket{\Psi_{11}} &\xrightarrow{\operatorname{CNOT}} \ket{-} \ket{1}
\end{align}
Alice and Bob output the label of the pre-image of their measurement outcome. 

\vspace{0.2cm}
\noindent \textbf{Correctness:} To check that this protocol is correct, we consider $(x,y)\in f^{-1}(0)$ and $(x,y)\in f^{-1}(1)$ instances separately. 
Begin with the $(x,y)\in f^{-1}(0)$ instances, in which case the protocol should implement the $\mathcal{B}^0$ channel. 
Write the input state in the computational basis, and then observe that before the measurement the protocol produces the state
\begin{align}
    \ket{\psi}_Q &= \alpha_{00} \ket{00} + \alpha_{01}\ket{01}+\alpha_{10}\ket{10} + \alpha_{11}\ket{11} \nonumber \\
    &\xrightarrow{\operatorname{CNOT}} \alpha_{00} \ket{00} + \alpha_{01}\ket{01}+\alpha_{10}\ket{11} + \alpha_{11}\ket{10}. 
\end{align}
Alice and Bob then both learn measurement outcomes from measuring both qubits in the computational basis, and so they obtain $\ket{a,b}$ with probability $|\alpha_{a,b\oplus a}|^2$.
They then output the pre-image of $\ket{a,b}$ under the CNOT, so they output $\ket{a,a\oplus b}$ with probability $|\alpha_{a,a\oplus b}|^2$. 
Their channel is therefore described by
\begin{align}
    \mathcal{N}^{x,y}_{Q\rightarrow OO'} (\psi_Q) = \sum_{a,b}|\alpha_{ab}|^2 \ketbra{a,b}{a,b}\otimes \ketbra{a,b}{a,b}
\end{align}
But since $|\alpha_{a,b}|^2=\bra{a,b} \psi_Q \ket{a,b}$, this is exactly the channel $\mathcal{B}^0$, as needed. 

Next we check correctness on $f(x,y)=1$ instances. 
We view the input in the Bell basis, and track the effect of the CNOT,
\begin{align}
    \ket{\psi}_Q &= \beta_{00}\ket{\Psi_{00}} + \beta_{01}\ket{\Psi_{01}} + \beta_{10}\ket{\Psi_{10}} + \beta_{11}\ket{\Psi_{11}} \nonumber \\
    &\xrightarrow{\operatorname{CNOT}} \beta_{00}\ket{+}\ket{0} + \beta_{01}\ket{-}\ket{0} + \beta_{10}\ket{+}\ket{1} + \beta_{11}\ket{-}\ket{1}.
\end{align}
The $f$-measure$(I,H)$ protocol now measures the first qubit in the Hadamard basis. 
Labelling outcome $\ket{+}$ as $0$ and $\ket{-}$ as $1$, Alice and Bob obtain outputs $(a',b')$, and then output $(a,b)$ according to
\begin{align}
    a'=0,b'=0\rightarrow a=0,b=0, \nonumber \\
    a'=0,b'=1\rightarrow a=1,b=0, \nonumber \\
    a'=1,b'=0\rightarrow a=0,b=1, \nonumber \\
    a'=1,b'=1\rightarrow a=1,b=1. 
\end{align}
This leads them to output $(a,b)$ with probability $|\beta_{ab}|^2$.
Thus, they implement the channel
\begin{align}
    \mathcal{N}_{Q\rightarrow ZZ'}^{x,y}(\psi_Q)=\sum_{a,b}|\beta_{a,b}|^2 \ketbra{a,b}{a,b}\otimes \ketbra{a,b}{a,b}
\end{align}
Since also $|\beta_{ab}|^2=\bra{\Psi_{ab}} \psi_Q \ket{\Psi_{ab}}$ this implements the measure-and-copy in the Bell basis channel $\mathcal{B}^1$, as needed.
\end{proof}

We now prove the reversed implication. 

\begin{lemma} \label{lem:fBelltofmeasure}
    \textbf{$f$-Bell $\Rightarrow$ $f$-measure$(I,H)$:} There is an $O(1)$ oracle implication from $f$-Bell to $f$-measure$(I,H)$.
\end{lemma}
\begin{proof}
We begin with a description of the $f$-measure$(I,H)$ protocol and then check correctness below. 

\vspace{0.2cm}
\noindent \textbf{Protocol:} We assume that Alice holds the quantum input, call it $Q$. 
As pre-processing, Alice prepares an ancilla qubit $\ket{0}_{A}$ and then applies $\operatorname{CNOT}_{Q\rightarrow A}$. 
Alice and Bob then run $f$-Bell on the QA system.
A circuit description of the protocol is given as \cref{fig:fBelltofH}. 

After the communication round, Alice and Bob both obtain the two bit output $(a,b)$ of the $f$-Bell protocol. 
Then, 
\begin{itemize}
    \item If $f=0$, they output the second bit, $b$. 
    \item If $f=1$, they output the first bit, $a$.
\end{itemize}
We claim this implements the $f$-measure$(I,H)$ task. 

\vspace{0.2cm}
\noindent \textbf{Correctness:} First suppose that $f(x,y)=0$. 
The input state, written in the computational basis, is
\begin{align}
    \ket{\psi}\ket{0}=\alpha_0 \ket{00} + \alpha_1\ket{10}.
\end{align}
After the CNOT is applied it is
\begin{align}
    \operatorname{CNOT} \ket{\psi}\ket{0}=\alpha_0 \ket{00} + \alpha_1\ket{11}.
\end{align}
If $f=0$, Alice and Bob output the result of measuring the second qubit in the computational basis, so they output $\ket{i}$ with probability $|\alpha_i|^2$. 
Since $|\alpha_i|^2=\tr(\Pi^0_i \ketbra{\psi}{\psi})$ with $\Pi_0^0=\ketbra{0}{0}, \Pi_1^0=\ketbra{1}{1}$, this means the channel they implement is
\begin{align}
    \mathcal{N}(\rho)=\sum_i \tr(\Pi^0_i \rho ) \ketbra{i}{i}\otimes \ketbra{i}{i}
\end{align}
as needed. 

Now suppose that $f(x,y)=1$. 
Then we view $\ket{\psi}$ in the Hadamard basis, so that the full input is
\begin{align}
    \ket{\psi}\ket{0} = \beta_0\ket{+}\ket{0} + \beta_1\ket{-}\ket{0}.
\end{align}
After the CNOT is applied in the pre-processing step, this becomes
\begin{align}
    \operatorname{CNOT}\ket{\psi}\ket{0} = \beta_0\ket{\Psi_{00}} + \beta_1\ket{\Psi_{01}}
\end{align}
Now Alice and Bob measure in the Bell basis, and output the first index $i$ in their result $\ket{\Psi_{ij}}$. 
We see that they get outcome $0$ with probability $|\beta_0|^2=|\braket{+}{\psi}|^2$, and outcome $1$ with probability $|\beta_1|^2=|\braket{-}{\psi}|^2$. 
Thus they implement exactly the channel
\begin{align}
    \mathcal{N}(\rho) = \sum_i \tr(\Pi_i^H\rho) \ketbra{i}{i}\otimes \ketbra{i}{i}.
\end{align}
where $\Pi_0^H=\ketbra{+}{+}, \Pi_1^H=\ketbra{-}{-}$, as needed. 
\end{proof}

\subsection{Clifford measurement protocols}\label{sec:Cliffordmeasurementprotocol}

Next we show that $f$-Bell implies $f$-controlled measurements in arbitrary Clifford bases. 

\begin{figure}
\centering
\makebox[\textwidth][c]{
\begin{tikzpicture}[on grid,node distance=10mm]
    \begin{pgfonlayer}{bg1}
    \path (0,0) pic (P1) {Bell};
    \end{pgfonlayer}
    \begin{pgfonlayer}{bg2}
    \path (0,1mm) pic (P2) [black] {Bell};
    \end{pgfonlayer}
    \begin{pgfonlayer}{bg3}
    \path (0,2mm) pic (P3) [black] {Bell};
    \end{pgfonlayer}
    \begin{pgfonlayer}{bg4}
    \path (0,3mm) pic (P4) [black] {Bell};
    \end{pgfonlayer}
    \begin{pgfonlayer}{front}
    \path (P1-M.base) node {$B^{f(x,y)}$};
    \node[fit=(P1-i1)(P2-i1)(P3-i1)(P4-i1),inner sep=0] (t1){};
    \node[anchor=east,left=5mm of t1] (state) {$\ket{\psi}$};
    \end{pgfonlayer}
    \begin{pgfonlayer}{bg5}
        \node[fit=(P1-i2)(P2-i2)(P3-i2)(P4-i2)(P1-i3)(P2-i3)(P3-i3)(P4-i3),inner sep=0.5mm,left delimiter=\{,minimum width=1cm,xshift=-0.1cm](E) {};
        \node[left=3mm of E.west,font=\small,anchor=east] (El) {$\ket{\Psi^+}^{\otimes n_Q}$};
        \node[fit=(P1-M)(P2-M)(P3-M)(P4-M),inner sep=1mm,densely dashed,rounded corners,draw=black,fill=blue,fill opacity=0.1,minimum width=1cm] (fB) {};
        \node[right=0mm of fB.north east,font=\small,anchor=north west] (fBl) {$(f\text{-Bell})^{\otimes n_Q}$};
        \node[fit=(P1-cM)(P2-cM)(P3-cM)(P4-cM),inner sep=1mm,densely dashed,rounded corners,draw=black,fill=red,fill opacity=0.1,minimum width=1cm] (CM) {};
        \node[right=0mm of CM.north east,font=\small,anchor=north west] (CMl) {$(\{\ket{0},\ket{1}\})^{\otimes n_Q}$};
    \end{pgfonlayer}
    \begin{pgfonlayer}{front}
        \node[fit=(P1-i1)(P2-i1)(P3-i1)(P4-i1),inner sep=0.3cm,minimum width=1cm,rectangle,draw=black,fill=white,xshift=3cm] (iCliff0){};
        \node[above=-1mm of iCliff0.base,anchor=base] (iCliff0l) {$C_0^\dagger$};
        \node[fit=(P1-i3)(P2-i3)(P3-i3)(P4-i3),inner sep=0.3cm,minimum width=1cm,rectangle,draw=black,fill=white,xshift=1cm] (Cliff0){};
        \node[above=-1mm of Cliff0.base,anchor=base] (Cliff0l) {$C_0$};
        \node[fit=(P1-i3)(P2-i3)(P3-i3)(P4-i3),inner sep=0.3cm,minimum width=1cm,rectangle,draw=black,fill=white,xshift=3cm] (Cliff1){};
        \node[above=-1mm of Cliff1.base,anchor=base] (Cliff1l) {$C_1^\dagger$};
    \end{pgfonlayer}
\end{tikzpicture}
}

\caption{Protocol to implement $f$-measure$(C_0,C_1)$ with $C_0,C_1$ Clifford using $f$-Bell as an oracle. For Cliffords on $n_Q$ qubits, $n_Q$ copies of $f$-Bell should be used. The curved wire denotes a maximally entangled state.}
\label{fig:fCliffordmeasurement}
\end{figure}
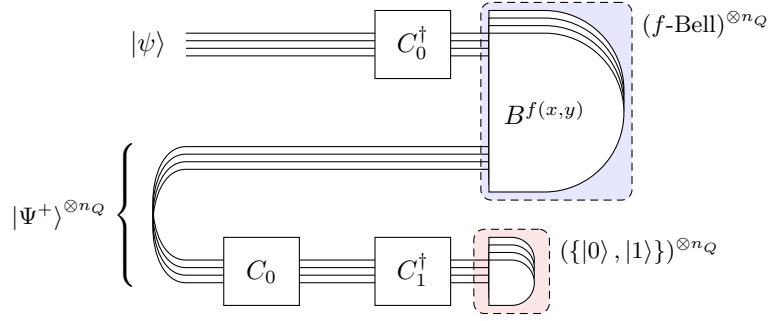

\CliffordMeasure*

\begin{proof} We first give the protocol then comment on correctness. 

\vspace{0.2cm}
\noindent \textbf{Protocol:} Alice prepares the maximally entangled state on $\mathcal{H}_{B}\otimes \mathcal{H}_C$ with dimensions $d_B=d_C=d_Q$, $Q$ the quantum input Hilbert space, then applies Clifford $(C_0)_Q^\dagger \otimes \mathcal{I}_B \otimes (C_1)_C^\dagger (C_0)_C $.
Alice measures the $C$ register in the computational basis and broadcasts the measurement outcome, call it $c$.
Then, Alice and Bob take $\mathcal{H}_Q\otimes \mathcal{H}_B$ as the input to $f$-Bell oracles (one oracle use per qubit of $Q$ is needed). 
The protocol is illustrated in \cref{fig:fCliffordmeasurement}. 

In the second round, Alice and Bob learn $x,y$ and hence $f(x,y)$. 
If $f(x,y)=0$, Alice and Bob record the first bit of the two bit measurement outcome from the $f$-Bell oracle as their outputs. 
Note that since $f(x,y)=0$ the oracle measures in the computational basis. 
If $f(x,y)=1$, observe that $QB$ has been measured in the Bell basis, so that the state in register $Q$ has been teleported to $\mathcal H_C$. 
The measurement outcomes from the $f$-Bell protocol then determine Pauli corrections on $C$. 
Alice and Bob determine these Pauli corrections and use them to undo the effect of the Pauli corrections on the measurement of $C$.
This process determines a measurement outcome $c'$, which Alice and Bob record as their output.

\vspace{0.2cm}
\noindent \textbf{Correctness:} We consider the $f=0$ and $f=1$ cases separately. 

Begin with the $f=0$ case. 
Then $f$-Bell measures the $QB$ subsystem in the computational basis, and upon obtaining measurement outcomes $a,b,c$ the circuit of \cref{fig:fCliffordmeasurement} becomes
\begin{center}

\begin{tikzpicture}[on grid,node distance=10mm]
    \begin{pgfonlayer}{bg1}
    \path (0,0) pic (P1) {Bell0};
    \end{pgfonlayer}
    \begin{pgfonlayer}{bg2}
    \path (0,1mm) pic (P2) [black] {Bell0};
    \end{pgfonlayer}
    \begin{pgfonlayer}{bg3}
    \path (0,2mm) pic (P3) [black] {Bell0};
    \end{pgfonlayer}
    \begin{pgfonlayer}{bg4}
    \path (0,3mm) pic (P4) [black] {Bell0};
    \end{pgfonlayer}
    \begin{pgfonlayer}{front}
    \path (P1-cM1.base) node {$\mathbf{a}$};
    \path (P1-cM2.base) node {$\mathbf{b}$};
    \path (P1-cM3.base) node {$\mathbf{c}$};
    \node[fit=(P1-i1)(P2-i1)(P3-i1)(P4-i1),inner sep=0] (t1){};
    \node[anchor=east,left=5mm of t1] (state) {$\ket{\psi}$};
    \end{pgfonlayer}
    \begin{pgfonlayer}{bg5}
        \node[fit=(P1-i2)(P2-i2)(P3-i2)(P4-i2)(P1-i3)(P2-i3)(P3-i3)(P4-i3),inner sep=0.5mm,left delimiter=\{,minimum width=1cm,xshift=-0.1cm](E) {};
        \node[left=3mm of E.west,font=\small,anchor=east] (El) {$\ket{\Psi^+}^{\otimes n_Q}$};
        \node[fit=(P1-cM1)(P2-cM1)(P3-cM1)(P4-cM1),inner sep=1mm,densely dashed,rounded corners,draw=black,fill=red,fill opacity=0.1,minimum width=1cm] (CM1) {};
        \node[right=0mm of CM1.north east,font=\small,anchor=north west] (CM1l) {$(\{\ket{0},\ket{1}\})^{\otimes n_Q}$};
    \end{pgfonlayer}
    \begin{pgfonlayer}{front}
        \node[fit=(P1-i1)(P2-i1)(P3-i1)(P4-i1),inner sep=0.3cm,minimum width=1cm,rectangle,draw=black,fill=white,xshift=3cm] (iCliff0){};
        \node[above=-1mm of iCliff0.base,anchor=base] (iCliff0l) {$C_0^\dagger$};
        \node[fit=(P1-i3)(P2-i3)(P3-i3)(P4-i3),inner sep=0.3cm,minimum width=1cm,rectangle,draw=black,fill=white,xshift=1cm] (Cliff0){};
        \node[above=-1mm of Cliff0.base,anchor=base] (Cliff0l) {$C_0$};
        \node[fit=(P1-i3)(P2-i3)(P3-i3)(P4-i3),inner sep=0.3cm,minimum width=1cm,rectangle,draw=black,fill=white,xshift=3cm] (Cliff1){};
        \node[above=-1mm of Cliff1.base,anchor=base] (Cliff1l) {$C_1^\dagger$};
    \end{pgfonlayer}
\end{tikzpicture}

\end{center}
\noindent We see that the input state has been measured exactly in the basis determined by $C_0$. 
Since both Alice and Bob output $a$, they correctly implement the measure-and-copy channel corresponding to $C_0$. 

Next consider the $f=1$ case. 
Then $QB$ are measured in the Bell basis, obtaining measurement outcomes $\mathbf{b},\mathbf{a}$ $\in \{0,1\}^{n_Q}$ corresponding to the Bell state $X_B^{\mathbf{b}}Z_B^{\mathbf{a}}\ket{\Psi_{00}}$. 
The $C$ system is measured in the computational basis; we label the outcome as $c$. 
Then the circuit of \cref{fig:fCliffordmeasurement} becomes
\begin{center}
\begin{tikzpicture}[on grid,node distance=10mm]
    \begin{pgfonlayer}{bg1}
    \path (0,0) pic (P1) {Bell1};
    \end{pgfonlayer}
    \begin{pgfonlayer}{bg2}
    \path (0,1mm) pic (P2) [black] {Bell1};
    \end{pgfonlayer}
    \begin{pgfonlayer}{bg3}
    \path (0,2mm) pic (P3) [black] {Bell1};
    \end{pgfonlayer}
    \begin{pgfonlayer}{bg4}
    \path (0,3mm) pic (P4) [black] {Bell1};
    \end{pgfonlayer}
    \begin{pgfonlayer}{front}
    \path (P1-cM3.base) node {$\mathbf{c}$};
    \node[fit=(P1-i1)(P2-i1)(P3-i1)(P4-i1),inner sep=0] (t1){};
    \node[anchor=east,left=5mm of t1] (state) {$\ket{\psi}$};
    \end{pgfonlayer}
    \begin{pgfonlayer}{bg5}
        \node[fit=(P1-i2)(P2-i2)(P3-i2)(P4-i2)(P1-i3)(P2-i3)(P3-i3)(P4-i3),inner sep=0.5mm,left delimiter=\{,minimum width=1cm,xshift=-0.1cm](E) {};
        \node[left=3mm of E.west,font=\small,anchor=east] (El) {$\ket{\Psi^+}^{\otimes n_Q}$};
    
    \end{pgfonlayer}
    \begin{pgfonlayer}{front}
        \node[fit=(P1-i1)(P2-i1)(P3-i1)(P4-i1),inner sep=0.3cm,minimum width=1cm,rectangle,draw=black,fill=white,xshift=3cm] (iCliff0){};
        \node[above=-1mm of iCliff0.base,anchor=base] (iCliff0l) {$C_0^\dagger$};
        \node[fit=(P1-i2)(P2-i2)(P3-i2)(P4-i2),inner sep=0.3cm,minimum width=1cm,rectangle,draw=black,fill=white,xshift=1cm] (Z){};
        \node[above=-1mm of Z.base,anchor=base] (Zl) {$Z^{\mathbf{a}}$};
        \node[fit=(P1-i2)(P2-i2)(P3-i2)(P4-i2),inner sep=0.3cm,minimum width=1cm,rectangle,draw=black,fill=white,xshift=3cm] (X){};
        \node[above=-1mm of X.base,anchor=base] (Xl) {$X^{\mathbf{b}}$};
        \node[fit=(P1-i3)(P2-i3)(P3-i3)(P4-i3),inner sep=0.3cm,minimum width=1cm,rectangle,draw=black,fill=white,xshift=1cm] (Cliff0){};
        \node[above=-1mm of Cliff0.base,anchor=base] (Cliff0l) {$C_0$};
        \node[fit=(P1-i3)(P2-i3)(P3-i3)(P4-i3),inner sep=0.3cm,minimum width=1cm,rectangle,draw=black,fill=white,xshift=3cm] (Cliff1){};
        \node[above=-1mm of Cliff1.base,anchor=base] (Cliff1l) {$C_1^\dagger$};
    \end{pgfonlayer}
\end{tikzpicture}
\end{center}

\noindent We can move the Pauli operators to the $C$ system, and then conjugate them through $C_0C_1^\dagger$ to obtain

\begin{center}
\begin{tikzpicture}[on grid,node distance=10mm]
    \begin{pgfonlayer}{bg1}
    \path (0,0) pic (P1) {Bell1b};
    \end{pgfonlayer}
    \begin{pgfonlayer}{bg2}
    \path (0,1mm) pic (P2) [black] {Bell1b};
    \end{pgfonlayer}
    \begin{pgfonlayer}{bg3}
    \path (0,2mm) pic (P3) [black] {Bell1b};
    \end{pgfonlayer}
    \begin{pgfonlayer}{bg4}
    \path (0,3mm) pic (P4) [black] {Bell1b};
    \end{pgfonlayer}
    \begin{pgfonlayer}{front}
    \path (P1-cM3.base) node {$\mathbf{c}$};
    \node[fit=(P1-i1)(P2-i1)(P3-i1)(P4-i1),inner sep=0] (t1){};
    \node[anchor=east,left=5mm of t1] (state) {$\ket{\psi}$};
    \end{pgfonlayer}
    \begin{pgfonlayer}{bg5}
        \node[fit=(P1-i2)(P2-i2)(P3-i2)(P4-i2)(P1-i3)(P2-i3)(P3-i3)(P4-i3),inner sep=0.5mm,left delimiter=\{,minimum width=1cm,xshift=-0.1cm](E) {};
        \node[left=3mm of E.west,font=\small,anchor=east] (El) {$\ket{\Psi^+}^{\otimes n_Q}$};
    
    \end{pgfonlayer}
    \begin{pgfonlayer}{front}
        \node[fit=(P1-i1)(P2-i1)(P3-i1)(P4-i1),inner sep=0.3cm,minimum width=1cm,rectangle,draw=black,fill=white,xshift=3cm] (iCliff0){};
        \node[above=-1mm of iCliff0.base,anchor=base] (iCliff0l) {$C_0^\dagger$};
        \node[fit=(P1-i3)(P2-i3)(P3-i3)(P4-i3),inner sep=0.3cm,minimum width=1cm,rectangle,draw=black,fill=white,xshift=7cm] (Z){};
        \node[above=-1mm of Z.base,anchor=base] (Zl) {$Z^{\mathbf{a'}}$};
        \node[fit=(P1-i3)(P2-i3)(P3-i3)(P4-i3),inner sep=0.3cm,minimum width=1cm,rectangle,draw=black,fill=white,xshift=5cm] (X){};
        \node[above=-1mm of X.base,anchor=base] (Xl) {$X^{\mathbf{b'}}$};
        \node[fit=(P1-i3)(P2-i3)(P3-i3)(P4-i3),inner sep=0.3cm,minimum width=1cm,rectangle,draw=black,fill=white,xshift=1cm] (Cliff0){};
        \node[above=-1mm of Cliff0.base,anchor=base] (Cliff0l) {$C_0$};
        \node[fit=(P1-i3)(P2-i3)(P3-i3)(P4-i3),inner sep=0.3cm,minimum width=1cm,rectangle,draw=black,fill=white,xshift=3cm] (Cliff1){};
        \node[above=-1mm of Cliff1.base,anchor=base] (Cliff1l) {$C_1^\dagger$};
    \end{pgfonlayer}
\end{tikzpicture}
\end{center}

\noindent Note that commuting the $Z$ and $X$ corrections through the Cliffords changes the exponents $\mathbf{a},\mathbf{b}$ to $\mathbf{a}',\mathbf{b}'$.
Notice that the $Z$ operator acts again only by a global phase on the output $\ket{c}$, so can be dropped. 
Then we obtain
\begin{center}
\begin{tikzpicture}[on grid,node distance=10mm]
    \begin{pgfonlayer}{bg1}
    \path (0,0) pic (P1) {Bell1c};
    \end{pgfonlayer}
    \begin{pgfonlayer}{bg2}
    \path (0,1mm) pic (P2) [black] {Bell1c};
    \end{pgfonlayer}
    \begin{pgfonlayer}{bg3}
    \path (0,2mm) pic (P3) [black] {Bell1c};
    \end{pgfonlayer}
    \begin{pgfonlayer}{bg4}
    \path (0,3mm) pic (P4) [black] {Bell1c};
    \end{pgfonlayer}
    \begin{pgfonlayer}{front}
    \path (P1-cM3.base) node {$\mathbf{c}\oplus\mathbf{b'}$};
    \node[fit=(P1-i1)(P2-i1)(P3-i1)(P4-i1),inner sep=0] (t1){};
    \node[anchor=east,left=5mm of t1] (state) {$\ket{\psi}$};
    \end{pgfonlayer}
    \begin{pgfonlayer}{bg5}
        \node[fit=(P1-i2)(P2-i2)(P3-i2)(P4-i2)(P1-i3)(P2-i3)(P3-i3)(P4-i3),inner sep=0.5mm,left delimiter=\{,minimum width=1cm,xshift=-0.1cm](E) {};
        \node[left=3mm of E.west,font=\small,anchor=east] (El) {$\ket{\Psi^+}^{\otimes n_Q}$};
    
    \end{pgfonlayer}
    \begin{pgfonlayer}{front}
        \node[fit=(P1-i1)(P2-i1)(P3-i1)(P4-i1),inner sep=0.3cm,minimum width=1cm,rectangle,draw=black,fill=white,xshift=3cm] (iCliff0){};
        \node[above=-1mm of iCliff0.base,anchor=base] (iCliff0l) {$C_0^\dagger$};
        \node[fit=(P1-i3)(P2-i3)(P3-i3)(P4-i3),inner sep=0.3cm,minimum width=1cm,rectangle,draw=black,fill=white,xshift=1cm] (Cliff0){};
        \node[above=-1mm of Cliff0.base,anchor=base] (Cliff0l) {$C_0$};
        \node[fit=(P1-i3)(P2-i3)(P3-i3)(P4-i3),inner sep=0.3cm,minimum width=1cm,rectangle,draw=black,fill=white,xshift=3cm] (Cliff1){};
        \node[above=-1mm of Cliff1.base,anchor=base] (Cliff1l) {$C_1^\dagger$};
    \end{pgfonlayer}
\end{tikzpicture}

\end{center}
\noindent Straightening the bent wire and using $C_0^\dagger C_0=I$, this is just
\begin{center}
\begin{tikzpicture}[on grid,node distance=10mm]
    \begin{pgfonlayer}{bg1}
    \path (0,0) pic (P1) {Bell1d};
    \end{pgfonlayer}
    \begin{pgfonlayer}{bg2}
    \path (0,1mm) pic (P2) [black] {Bell1d};
    \end{pgfonlayer}
    \begin{pgfonlayer}{bg3}
    \path (0,2mm) pic (P3) [black] {Bell1d};
    \end{pgfonlayer}
    \begin{pgfonlayer}{bg4}
    \path (0,3mm) pic (P4) [black] {Bell1d};
    \end{pgfonlayer}
    \begin{pgfonlayer}{front}
    \path (P1-cM3.base) node {$\mathbf{c}\oplus\mathbf{b'}$};
    \node[fit=(P1-i3)(P2-i3)(P3-i3)(P4-i3),inner sep=0] (t1){};
    \node[anchor=east,left=5mm of t1] (state) {$\ket{\psi}$};
    \end{pgfonlayer}
    \begin{pgfonlayer}{front}
        \node[fit=(P1-i3)(P2-i3)(P3-i3)(P4-i3),inner sep=0.3cm,minimum width=1cm,rectangle,draw=black,fill=white,xshift=1cm] (Cliff1){};
        \node[above=-1mm of Cliff1.base,anchor=base] (Cliff1l) {$C_1^\dagger$};
    \end{pgfonlayer}
\end{tikzpicture}
\end{center}
Thus if Alice and Bob compute $\mathbf{b}'$ (from $\mathbf{b},\mathbf{a}$ and given the choices of $C_0,C_1$), they can correct the measurement outcome by XORing with $\mathbf{b}'$. 
They then output $c$, which is distributed as if they measured directly in the $C_1$ basis, as needed. 
\end{proof}

Note that we could also have allowed only Alice to know the descriptions of $C_0, C_1$ at the start of the protocol, since only she applies them in the first round. 
In the second round Bob needs to know $C_0, C_1$ but Alice can communicate their description in the communication round. 

\subsection{Clifford unitary protocols}

We will now show that any $f$-unitary$(C_0,C_1)$ task, where $C_0,C_1$ are $n$ qubit Clifford operations, reduces to $3n_Q$ $f$-Bell measure protocols.
Note that whereas previously we reduced measuring in Clifford bases to $f$-Bell, now we reduce the application of Clifford unitaries to $f$-Bell. 
This holds for any splitting of where the outputs of $C_0$, $C_1$ should end up (i.e.\ any subset of qubits can be output on Alice's side, and the complement output on Bob's side).
The key idea is to use an $f$-measure$(B_{QB}\otimes I_C, B_{QC}\otimes I_B)$ protocol, where $B$ is the unitary mapping the computational to the Bell basis. 
More concretely, for $X_i, Y_i$ single qubits we have
\begin{align}
    B_{X_iY_i}=\operatorname{CNOT}_{X_i\rightarrow Y_i} H_{X_i}
\end{align}
and then the $n$ qubit generalization repeats this in parallel for each pair of qubits in $X$, $Y$.
We call this the $n$-qubit Bell basis unitary. 
The $f$-measure$(B_{QB}\otimes I_C, B_{QC}\otimes I_B)$ protocol effectively teleports the input into either $B$ or $C$, where we can selectively perform either $C_0$ or $C_1$. 

\begin{figure}
\begin{center}
\begin{tikzpicture}[on grid,node distance=10mm]
    \begin{pgfonlayer}{bg1}
    \path (0,0) pic (P1) {AlternatingBellwunitary};
    \end{pgfonlayer}
    \begin{pgfonlayer}{bg2}
    \path (0,1mm) pic (P2) [black] {AlternatingBellwunitary};
    \end{pgfonlayer}
    \begin{pgfonlayer}{bg3}
    \path (0,2mm) pic (P3) [black] {AlternatingBellwunitary};
    \end{pgfonlayer}
    \begin{pgfonlayer}{bg4}
    \path (0,3mm) pic (P4) [black] {AlternatingBellwunitary};
    \end{pgfonlayer}
    \begin{pgfonlayer}{front}
    \path (P1-M.base) node {$\mathfrak{B}^{f(x,y)}$};
    \node[fit=(P1-i1)(P2-i1)(P3-i1)(P4-i1),inner sep=0] (t1){};
    \node[anchor=east,left=5mm of t1] (state) {$\ket{\psi}$};
    \end{pgfonlayer}
    \begin{pgfonlayer}{bg5}
        \node[fit=(P1-i2)(P2-i2)(P3-i2)(P4-i2)(P1-i5)(P2-i5)(P3-i5)(P4-i5),inner sep=0.5mm,left delimiter=\{,minimum width=1cm,xshift=-0.1cm](E) {};
        \node[left=3mm of E.west,font=\small,anchor=east] (El) {$\ket{\Phi^+}^{\otimes 2n_Q}$};
        \node[fit=(P1-M)(P2-M)(P3-M)(P4-M),inner sep=1mm,densely dashed,rounded corners,draw=black,fill=blue,fill opacity=0.1,minimum width=1cm] (fB) {};
        \node[right=0mm of fB.north east,font=\small,anchor=north west] (fBl) {$(\mathfrak{B}^{f(x,y)})^{\otimes n_Q}$};

    \end{pgfonlayer}
    \begin{pgfonlayer}{front}
        \node[fit=(P1-i4)(P2-i4)(P3-i4)(P4-i4),inner sep=0.3cm,minimum width=1cm,rectangle,draw=black,fill=white,xshift=3cm] (Cliff0){};
        \node[above=-1mm of Cliff0.base,anchor=base] (Cliff0l) {$C_0$};
        \node[fit=(P1-i5)(P2-i5)(P3-i5)(P4-i5),inner sep=0.3cm,minimum width=1cm,rectangle,draw=black,fill=white,xshift=3cm] (Cliff1){};
        \node[above=-1mm of Cliff1.base,anchor=base] (Cliff1l) {$C_1$};
        \node[fit=(P1-m4)(P2-m4)(P3-m4)(P4-m4),inner sep=0.3cm,minimum width=1cm,rectangle,draw=black,fill=white,xshift=-1cm] (Z0){};
        \node[above=-1mm of Z0.base,anchor=base] (Z0l) {$Z^{\bold{a'}}$};
        \node[fit=(P1-m4)(P2-m4)(P3-m4)(P4-m4),inner sep=0.3cm,minimum width=1cm,rectangle,draw=black,fill=white,xshift=-2.5cm] (X0){};
        \node[above=-1mm of X0.base,anchor=base] (X0l) {$X^{\bold{b'}}$};
        \node[fit=(P1-m5)(P2-m5)(P3-m5)(P4-m5),inner sep=0.3cm,minimum width=1cm,rectangle,draw=black,fill=white,xshift=-1cm] (Z1){};
        \node[above=-1mm of Z1.base,anchor=base] (Z1l) {$Z^{\bold{a'}}$};
        \node[fit=(P1-m5)(P2-m5)(P3-m5)(P4-m5),inner sep=0.3cm,minimum width=1cm,rectangle,draw=black,fill=white,xshift=-2.5cm] (X1){};
        \node[above=-1mm of X1.base,anchor=base] (X1l) {$X^{\bold{b'}}$};
    \end{pgfonlayer}
    \begin{pgfonlayer}{bg5}
        \node[fit=(X0)(Z0),inner sep=1mm,densely dashed,rounded corners,draw=black,fill=red,fill opacity=0.1,minimum width=1cm] (CM0) {};
        \node[right=0mm of CM0.north east,font=\small,anchor=north west,yshift=0.1cm] (CMl) {$f=0$};
        \node[fit=(X1)(Z1),inner sep=1mm,densely dashed,rounded corners,draw=black,fill=red,fill opacity=0.1,minimum width=1cm] (CM1) {};
        \node[right=0mm of CM1.north east,font=\small,anchor=north west,yshift=0.1cm] (CM1l) {$f=1$};
    \end{pgfonlayer}
\end{tikzpicture}
\end{center}
    \caption{Circuit for implementing $f$-unitary$(C_0,C_1)$ using an $f$-measure oracle. The $f$-measure oracle, labelled ${B}^{f(x,y)}$, measures either the first and second wires in the Bell basis (pairwise, for each of the $n_Q$ qubits in each wire), or the first and third wires in the Bell basis. This teleports the input into the selected unitary.}
    \label{fig:cliffordunitary}
\end{figure}
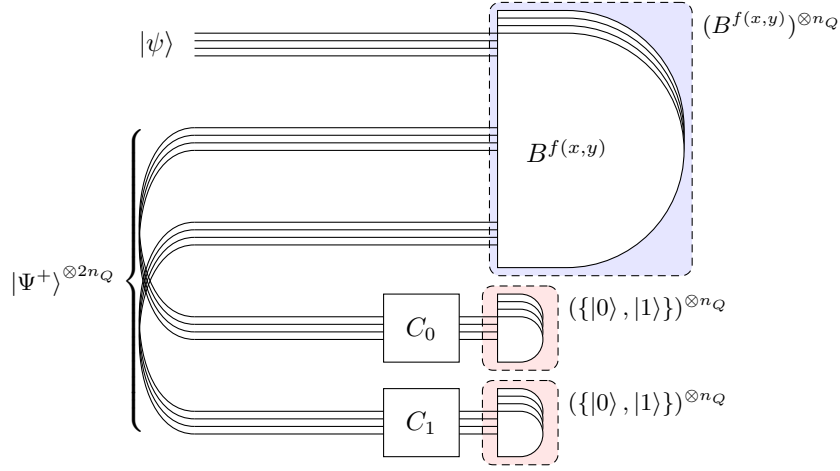

This is stated more carefully in the next theorem.

\begin{theorem}[$f$-Bell $\Rightarrow$ $f$-unitary Clifford] \label{thm:fcliffordunitary}
There is an $O(n_Q)$ oracle implication from $f$-Bell
to $f$-unitary$(C_0,C_1)$, for any pair of Clifford unitaries $C_0,C_1$
acting on $n_Q$ qubits.
\end{theorem}

\begin{proof}
    The protocol to implement $f$-unitary$(C_0,C_1)$ is shown in \cref{fig:cliffordunitary}. 
    The basic idea is to start with two maximally entangled states, $\ket{\Psi^+}_{BB'}$ and $\ket{\Psi^+}_{CC'}$. 
    We apply $C_0$ to $B'$, and $C_1$ to $C'$. 
    Letting $Q$ denote the input system, we use a controlled Clifford measurement to either measure $QC$ in the Bell basis, or measure $QB$ in the Bell basis. 
    Formally, we use a $f$-measure$(B_{QB}\otimes I_C,B_{QC}\otimes I_B)$ as an oracle; the measurement outcomes from this protocol are made available on both sides after the communication round.
    This teleports the input into either $B'$ or $C'$, up to Pauli corrections, where either $C_0$ or $C_1$ has been applied. 
    The $B'$ and $C'$ wires, which may consist of $n$ qubits, are split into two portions, with the qubits output from Alice's side kept by Alice, and the qubits output on Bob's side sent to Bob in the communication round. 
    
    This circuit will implement the correct unitary up to Pauli corrections, which are fully determined by the measurement outcomes. 
    Alice and Bob use the measurement outcomes to determine the needed Pauli corrections and then implement them in the second round, which they can always as they are single qubit operations. 

    We omit a formal proof of correctness, which is clear. 
\end{proof}

We also obtain the following corollary. 

\begin{corollary}[$f$-measure$(I,H)$ $\Rightarrow$ $f$-route]
\label{cor:routereducedtomeasureH}
There is an $O(1)$ oracle implication from $f$-measure$(I,H)$ to $f$-route.
\end{corollary}

\begin{proof}
By \Cref{lem:fBB84tofBell}, $f$-measure$(I,H)$ implies $f$-measure(Bell).
By \Cref{thm:fcliffordunitary}, $f$-measure(Bell) implies
$f$-unitary$(C_0,C_1)$ for any pair of Clifford unitaries $C_0,C_1$ acting on
the relevant constant-size input system, and for arbitrary choices of how the
output systems are split between Alice and Bob.

We choose the output split in which the first output system is held by Alice
and the second output system is held by Bob. We then take
\begin{align}
    C_0 = I, \qquad C_1 = \mathrm{SWAP}.
\end{align}
This gives an implementation of $f$-unitary$(I,\mathrm{SWAP})$ with one input
system held by Alice and one output system assigned to each party.

To implement $f$-route on an input qubit $Q$, Alice runs this
$f$-unitary$(I,\mathrm{SWAP})$ protocol on $Q$ together with a dummy qubit
initialized to $\ket{0}$. If $f(x,y)=0$, the identity is applied and the input
qubit remains in Alice's output system. If $f(x,y)=1$, the swap is applied and
the input qubit is moved to Bob's output system. Tracing out the dummy output
therefore gives exactly the $f$-route task.
\end{proof}

\section{Controlled diagonal unitaries}\label{sec:controlleddiagonal}

Our goal in this section is to show that we can do $f$-unitary$(I,U)$ for any $U$ of the form $U=C_1 DC_0$, where $D$ is an arbitrary diagonal unitary and $C_0$, $C_1$ are Cliffords, via an $O(1)$ reduction to $f$-measure$(I,H)$. 
This is the most general $2\rightarrow 2$ task that we are able to reduce to $f$-measure$(I,H)$.
The inspiration behind our construction is that we can perform (non-adaptive, graph state) measurement based quantum computation, using the single qubit $f$-measure protocols that are implied by $f$-measure$(I,H)$.
That diagonal unitaries can be implemented using non-adaptive graph state MBQC is well known, and we refer to the tutorial \cite{browne2016one} for a detailed explanation.
In the following we will present the proof in a simplified form without using the language of measurement based quantum computation.

\subsection{\texorpdfstring{$Z$}{TEXT}-rotation injection gadgets}

In our protocol to implement unitaries of the form $U=C_1DC_0$, we will decompose $D$ into gates of the form
\begin{align}
    R_{Z,S}(\phi)=\exp\left(-\frac{i\phi}{2} \bigotimes_{j\in S}Z_j\right).
\end{align}
Here $S$ labels a subset of qubits in an $n$ qubit Hilbert space. 
Any diagonal unitary can be written as a product of these unitaries.\footnote{This is commonly used, e.g.\ \cite{browne2016one}, but typically stated without proof. 
We prove it here for convenience, see \cref{lemma:Fourierdecomposition}.} Our protocol will require that we implement such gates controlled on the value of $f(x,y)$ in the form of an NLQC. 
While we know how to reduce single qubit rotations to $f$-measure$(I,H)$, the challenge now is to address the multi-qubit case. 
Our strategy will be to relate single qubit $X$ rotations to multi-qubit $Z$ rotations. 

To understand how to do this, our starting point is the identity,
\begin{center}
\begin{tikzpicture}[on grid,node distance=10mm]
          \node at (0,0)      (i1) {};
          \node[below=of i1]  (i2) {};
          \path (i1)  ++(3cm,0)     node (o1) {}
                      ++(0,-10mm)   node (o2) {};
          \path   
                  (i1)    ++(1cm,0)  node[mycontrol]       (c1) {}
                          ++(0,-1cm) node[mycontrol]     (c2) {}
                  (c1) edge (c2)
                  (i2)    ++(2cm,0)  node[mySQG]       (g1) {$X$}  
                  ;
      \begin{scope}[on background layer]
          \path       (i1) edge (o1)
                      (i2) edge (o2);
      \end{scope}
        \path (o1) ++(1cm,-0.5cm) node (e1) {$=$};
          \path (o1) ++(2cm,0) node (Ri1) {};
          \node[below=of Ri1]        (Ri2) {};
          \path (Ri1)  ++(4cm,0)     node (Ro1) {}
                      ++(0,-10mm)    node (Ro2) {};
          \path   
                  (Ri1)    ++(2cm,0)  node[mycontrol]       (Rc1) {}
                          ++(0,-1cm) node[mycontrol]     (Rc2) {}
                  (Rc1) edge (Rc2)
                  (Ri1)    ++(3cm,0)  node[mySQG]       (Rg1) {$Z$}  
                  (Ri2)    ++(1cm,0)  node[mySQG]       (Rg2) {$X$}  
                  ;
      \begin{scope}[on background layer]
          \path       (Ri1) edge (Ro1)
                      (Ri2) edge (Ro2);
      \end{scope}
\end{tikzpicture}
\end{center}
\noindent Note that the $Z$ gate commutes with the $CZ$, so it is just a convention to write the $Z$ gate on the right side of the $CZ$ in this identity.

A key observation we use is that we can use multiple $CZ$ gates to relate single qubit $X$ operators and multi-qubit $Z$ operators via this identity. 
For instance, consider that
\begin{center}
\begin{tikzpicture}[on grid,node distance=10mm]
          \node at (0,0)      (i1) {};
          \node[below=of i1]  (i2) {};
          \node[below=of i2]  (i3) {};
          \path (i1)  ++(3cm,0)     node (o1) {}
                      ++(0,-10mm)   node (o2) {}
                      ++(0,-10mm)   node (o3) {};
          \path   
                  (i1)    ++(1cm,0)  node[mycontrol]       (c1) {}
                          ++(0,-2cm) node[mycontrol]     (c3) {}
                  (c1) edge (c3)
                  (i2)    ++(2cm,0)  node[mycontrol]       (c2) {}
                          ++(0,-1cm) node[mycontrol]     (c3b) {}
                  (c2) edge (c3b)
                  (i3)    ++(3cm,0)  node[mySQG]       (g3) {$X$}  
                  ;
      \begin{scope}[on background layer]
          \path       (i1) edge (o1)
                      (i2) edge (o2)
                      (i3) edge (o3);
      \end{scope}
        \path (o1) ++(1cm,-1cm) node (e1) {$=$};
          \path (o1) ++(2cm,0) node (Ri1) {};
          \node[below=of Ri1]        (Ri2) {};
          \node[below=of Ri2]        (Ri3) {};
          \path (Ri1)  ++(5cm,0)     node (Ro1) {}
                      ++(0,-10mm)    node (Ro2) {}
                      ++(0,-10mm)    node (Ro3) {};
          \path   
                  (Ri1)    ++(2cm,0)  node[mycontrol]       (Rc1) {}
                          ++(0,-2cm) node[mycontrol]     (Rc3) {}
                  (Rc1) edge (Rc3)
                  (Ri2)    ++(3cm,0)  node[mycontrol]       (Rc1) {}
                          ++(0,-10mm) node[mycontrol]     (Rc3) {}
                  (Rc1) edge (Rc3)
                  (Ri1)    ++(4cm,0)  node[mySQG]       (Rg1) {$Z$}
                  (Ri2)    ++(4cm,0)  node[mySQG]       (Rg2) {$Z$}  
                  (Ri3)    ++(1cm,0)  node[mySQG]       (Rg3) {$X$}  
                  ;
      \begin{scope}[on background layer]
          \path       (Ri1) edge (Ro1)
                      (Ri2) edge (Ro2)
                      (Ri3) edge (Ro3);
      \end{scope}
\end{tikzpicture}
\end{center}
\noindent This follows by applying the identity above twice. 

We can generalize this observation to relate $\bigotimes_{j\in S} Z_j$ with an arbitrary choice of subset $S$ to a single-qubit $X$ operation. 
In particular, we have 
\begin{align}
    X_{k}\left(\bigotimes_{j\in S} CZ_{jk} \right) = \left( \bigotimes_{j\in S} Z_j \right)\left(\bigotimes_{j\in S} CZ_{jk} \right) X_k.
\end{align}
Using that
\begin{align}
    R_{Z,S}(\phi)&=\cos\left(\frac{\phi}{2}\right) I -i\sin\left(\frac{\phi}{2}\right)\bigotimes_{j\in S}Z_j \nonumber \\
    R_{X}(\phi)&=\cos\left(\frac{\phi}{2}\right) I -i\sin\left(\frac{\phi}{2}\right)X
\end{align}
this also means that 
\begin{align}
    R_{X_k}(\phi)\left(\bigotimes_{j\in S} CZ_{jk} \right) &= \left( \cos\left(\frac{\phi}{2}\right)I-i\sin\left(\frac{\phi}{2} \right)X_k\right)\left(\bigotimes_{j\in S} CZ_{jk} \right) \nonumber \\
    &=  \left(\bigotimes_{j\in S} CZ_{jk} \right)\left(\cos\left(\frac{\phi}{2}\right)I-i\sin\left(\frac{\phi}{2} \right) \left(\bigotimes_{i\in S} Z_{i}\right) X_k\right)
\end{align}
The factor in brackets on the right has almost become an $R_{Z,S}$ rotation, with the obstruction being the $X_k$ gate. 
If we choose to let this act on the $\ket{+}_{k}$, then we obtain
\begin{align}\label{eq:RXandRZ}
    R_{X_k}(\phi)\left(\bigotimes_{j\in S} CZ_{jk} \right)\ket{+}_k = \left(\bigotimes_{j\in S} CZ_{jk} \right)R_{Z,S}(\phi)\ket{+}_k
\end{align}
This identity will allow us to apply controlled multi-qubit $Z$ rotations using controlled single qubit $X$ rotations. 


\begin{figure}
\centering
\scalebox{0.8}{
\begin{tikzpicture}[on grid,node distance=10mm]
    \begin{pgfonlayer}{bg2}
        \foreach \x in {1,2,3,4,5} 
        \node[name=i-\x] at (0,\x mm) {};
        \foreach \x in {1,2,3,4,5} 
        \node[name=o-\x] at (7cm,\x mm) {};
        \node[name=i-6] at (0,-2cm) {};
        \node[name=o-6,mymeasureselect] at (7cm,-2cm) {$a$};
        \foreach \x in {1,2,3,4,5,6} 
        \path (i-\x) edge (o-\x);
        \path (o-1) ++ (1cm,-1cm) node (eq) {$=$};
        \foreach \x in {1,2,3,4,5} 
        \node[name=Ri-\x] at (10,\x mm) {};
        \foreach \x in {1,2,3,4,5} 
        \node[name=Ro-\x] at (16.5cm,\x mm) {};
        \node[name=Ri-6] at (10,-2cm) {};
        \node[name=Ro-6,mymeasureselect] at (11cm,-2cm) {$a$};
        \foreach \x in {1,2,3,4,5,6} 
        \path (Ri-\x) edge (Ro-\x);
    \end{pgfonlayer}
    \begin{pgfonlayer}{bg2}
        \foreach \x in {1,2} 
        \node[name=c-\x,mycontrol] at (\x cm,\x mm) {};
        \foreach \x in {1,2} 
        \node[name=cb-\x,mycontrol] at (\x cm,-2cm) {};
        \node[name=c-4,mycontrol] at (3cm,4mm) {};
        \node[name=cb-4,mycontrol] at (3cm,-2cm) {};
        \foreach \x in {1,2,4} 
        \path (c-\x) edge (cb-\x) {};
    \end{pgfonlayer}
    \begin{pgfonlayer}{front}
        \node[left=3mm of i-6.west,anchor=east,xshift=6mm]  {$\ket{+}_{B_k}$};
        \path (i-6) ++(5cm,0cm) node[mySQG] (g6) {$R^{f(x,y)}_X(\theta)$};
        \node[myMQG=(Ri-1)(Ri-5),xshift=1.5cm,minimum width=2cm] (Z) {};
        \node[below=1mm of Z.base,anchor=base] (Zl) {$\left(\underset{j\in S}{\bigotimes}Z_{A_j}\right)^a$};
        \node[myMQG=(Ri-1)(Ri-5),minimum width=2.4cm,xshift=4.5cm] (R) {};
        \node[below=0mm of R.base,anchor=base] (Zl) {$(R_{Z,S}(\theta))^{f(x,y)}$};
        \node[left=3mm of Ri-6.west,anchor=east,xshift=6mm]  {$\ket{+}_{B_k}$};
    \end{pgfonlayer}
    \begin{pgfonlayer}{bg3}
        \node[fit=(c-1)(c-4)(cb-4),inner sep=2mm,densely dashed,rounded corners,draw=black,fill=red,fill opacity=0.1,minimum width=1cm] (C) {};
        \node[below=0mm of C.south,font=\small,anchor=north] (Cl) {$\underset{j\in S}{\prod}CZ_{A_jB_k}$};
        \node[fit=(g6)(o-6),inner sep=2mm,densely dashed,rounded corners,draw=black,fill=blue,fill opacity=0.1,minimum width=1cm] (Or) {};
        \node[below=0mm of Or.south,font=\small,anchor=north] (Orl) {$f$-measure $(I,R_X(-\theta))$};
    \end{pgfonlayer}
\end{tikzpicture}}
\caption{Circuit identity we use to apply multi-qubit $Z$ rotations. This identity follows from equation \eqref{eq:RXandRZ} and circuit identity 2. The top wire carries $n$ qubits.}
\label{fig:RZtrick}
\end{figure}
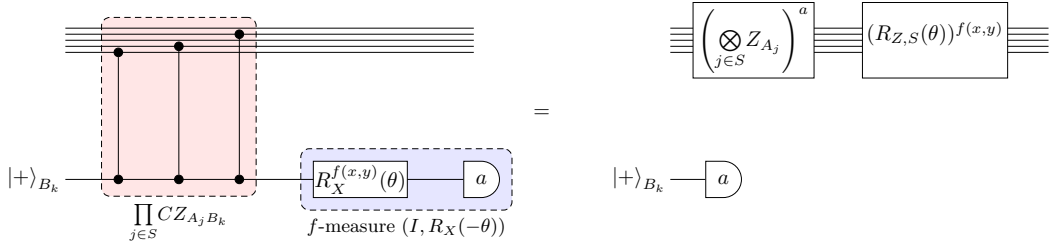

The full multi-qubit $Z$ rotation gadget is shown in \cref{fig:RZtrick}. 
The top wire denotes $n$ qubits, to which we want to apply a classically controlled multi-qubit Z rotation.
The final wire is an ancilla, prepared in the $\ket{+}$ state. 
The strategy is to apply $\prod_{j\in S} CZ_{A_jB_k}$ with $B_k$ taken to be the ancilla, then insert $B_k$ into an $f$-measure$(I,R_X(\phi))$ oracle. 
This will apply $R_X^{f(x,y)}(\phi)$, which by equation \eqref{eq:RXandRZ} is equivalent to $R_{Z,S}(\phi)$ as needed.
We also measure the $B_k$ qubit, which (using circuit identity 2) has the effect of inserting additional $Z$ operations on the target qubits. 
These additional $Z$ operations will remain at the end of our protocol and must be accounted for when we later make use of these gadgets. 

\subsection{Protocol for controlled diagonal unitaries}

Our goal in this section is to show that we can do $f$-unitary$(I,U)$ for any $U$ of the form $U=C_1DC_0$, where $C_0,C_1$ are Clifford, by using $O(1)$ $f$-measure$(I,H)$ oracles.

We will first need the following simple statement about diagonal unitaries. 

\begin{restatable}{lemma}{fourierdecomposition}\label{lemma:Fourierdecomposition}
    Any diagonal unitary $D$ can be written in the form
    \begin{align}
    D = \prod_S R_{Z,S}(\phi_S)
\end{align}
where $S$ ranges over all of the subsets of the $n$ qubits.
\end{restatable}
This is proven in appendix \ref{sec:lemmaproofs}.

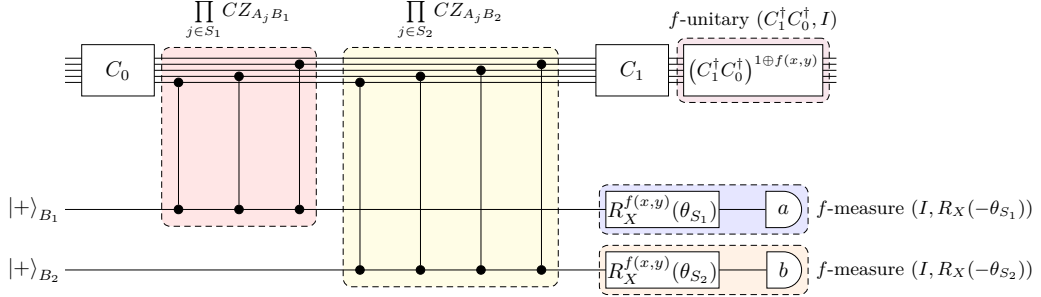
\begin{figure}
\centering
\scalebox{0.8}{
\begin{tikzpicture}[on grid,node distance=10mm]
    \begin{pgfonlayer}{bg2}
        \foreach \x in {1,2,3,4,5} 
        \node[name=i-\x] at (0,\x mm) {};
        \foreach \x in {1,2,3,4,5} 
        \node[name=o-\x] at (13cm,\x mm) {};
        \node[name=i-6] at (0,-2cm) {};
        \node[name=o-6,mymeasureselect] at (12cm,-2cm) {$a$};
        \node[name=i-7] at (0,-3cm) {};
        \node[name=o-7,mymeasureselect] at (12cm,-3cm) {$b$};
        \foreach \x in {1,2,3,4,5,6,7} 
        \path (i-\x) edge (o-\x);
    \end{pgfonlayer}
    \begin{pgfonlayer}{bg2}
        \foreach \x in {1,2} 
        \path (1cm,0) ++(\x cm,\x mm) node[name=c-\x,mycontrol] {};
        \foreach \x in {1,2} 
        \path (1cm,-2cm) ++(\x cm,0) node[name=cb-\x,mycontrol] {};
        \path (4cm,4mm) node[name=c-4,mycontrol] {};
        \path (4cm,-2cm) node[name=cb-4,mycontrol] {};
        \foreach \x in {1,2,4} 
        \path (c-\x) edge (cb-\x) {};
        \foreach \x in {7,8,9,10} 
        \path (-2cm,-6mm) ++(\x cm,\x mm) node[name=c-\x,mycontrol] {};
        \foreach \x in {7,8,9,10} 
        \path (-2cm,-3cm) ++(\x cm,0) node[name=cb-\x,mycontrol] {};
        \foreach \x in {7,8,9,10}  
        \path (c-\x) edge (cb-\x) {};
    \end{pgfonlayer}
    \begin{pgfonlayer}{front}
        \node[myMQG=(i-1)(i-5),xshift=1cm,inner ysep=1mm] (Cliff0) {};
        \node[above=-1mm of Cliff0.base,anchor=base](Cliff0l){$C_0$};
        \node[myMQG=(i-1)(i-5),xshift=9.5cm,inner ysep=1mm] (Cliff1) {};
        \node[above=-1mm of Cliff1.base,anchor=base](Cliff1l){$C_1$};
        \node[myMQG=(i-1)(i-5),xshift=11.5cm,inner ysep=1mm,minimum width=2.3cm] (fCliff) {};
        \node[above=-1mm of fCliff.base,anchor=base,font=\small](Cliff1l){$
        \left(C^\dagger_0C_1^\dagger\right)^{1\oplus f(x,y)}$};
        \node[left=3mm of i-6.west,anchor=east,xshift=6mm]  {$\ket{+}_{B_1}$};
        \path (i-6) ++(10cm,0cm) node[mySQG] (g6) {$R^{f(x,y)}_X(\theta_{S_1})$};
        \node[left=3mm of i-7.west,anchor=east,xshift=6mm]  {$\ket{+}_{B_2}$};
        \path (i-7) ++(10cm,0cm) node[mySQG] (g7) {$R^{f(x,y)}_X(\theta_{S_2})$};
    \end{pgfonlayer}
    \begin{pgfonlayer}{bg3}
        \node[fit=(c-1)(c-4)(cb-4),inner sep=2mm,densely dashed,rounded corners,draw=black,fill=red,fill opacity=0.1,minimum width=1cm] (C1) {};
        \node[above=0mm of C1.north,font=\small,anchor=south] (C1l) {$\underset{j\in S_1}{\bigotimes}CZ_{A_jB_1}$};
        \node[fit=(c-7)(c-10)(cb-10),inner sep=2mm,densely dashed,rounded corners,draw=black,fill=yellow,fill opacity=0.1,minimum width=1cm] (C2) {};
        \node[above=0mm of C2.north,font=\small,anchor=south] (C2l) {$\underset{j\in S_2}{\bigotimes}CZ_{A_jB_2}$};
        \node[fit=(fCliff),inner sep=1mm,densely dashed,rounded corners,draw=black,fill=purple,fill opacity=0.1,minimum width=1cm] (C3) {};
        \node[above=0mm of C3.north,font=\small,anchor=south] (C3l) {$f$-unitary $(C_0^\dagger C_1^\dagger,\mathbb{1})$};
        \node[fit=(g6)(o-6),inner sep=1mm,densely dashed,rounded corners,draw=black,fill=blue,fill opacity=0.1,minimum width=1cm] (Or1) {};
        \node[right=0mm of Or1.east,font=\small,anchor=west] (Or1l) {$f$-measure $(\mathbb{1},R_X(-\theta_{S_1}))$};
        \node[fit=(g7)(o-7),inner sep=1mm,densely dashed,rounded corners,draw=black,fill=orange,fill opacity=0.1,minimum width=1cm] (Or2) {};
        \node[right=0mm of Or2.east,font=\small,anchor=west] (Or2l) {$f$-measure $(\mathbb{1},R_X(-\theta_{S_2}))$};
    \end{pgfonlayer}
\end{tikzpicture}}
\caption{Circuit showing the reduction from $f$-unitary$(I,C_1DC_0)$ to oracle implementations of $f$-measure$(I,R_X(\theta))$. We defined $CZ_{A[S_i]B_k}=\otimes_{j\in S_i} CZ_{A_jB_k}$. We show a case where $D$ is the product of two multi-qubit $Z$ rotations for simplicity; the generalization to many rotations is straightforward.}
\label{fig:fcontrolledUprotocol}
\end{figure}

\beyondClifford*

\begin{proof} We give the protocol and then the proof of correctness below. 

\vspace{0.2cm}
\noindent \textbf{Protocol:} We assume Alice holds all of the input qubits initially, as usual. 
The circuit description of the protocol is shown in \cref{fig:fcontrolledUprotocol}.
Alice first applies $C_0$. 
Then, we consider a decomposition of $D$ into multi-qubit $Z$ rotations, as in \cref{lemma:Fourierdecomposition},
\begin{align}
    D = \prod_S R_{Z,S}(\phi_S) \,.
\end{align}
There are at most $2^{n_Q}$ terms in this decomposition, which, since we take $n_Q$ to be $O(1)$, is still $O(1)$. 
Alice and Bob apply a $Z$-rotation injection gadget (see \cref{fig:RZtrick}) for each $R_{Z,S}(\phi_S)$ appearing in this decomposition, using a fresh ancilla qubit for each rotation. 
Then, Alice applies $C_1$. 
Finally, Alice and Bob execute $f$-unitary$(C_0^\dagger C_1^\dagger,I)$. 
After the communication round, Alice and Bob learn the measurement outcomes from the $Z$ rotation gadgets. 
They use these measurement outcomes, which we will call $b_1,b_2,\dots$ where each $b_i$ corresponds to the measurements of the ancilla used for each multi-qubit $Z$-rotation, to determine $P_{f(x,y)}[b_1,b_2,\dots]= C_{f(x,y)} Z^{g(b_1,b_2,\dots)} C_{f(x,y)}^\dagger$ with and apply the inverse of these Pauli operators to their locally held qubits. Note that the function $g(...)$ determining where the Pauli $Z$ corrections appear depends on the set $S$. 

Note that the $f$-measure$(I,R_X(\theta))$ oracles can be $O(1)$ reduced to $f$-measure$(I,H)$ by \cref{thm:fmeasureEquivalence}, and $f$-unitary$(C_1^\dagger C_0^\dagger,I)$ can be reduced to $f$-measure$(I,H)$ by \cref{thm:fcliffordunitary}, \cref{thm:cliffordmeasure} and \cref{lem:fBelltofmeasure}.

\vspace{0.2cm}
\noindent \textbf{Correctness:} First consider what happens when $f=0$. 
Then the circuit shown in \cref{fig:fcontrolledUprotocol} simplifies to 

\begin{center}
\hspace*{-1.8cm}
\begin{tikzpicture}[on grid,node distance=10mm]
    \begin{pgfonlayer}{bg2}
        \foreach \x in {1,2,3,4,5} 
        \node[name=i-\x] at (0,\x mm) {};
        \foreach \x in {1,2,3,4,5} 
        \node[name=o-\x] at (13.5cm,\x mm) {};
        \node[name=i-6] at (0,-2cm) {};
        \node[name=o-6,mymeasureselect] at (9cm,-2cm) {$a$};
        \node[name=i-7] at (0,-3cm) {};
        \node[name=o-7,mymeasureselect] at (9cm,-3cm) {$b$};
        \foreach \x in {1,2,3,4,5,6,7} 
        \path (i-\x) edge (o-\x);
    \end{pgfonlayer}
    \begin{pgfonlayer}{bg2}
        \foreach \x in {1,2} 
        \path (1cm,0) ++(\x cm,\x mm) node[name=c-\x,mycontrol] {};
        \foreach \x in {1,2} 
        \path (1cm,-2cm) ++(\x cm,0) node[name=cb-\x,mycontrol] {};
        \path (4cm,4mm) node[name=c-4,mycontrol] {};
        \path (4cm,-2cm) node[name=cb-4,mycontrol] {};
        \foreach \x in {1,2,4} 
        \path (c-\x) edge (cb-\x) {};
        \foreach \x in {7,8,9,10} 
        \path (-2cm,-6mm) ++(\x cm,\x mm) node[name=c-\x,mycontrol] {};
        \foreach \x in {7,8,9,10} 
        \path (-2cm,-3cm) ++(\x cm,0) node[name=cb-\x,mycontrol] {};
        \foreach \x in {7,8,9,10}  
        \path (c-\x) edge (cb-\x) {};
    \end{pgfonlayer}
    \begin{pgfonlayer}{front}
        \node[myMQG=(i-1)(i-5),xshift=1cm,inner ysep=1mm] (Cliff0) {};
        \node[above=-1mm of Cliff0.base,anchor=base](Cliff0l){$C_0$};
        \node[myMQG=(i-1)(i-5),xshift=9.5cm,inner ysep=1mm] (Cliff1) {};
        \node[above=-1mm of Cliff1.base,anchor=base](Cliff1l){$C_1$};
        \node[myMQG=(i-1)(i-5),xshift=11cm,inner ysep=1mm] (fCliff1) {};
        \node[above=-1mm of fCliff1.base,anchor=base,font=\small](fCliff1l){$C_1^\dagger$};
        \node[myMQG=(i-1)(i-5),xshift=12.5cm,inner ysep=1mm] (fCliff2) {};
        \node[above=-1mm of fCliff2.base,anchor=base,font=\small](fCliff2l){$C^\dagger_0$};
        \node[left=3mm of i-6.west,anchor=east,xshift=6mm]  {$\ket{+}_{B_1}$};
        \node[left=3mm of i-7.west,anchor=east,xshift=6mm]  {$\ket{+}_{B_2}$};
    \end{pgfonlayer}
    \begin{pgfonlayer}{bg3}
        \node[fit=(c-1)(c-4)(cb-4),inner sep=2mm,densely dashed,rounded corners,draw=black,fill=red,fill opacity=0.1,minimum width=1cm] (C1) {};
        \node[above=0mm of C1.north,font=\small,anchor=south] (C1l) {$\underset{j\in S_1}{\prod}CZ_{A_jB_1}$};
        \node[fit=(c-7)(c-10)(cb-10),inner sep=2mm,densely dashed,rounded corners,draw=black,fill=yellow,fill opacity=0.1,minimum width=1cm] (C2) {};
        \node[above=0mm of C2.north,font=\small,anchor=south] (C2l) {$\underset{j\in S_2}{\prod}CZ_{A_jB_2}$};
        \node[fit=(fCliff1)(fCliff2),inner sep=1mm,densely dashed,rounded corners,draw=black,fill=purple,fill opacity=0.1,minimum width=1cm] (C3) {};
        \node[above=0mm of C3.north,font=\small,anchor=south] (C3l) {$C_0^\dagger C_1^\dagger$};
    \end{pgfonlayer}
\end{tikzpicture}
\end{center}

\noindent Using circuit identity 2, this leads to 

\begin{center}
\scalebox{1}{
\begin{tikzpicture}[on grid,node distance=10mm]
    \begin{pgfonlayer}{bg2}
        \foreach \x in {1,2,3,4,5} 
        \node[name=i-\x] at (0,\x mm) {};
        \foreach \x in {1,2,3,4,5} 
        \node[name=o-\x] at (9.5cm,\x mm) {};
        \node[name=i-6] at (0,-1cm) {};
        \node[name=o-6,mymeasureselect] at (1cm,-1cm) {$a$};
        \node[name=i-7] at (0,-2cm) {};
        \node[name=o-7,mymeasureselect] at (1cm,-2cm) {$b$};
        \foreach \x in {1,2,3,4,5,6,7} 
        \path (i-\x) edge (o-\x);
    \end{pgfonlayer}
    \begin{pgfonlayer}{front}
        \node[myMQG=(i-1)(i-5),xshift=1cm,inner ysep=1mm] (Cliff0) {};
        \node[above=-1mm of Cliff0.base,anchor=base](Cliff0l){$C_0$};
        \node[myMQG=(i-1)(i-5),xshift=2.5cm,inner ysep=1mm] (Z1) {};
        \node[above=-1mm of Z1.base,anchor=base](Z1l){$Z^{a}_{A[S_1]}$};
        \node[myMQG=(i-1)(i-5),xshift=4cm,inner ysep=1mm] (Z1) {};
        \node[above=-1mm of Z1.base,anchor=base](Z1l){$Z^{b}_{A[S_2]}$};
        \node[myMQG=(i-1)(i-5),xshift=5.5cm,inner ysep=1mm] (Cliff1) {};
        \node[above=-1mm of Cliff1.base,anchor=base](Cliff1l){$C_1$};
        \node[myMQG=(i-1)(i-5),xshift=7cm,inner ysep=1mm] (fCliff1) {};
        \node[above=-1mm of fCliff1.base,anchor=base,font=\small](fCliff1l){$C_1^\dagger$};
        \node[myMQG=(i-1)(i-5),xshift=8.5cm,inner ysep=1mm] (fCliff2) {};
        \node[above=-1mm of fCliff2.base,anchor=base,font=\small](fCliff2l){$C^\dagger_0$};
        \node[left=3mm of i-6.west,anchor=east,xshift=6mm]  {$\ket{+}_{B_1}$};
        \node[left=3mm of i-7.west,anchor=east,xshift=6mm]  {$\ket{+}_{B_2}$};
    \end{pgfonlayer}
    \begin{pgfonlayer}{bg3}
        \node[fit=(fCliff1)(fCliff2),inner sep=1mm,densely dashed,rounded corners,draw=black,fill=purple,fill opacity=0.1,minimum width=1cm] (C3) {};
        \node[above=0mm of C3.north,font=\small,anchor=south] (C3l) {$C_0^\dagger C_1^\dagger$};
    \end{pgfonlayer}
\end{tikzpicture}}
\end{center}

\noindent This simplifies to only a Pauli string acting on the $A$ qubits. 
The final step in the protocol is for Alice and Bob to, having learned the measurement outcomes $a,b\in\{0,1\}^{n}$, undo these Pauli operators, so that overall the circuit implements the identity. 

Now consider the $f=1$ case. 
Then, fixing $f=1$, the circuit in \cref{fig:fcontrolledUprotocol} simplifies to 

\begin{center}
\hspace*{-1.5cm}
\begin{tikzpicture}[on grid,node distance=10mm]
    \begin{pgfonlayer}{bg2}
        \foreach \x in {1,2,3,4,5} 
        \node[name=i-\x] at (0,\x mm) {};
        \foreach \x in {1,2,3,4,5} 
        \node[name=o-\x] at (11cm,\x mm) {};
        \node[name=i-6] at (0,-2cm) {};
        \node[name=o-6,mymeasureselect] at (11cm,-2cm) {$a$};
        \node[name=i-7] at (0,-3cm) {};
        \node[name=o-7,mymeasureselect] at (11cm,-3cm) {$b$};
        \foreach \x in {1,2,3,4,5,6,7} 
        \path (i-\x) edge (o-\x);
    \end{pgfonlayer}
    \begin{pgfonlayer}{bg2}
        \foreach \x in {1,2} 
        \path (1cm,0) ++(\x cm,\x mm) node[name=c-\x,mycontrol] {};
        \foreach \x in {1,2} 
        \path (1cm,-2cm) ++(\x cm,0) node[name=cb-\x,mycontrol] {};
        \path (4cm,4mm) node[name=c-4,mycontrol] {};
        \path (4cm,-2cm) node[name=cb-4,mycontrol] {};
        \foreach \x in {1,2,4} 
        \path (c-\x) edge (cb-\x) {};
        \foreach \x in {7,8,9,10} 
        \path (-2cm,-6mm) ++(\x cm,\x mm) node[name=c-\x,mycontrol] {};
        \foreach \x in {7,8,9,10} 
        \path (-2cm,-3cm) ++(\x cm,0) node[name=cb-\x,mycontrol] {};
        \foreach \x in {7,8,9,10}  
        \path (c-\x) edge (cb-\x) {};
    \end{pgfonlayer}
    \begin{pgfonlayer}{front}
        \node[myMQG=(i-1)(i-5),xshift=1cm,inner ysep=1mm] (Cliff0) {};
        \node[above=-1mm of Cliff0.base,anchor=base](Cliff0l){$C_0$};
        \node[myMQG=(i-1)(i-5),xshift=9.5cm,inner ysep=1mm] (Cliff1) {};
        \node[above=-1mm of Cliff1.base,anchor=base](Cliff1l){$C_1$};
        \node[left=3mm of i-6.west,anchor=east,xshift=6mm]  {$\ket{+}_{B_1}$};
        \path (i-6) ++(9.5cm,0cm) node[mySQG] (g6) {$R_X(\theta_{S_1})$};
        \node[left=3mm of i-7.west,anchor=east,xshift=6mm]  {$\ket{+}_{B_2}$};
        \path (i-7) ++(9.5cm,0cm) node[mySQG] (g7) {$R_X(\theta_{S_2})$};
    \end{pgfonlayer}
    \begin{pgfonlayer}{bg3}
        \node[fit=(c-1)(c-4)(cb-4),inner sep=2mm,densely dashed,rounded corners,draw=black,fill=red,fill opacity=0.1,minimum width=1cm] (C1) {};
        \node[above=0mm of C1.north,font=\small,anchor=south] (C1l) {$\underset{j\in S_1}{\prod}CZ_{A_jB_1}$};
        \node[fit=(c-7)(c-10)(cb-10),inner sep=2mm,densely dashed,rounded corners,draw=black,fill=yellow,fill opacity=0.1,minimum width=1cm] (C2) {};
        \node[above=0mm of C2.north,font=\small,anchor=south] (C2l) {$\underset{j\in S_2}{\prod}CZ_{A_jB_2}$};
    \end{pgfonlayer}
\end{tikzpicture}
\end{center}

\noindent Now we use \eqref{eq:RXandRZ} to move the $R_X$ rotations through the CZ components, picking up multi-qubit $Z$ rotations on the $A$ qubits, 

\begin{center}
\hspace*{-1.5cm}
\scalebox{0.9}{
\begin{tikzpicture}[on grid,node distance=10mm]
    \begin{pgfonlayer}{bg2}
        \foreach \x in {1,2,3,4,5} 
        \node[name=i-\x] at (0,\x mm) {};
        \foreach \x in {1,2,3,4,5} 
        \node[name=o-\x] at (15.5cm,\x mm) {};
        \node[name=i-6] at (0,-2cm) {};
        \node[name=o-6,mymeasureselect] at (9.5cm,-2cm) {$a$};
        \node[name=i-7] at (0,-3cm) {};
        \node[name=o-7,mymeasureselect] at (9.5cm,-3cm) {$b$};
        \foreach \x in {1,2,3,4,5,6,7} 
        \path (i-\x) edge (o-\x);
    \end{pgfonlayer}
    \begin{pgfonlayer}{bg2}
        \foreach \x in {1,2} 
        \path (1.5cm,0) ++(\x cm,\x mm) node[name=c-\x,mycontrol] {};
        \foreach \x in {1,2} 
        \path (1.5cm,-2cm) ++(\x cm,0) node[name=cb-\x,mycontrol] {};
        \path (4.5cm,4mm) node[name=c-4,mycontrol] {};
        \path (4.5cm,-2cm) node[name=cb-4,mycontrol] {};
        \foreach \x in {1,2,4} 
        \path (c-\x) edge (cb-\x) {};
        \foreach \x in {7,8,9,10} 
        \path (-1.5cm,-6mm) ++(\x cm,\x mm) node[name=c-\x,mycontrol] {};
        \foreach \x in {7,8,9,10} 
        \path (-1.5cm,-3cm) ++(\x cm,0) node[name=cb-\x,mycontrol] {};
        \foreach \x in {7,8,9,10}  
        \path (c-\x) edge (cb-\x) {};
    \end{pgfonlayer}
    \begin{pgfonlayer}{front}
        \node[myMQG=(i-1)(i-5),xshift=1cm,inner ysep=1mm] (Cliff0) {};
        \node[above=-1mm of Cliff0.base,anchor=base](Cliff0l){$C_0$};
        \node[myMQG=(i-1)(i-5),xshift=14.5cm,inner ysep=1mm] (Cliff1) {};
        \node[above=-1mm of Cliff1.base,anchor=base](Cliff1l){$C_1$};
        \node[left=3mm of i-6.west,anchor=east,xshift=6mm]  {$\ket{+}_{B_1}$};
        \path (i-6) ++(1cm,0cm) node[mySQG] (g6) {$R_X(\theta_{S_1})$};
        \node[left=3mm of i-7.west,anchor=east,xshift=6mm]  {$\ket{+}_{B_2}$};
        \path (i-7) ++(1cm,0cm) node[mySQG] (g7) {$R_X(\theta_{S_2})$};
        \node[myMQG=(i-1)(i-5),xshift=10cm,inner ysep=1mm,minimum width=2cm] (Rot1) {};
        \node[above=-1mm of Rot1.base,anchor=base](Rot1l){$R_{Z,S_1}(\theta_{S_1})$};
        \node[myMQG=(i-1)(i-5),xshift=12.5cm,inner ysep=1mm,minimum width=2cm] (Rot2) {};
        \node[above=-1mm of Rot2.base,anchor=base](Rot1l){$R_{Z,S_2}(\theta_{S_2})$};
    \end{pgfonlayer}
    \begin{pgfonlayer}{bg3}
        \node[fit=(c-1)(c-4)(cb-4),inner sep=2mm,densely dashed,rounded corners,draw=black,fill=red,fill opacity=0.1,minimum width=1cm] (C1) {};
        \node[above=0mm of C1.north,font=\small,anchor=south] (C1l) {$\underset{j\in S_1}{\prod}CZ_{A_jB_1}$};
        \node[fit=(c-7)(c-10)(cb-10),inner sep=2mm,densely dashed,rounded corners,draw=black,fill=yellow,fill opacity=0.1,minimum width=1cm] (C2) {};
        \node[above=0mm of C2.north,font=\small,anchor=south] (C2l) {$\underset{j\in S_2}{\prod}CZ_{A_jB_2}$};
    \end{pgfonlayer}
\end{tikzpicture}}
\end{center}

\noindent Now we use circuit identity 2 again to simplify, obtaining $Z$ gates in place of the control $Z$ components, 

\begin{center}
\begin{tikzpicture}[on grid,node distance=10mm]
    \begin{pgfonlayer}{bg2}
        \foreach \x in {1,2,3,4,5} 
        \node[name=i-\x] at (0,\x mm) {};
        \foreach \x in {1,2,3,4,5} 
        \node[name=o-\x] at (11.5cm,\x mm) {};
        \node[name=i-6] at (0,-1cm) {};
        \node[name=o-6,mymeasureselect] at (2.5cm,-1cm) {$a$};
        \node[name=i-7] at (0,-2cm) {};
        \node[name=o-7,mymeasureselect] at (2.5cm,-2cm) {$b$};
        \foreach \x in {1,2,3,4,5,6,7} 
        \path (i-\x) edge (o-\x);
    \end{pgfonlayer}
    \begin{pgfonlayer}{front}
        \node[myMQG=(i-1)(i-5),xshift=1cm,inner ysep=1mm] (Cliff0) {};
        \node[above=-1mm of Cliff0.base,anchor=base](Cliff0l){$C_0$};
        \node[myMQG=(i-1)(i-5),xshift=10.5cm,inner ysep=1mm] (Cliff1) {};
        \node[above=-1mm of Cliff1.base,anchor=base](Cliff1l){$C_1$};
        \node[left=3mm of i-6.west,anchor=east,xshift=6mm]  {$\ket{+}_{B_1}$};
        \path (i-6) ++(1cm,0cm) node[mySQG] (g6) {$R_X(\theta_{S_1})$};
        \node[left=3mm of i-7.west,anchor=east,xshift=6mm]  {$\ket{+}_{B_2}$};
        \path (i-7) ++(1cm,0cm) node[mySQG] (g7) {$R_X(\theta_{S_2})$};
        \node[myMQG=(i-1)(i-5),xshift=2.5cm,inner ysep=1mm,minimum width=1.2cm] (Cor1) {};
        \node[above=-1mm of Cor1.base,anchor=base](Cor1l){$Z^a_{A[S_1]}$};
        \node[myMQG=(i-1)(i-5),xshift=4cm,inner ysep=1mm,minimum width=1.2cm] (Cor2) {};
        \node[above=-1mm of Cor2.base,anchor=base](Cor1l){$Z^b_{A[S_2]}$};
        \node[myMQG=(i-1)(i-5),xshift=6cm,inner ysep=1mm,minimum width=2cm] (Rot1) {};
        \node[above=-1mm of Rot1.base,anchor=base](Rot1l){$R_{Z,S_1}(\theta_{S_1})$};
        \node[myMQG=(i-1)(i-5),xshift=8.5cm,inner ysep=1mm,minimum width=2cm] (Rot2) {};
        \node[above=-1mm of Rot2.base,anchor=base](Rot1l){$R_{Z,S_2}(\theta_{S_2})$};
    \end{pgfonlayer}
\end{tikzpicture}
\end{center}

\noindent We see that the circuit acting on the $A$ qubits is as needed, up to Pauli corrections. 
These are determined by the $a,b$ measurement outcomes, which are known to Alice and Bob in the second round. 
Thus they can correct the Pauli operators as needed on their locally held qubits, and therefore implement exactly $U=C_1DC_0$. 
\end{proof}

Note that variants of this protocol allow one to obtain reductions to various other classes of $f$-measure unitary protocols. 
For instance, we could use $f$-measure$(R_X(\theta_1), R_X(\theta_2))$ oracles in place of the current $f$-measure$(I, R_X(\theta_2))$. 
This would allow the implementation of two different diagonal unitaries $D_0, D_1$ controlled on the value of $f(x,y)$. 
Further, one could replace the $f$-unitary$(C_1^\dagger C_0^\dagger, I)$ oracle with anything of the form $f$-unitary$(C,C')$ with $C,C'$ Clifford to obtain some further cases. 

\section{AI Disclosure Statement}

ChatGPT was used to write code that helped find earlier versions of the circuits used in \cref{lemma:anglesplitting}.

\appendix

\section{Proofs of some lemmas}\label{sec:lemmaproofs}

\densitytoDN*

\begin{proof}
Let $\ket{\psi}_{RA}=\sum_{s}\sqrt{\lambda_s}\ket{s}_R\ket{s}_A$ be the state which achieves the maximum appearing in the definition of the diamond norm.
Then, 
\begin{align}
    \Vert \mathcal{N}-\mathcal{M}\Vert_\diamond &= \Vert \sum_{s,t} \sqrt{\lambda_s\lambda_t}\ketbra{s}{t}_R\otimes \mathcal{N}(\ketbra{s}{t})-\sum_{s,t} \sqrt{\lambda_s\lambda_t}\ketbra{s}{t}_R\otimes\mathcal{M}(\ketbra{s}{t})\Vert_1 \nonumber \\
    &\leq \sum_{s,t} \sqrt{\lambda_s\lambda_t}\Vert \ketbra{s}{t}\otimes\left( \mathcal{N}(\ketbra{s}{t})-\mathcal{M}(\ketbra{s}{t})\right)\Vert_1 \nonumber \\
    &\leq \sum_{s,t} \sqrt{\lambda_s\lambda_t}\Vert \ketbra{s}{t}\Vert_\infty \Vert \mathcal{N}(\ketbra{s}{t})-\mathcal{M}(\ketbra{s}{t}) \Vert_1 \nonumber \\
    &=\sum_{s,t} \sqrt{\lambda_s\lambda_t}\Vert \mathcal{N}(\ketbra{s}{t})-\mathcal{M}(\ketbra{s}{t}) \Vert_1.
\end{align}
Next, we express the operator $\ketbra{s}{t}$ as a sum of density matrices. 
Define
\begin{align}
    \ket{\phi_\pm} &= \frac{1}{\sqrt{2}} \left( \ket{s}\pm \ket{t}\right) \nonumber \\
    \ket{\tilde{\phi}_\pm} &= \frac{1}{\sqrt{2}} \left( \ket{s}\pm i\ket{t}\right)
\end{align}
so that
\begin{align}
    \ketbra{s}{t} = \frac{1}{2}\left(\ketbra{\phi_+}{\phi_+} - \ketbra{\phi_-}{\phi_-} + i\ketbra{\tilde{\phi}_+}{\tilde{\phi}_+}- i\ketbra{\tilde{\phi}_-}{\tilde{\phi}_-} \right).
\end{align}
Then, using this to continue our calculation of the diamond norm we can expand $\ketbra{s}{t}$ into the four terms above, use the triangle inequality, and then (absolute) linearity of the one-norm to obtain
\begin{align}
    \Vert \mathcal{N}(\ketbra{s}{t})-\mathcal{M}(\ketbra{s}{t}) \Vert_1 \leq 2\epsilon.
\end{align}
Using this in our upper bound of the diamond norm, we have
\begin{align}
    \Vert \mathcal{N}-\mathcal{M}\Vert_\diamond \leq 2\epsilon \left(\sum_s\sqrt{\lambda_s}\right)^2.
\end{align}
It remains to bound the sum of the square root of the Schmidt coefficients. 
Since $\sum_i\lambda_i=1$ and there are at most $d$ coefficients, we obtain that
\begin{align}
    \sum_i\sqrt{\lambda_i} \leq \sqrt{d}.
\end{align}
This follows for instance by considering the vector $x=(\sqrt{\lambda_1}, \dots ,\sqrt{\lambda_d})$ and using that $\Vert x\Vert_1\leq \sqrt{d}\Vert x\Vert_2$, $\sum_i \lambda_i=1$. 
Thus finally we obtain
\begin{align}
    \Vert \mathcal{N}-\mathcal{M}\Vert_\diamond \leq 2\epsilon d.
\end{align}
as claimed. 
\end{proof}

\fourierdecomposition*

\begin{proof}
Consider an arbitrary diagonal unitary, 
\begin{align}
    D=\sum_x e^{-i\theta(x)/2}\ketbra{x}{x}
\end{align}
where the $1/2$ is inserted for convenience later on. 
Any real valued function on $n$ bits can be expanded in its Fourier basis. 
We apply this to the function $\theta(x)$,
\begin{align}
    \theta(x)=\sum_S  (-1)^{\sum_{i\in S}x_i} \hat{\theta}_S = \sum_S  \chi_S(x) \hat{\theta}_S,
\end{align}
where the $\hat{\theta}_S$ are real numbers.
The second equality defines $\chi_S(x)$. 
Using this we can expand the exponential $e^{-(i/2)\theta(x)}$ as
\begin{align}
    e^{-(i/2)\theta(x)} = \prod_S e^{-(i/2)\chi_S(x) \hat{\theta}_S},
\end{align}
so that 
\begin{align}
    D=\sum_x \prod_{S} e^{-(i/2)\chi_S(x) \hat{\theta}_S} \ketbra{x}{x}.
\end{align}

Now consider $R_{Z,S}(\phi_S)$:
\begin{align}
    R_{Z,S}(\phi_S) &= \exp\left(-\frac{i\phi_S}{2} \bigotimes_{i\in S}Z_i\right) \nonumber \\
    &= \sum_ x\exp\left(-\frac{i\phi_S}{2} (-1)^{\sum_{i\in S}x_i}\right) \ketbra{x}{x}\nonumber \\
    &= \sum_x \exp\left(-\frac{i\phi_S}{2} \chi_S(x)\right) \ketbra{x}{x}
\end{align}
where in the second line we inserted the identity in the form $I=\sum_x\ketbra{x}{x}$. 
Now choose $\phi_S=\hat{\theta}_S$, and observe that
\begin{align}
    \prod_S R_{Z,S}(\phi_S) &= \prod_S \sum_x \exp\left(-\frac{i\hat{\theta}_S}{2} \chi_S(x)\right) \ketbra{x}{x} \nonumber \\
    &= \sum_x \prod_S \exp\left(-\frac{i\hat{\theta}_S}{2} \chi_S(x)\right) \ketbra{x}{x} \nonumber \\
    &= \sum_x \exp\left(-\frac{i}{2}\sum_S \hat{\theta}_S\chi_S(x)\right) \ketbra{x}{x} \nonumber \\
    &= \sum_x \exp\left(-\frac{i}{2}\theta(x)\right) \ketbra{x}{x}\nonumber \\
    &= D
\end{align}
as needed. 
\end{proof}

\section{Redistribution of quantum inputs}\label{sec:PTtrick}

In this section we explain why, up to $O(1)$ overheads, classically controlled $2\rightarrow 2$ tasks with different distributions of the quantum inputs are equivalent.
To see this, we give the argument in two steps:
First, we show that a task with inputs $Q,Q'$ starting on the left and right respectively can be reduced to one with $QQ'$ on one side. 
Second, we show that tasks with input $QQ'$ on the left can be reduced to one with $Q$ on the left and $Q'$ on the right. 
Combining these two reductions we can turn any input distribution into any other. 
Note that the reductions require a number of maximally entangled states as overhead which scales polynomially in the dimension of the input systems, so that this number is exponential in the number of qubits in $QQ'$. 
Since this is taken to be $O(1)$ in our context, the transformation is efficient, but in other contexts (where $n_Q, n_{Q'}$ are not constant) these reductions incur a large overhead. 

The reductions are based on the port-based teleportation primitive originally introduced in \cite{ishizaka2008asymptotic}, and applied to non-local quantum computation in \cite{beigi2011simplified}.
In the following we will denote the maximally entangled Bell state on qudit systems by
\begin{equation}
    \ket{\Psi^+_d}:=\frac{1}{\sqrt{d}}\sum_{j=1}^{d}\ket{j}\ket{j} \,.
\end{equation}

We first consider the case which we will call merging, where we wish to implement a specified $2\rightarrow 2$ task, call it $G$, which takes $QQ'$ as input, with $Q$ on the left and $Q'$ on the right.
We will see how to implement this using (some number of copies of) protocols that implement $G'$, which has the same outputs but accepts inputs $QQ'$ both on the left. 
 
To achieve this, we first need some additional resources shared between the left and right:
\begin{itemize}
    \item Share $\ket{\Psi^+_d}_{XY}$ with $X$ on the left, $Y$ on the right, and $d_X=d_Y=d_Q=d$.
    \item Share $\bigotimes_{i=1}^N \ket{\Psi^+_d}_{A_iB_i}$, where we specify $N$ later.
\end{itemize}
Then, we are ready to execute the following protocol. 
\begin{enumerate}
    \item Alice measures $QX$ in the Bell basis, obtaining measurement outcome $k$. The value of $k$ is sent to both players in the communication round.
    \item Bob measures $YB$ using the port-based teleportation POVM. Label the measurement outcome as $i^*$.
    \item Alice undoes the Pauli corresponding to the measurement outcome $k$ on every port $A_i$.
    \item Alice and Bob implement $N$ copies of $G'$, taking $A_i$ as input to the $i$th copy. As the classical inputs to these protocols they use copies of the classical information $x,y$ they receive. During the communication round, Alice sends $i^*$ to Bob.
    \item After the $G'$ protocols are completed, Alice and Bob trace out all outputs from the $G'$ instances except those from the $i^*$ copy. 
\end{enumerate}

Next consider the case where we wish to implement $G'$, which takes both inputs on the left, using copies of $G$, which takes input $Q$ on the left and input $Q'$ on the right. 
We call this procedure splitting.
To achieve this, we first share some extra quantum resources:
\begin{itemize}
    \item Share $\bigotimes_{i=1}^N \ket{\Psi^+}_{A_iB_i}$. 
\end{itemize}
Then, we are ready to execute the following protocol:
\begin{enumerate}
    \item Alice measures $Q'A$ using the port-teleportation POVM, obtaining measurement outcome $i^*$. The value of $i^*$ is sent to Bob in the communication round.
    \item Alice and Bob run $N$ instances of $G$; Bob uses $B_i$ as the input to instance $i$, Alice uses $Q$ as input to instance $i^*$, and inserts $\ket{0}$ states into the remaining copies. Again they use copies of the classical inputs $x,y$ as needed.
    \item After the $G$ protocols are completed, Alice and Bob trace out all outputs from the $G$ instances, except those from the $i^*$ copy. 
\end{enumerate}
The properties of these two reductions are given in \cref{table:Port based}.
They follow from \cref{remark:diamondadditive} and Corollary II.2 in \cite{beigi2011simplified}.
A small subtlety is that, even though we are using $N$ copies of an oracle of precision $\delta$ in parallel, the error does not accumulate to $N\delta$, due to the fact that Alice and Bob discard all of the outputs of the oracles, with the exception of the one corresponding to the outcomes of the port-based schemes.
Instead, only $\delta$ (not $N\delta$) is added to the error.

\begin{figure}
\centering
\makebox[\textwidth][c]{%
\begin{tabular}{|p{1.5cm}|p{10cm}|p{1.75cm}|}
     \hline 
     Task & Resources used & Precision\\
     \hline
     Merging & A single copy of $\ket{\Psi^+_{d_Q}}$, $N$ copies of $\ket{\Psi^+_{d_{QQ'}}}$, $N$ oracles with $\delta$ precision & $\delta+\frac{4d_{QQ'}^2}{\sqrt{N}}$
     \\
     \hline
     Splitting & $N$ copies of $\ket{\Psi^+_{d_{Q'}}}$, $N$ oracles with $\delta$ precision & $\delta+\frac{4d_{Q'}^2}{\sqrt{N}}$
     \\
     \hline
\end{tabular}
}
\caption{Properties of the reductions between non-local strategies of different (spatial) input distributions. $\ket{\Psi^+_d}$ denotes the usual Bell state on two qudit systems, and the precision denotes how well a strategy approximates the desired non-local task with respect to the diamond norm. $N>0$ denotes the number of ports used in the port-based teleportation scheme and can be chosen freely.}
\label{table:Port based}
\end{figure}

\bibliographystyle{unsrt}
\bibliography{biblio}

@article{bluhm2025complexity,
  title={A complexity theory for non-local quantum computation},
  author={Bluhm, Andreas and H{\"o}fer, Simon and May, Alex and Stasiuk, Mikka and Lunel, Philip Verduyn and Yuen, Henry},
  journal={Quantum},
  volume={10},
  pages={2152},
  year={2026},
  publisher={Verein zur F{\"o}rderung des Open Access Publizierens in den Quantenwissenschaften}
}

@article{vaidman2003,
  title = {Instantaneous Measurement of Nonlocal Variables},
  author = {Vaidman, Lev},
  journal = {Phys. Rev. Lett.},
  volume = {90},
  issue = {1},
  pages = {010402},
  numpages = {4},
  year = {2003},
  month = {Jan},
  publisher = {American Physical Society},
  doi = {10.1103/PhysRevLett.90.010402},
  url = {https://link.aps.org/doi/10.1103/PhysRevLett.90.010402}
}

@article{Zhou_2000,
   title={Methodology for quantum logic gate construction},
   volume={62},
   ISSN={1094-1622},
   url={http://dx.doi.org/10.1103/PhysRevA.62.052316},
   DOI={10.1103/physreva.62.052316},
   number={5},
   journal={Physical Review A},
   publisher={American Physical Society (APS)},
   author={Zhou, Xinlan and Leung, Debbie W. and Chuang, Isaac L.},
   year={2000},
   month=oct }

@article{allerstorfer2024relating,
  title={Relating non-local quantum computation to information theoretic cryptography},
  author={Allerstorfer, Rene and Buhrman, Harry and May, Alex and Speelman, Florian and Lunel, Philip Verduyn},
  journal={Quantum},
  volume={8},
  pages={1387},
  year={2024},
  publisher={Verein zur F{\"o}rderung des Open Access Publizierens in den Quantenwissenschaften}
}

@article{kent2011quantum,
  title={Quantum tagging: Authenticating location via quantum information and relativistic signaling constraints},
  author={Kent, Adrian and Munro, William J and Spiller, Timothy P},
  journal={Physical Review A—Atomic, Molecular, and Optical Physics},
  volume={84},
  number={1},
  pages={012326},
  year={2011},
  publisher={APS}
}

@inproceedings{buhrman2013garden,
  title={The garden-hose model},
  author={Buhrman, Harry and Fehr, Serge and Schaffner, Christian and Speelman, Florian},
  booktitle={Proceedings of the 4th Conference on Innovations in Theoretical Computer Science},
  pages={145--158},
  year={2013}
}

@article{buhrman2014position,
  title={Position-based quantum cryptography: Impossibility and constructions},
  author={Buhrman, Harry and Chandran, Nishanth and Fehr, Serge and Gelles, Ran and Goyal, Vipul and Ostrovsky, Rafail and Schaffner, Christian},
  journal={SIAM Journal on Computing},
  volume={43},
  number={1},
  pages={150--178},
  year={2014},
  publisher={SIAM}
}

@article{beigi2011simplified,
  title={Simplified instantaneous non-local quantum computation with applications to position-based cryptography},
  author={Beigi, Salman and K{\"o}nig, Robert},
  journal={New Journal of Physics},
  volume={13},
  number={9},
  pages={093036},
  year={2011}
}

@article{bluhm2022single,
  title={A single-qubit position verification protocol that is secure against multi-qubit attacks},
  author={Bluhm, Andreas and Christandl, Matthias and Speelman, Florian},
  journal={Nature Physics},
  volume={18},
  number={6},
  pages={623--626},
  year={2022},
  publisher={Nature Publishing Group UK London}
}

@article{asadi2024rank,
  title={Rank lower bounds on non-local quantum computation},
  author={Asadi, Vahid R and Culf, Eric and May, Alex},
  journal={arXiv preprint arXiv:2402.18647},
  year={2024}
}

@article{asadi2025linear,
  title={Linear gate bounds against natural functions for position-verification},
  author={Asadi, Vahid and Cleve, Richard and Culf, Eric and May, Alex},
  journal={Quantum},
  volume={9},
  pages={1604},
  year={2025},
  publisher={Verein zur F{\"o}rderung des Open Access Publizierens in den Quantenwissenschaften}
}

@article{cree2023code,
  title={Code-routing: a new attack on position verification},
  author={Cree, Joy and May, Alex},
  journal={Quantum},
  volume={7},
  pages={1079},
  year={2023},
  publisher={Verein zur F{\"o}rderung des Open Access Publizierens in den Quantenwissenschaften}
}

@inproceedings{gertner1998protecting,
  title={Protecting data privacy in private information retrieval schemes},
  author={Gertner, Yael and Ishai, Yuval and Kushilevitz, Eyal and Malkin, Tal},
  booktitle={Proceedings of the thirtieth annual ACM symposium on Theory of computing},
  pages={151--160},
  year={1998}
}

@inproceedings{gay2015communication,
  title={Communication complexity of conditional disclosure of secrets and attribute-based encryption},
  author={Gay, Romain and Kerenidis, Iordanis and Wee, Hoeteck},
  booktitle={Annual Cryptology Conference},
  pages={485--502},
  year={2015},
  organization={Springer}
}

@article{applebaum2021placing,
  title={Placing conditional disclosure of secrets in the communication complexity universe},
  author={Applebaum, Benny and Vasudevan, Prashant Nalini},
  journal={Journal of Cryptology},
  volume={34},
  number={2},
  pages={11},
  year={2021},
  publisher={Springer}
}

@inproceedings{feige1994minimal,
  title={A minimal model for secure computation},
  author={Feige, Uri and Killian, Joe and Naor, Moni},
  booktitle={Proceedings of the twenty-sixth annual ACM symposium on Theory of computing},
  pages={554--563},
  year={1994}
}

@article{olivo2020breaking,
  title={Breaking simple quantum position verification protocols with little entanglement},
  author={Olivo, Andrea and Chabaud, Ulysse and Chailloux, Andr{\'e} and Grosshans, Fr{\'e}d{\'e}ric},
  journal={arXiv preprint arXiv:2007.15808},
  year={2020}
}

@article{barenco1995elementary,
  title={Elementary gates for quantum computation},
  author={Barenco, Adriano and Bennett, Charles H and Cleve, Richard and DiVincenzo, David P and Margolus, Norman and Shor, Peter and Sleator, Tycho and Smolin, John A and Weinfurter, Harald},
  journal={Physical Review A},
  volume={52},
  number={5},
  pages={3457},
  year={1995},
  publisher={APS}
}

@article{hoeffding1963probability,
  title={Probability inequalities for sums of bounded random variables},
  author={Hoeffding, Wassily},
  journal={Journal of the American statistical association},
  volume={58},
  number={301},
  pages={13--30},
  year={1963},
  publisher={Taylor \& Francis}
}

@article{ishizaka2008asymptotic,
  title={Asymptotic teleportation scheme as a universal programmable quantum processor},
  author={Ishizaka, Satoshi and Hiroshima, Tohya},
  journal={Physical Review Letters},
  volume={101},
  number={24},
  pages={240501},
  year={2008},
  publisher={APS}
}

@inproceedings{kavuri2025device,
  title={Device-independent quantum position verification},
  author={Kavuri, Gautam A and Gookin, Abigail and Zhang, Yanbao and Bienfang, Joshua C and Fu, Honghao and Alnawakhtha, Yusuf and Patra, Soumyadip and Reddy, Dileep V and Mazurek, Michael D and Abell{\'a}n, Carlos and others},
  booktitle={Quantum 2.0},
  pages={QM3B--3},
  year={2025},
  organization={Optica Publishing Group}
}

@article{allerstorfer2025making,
  title={Making existing quantum position verification protocols secure against arbitrary transmission loss},
  author={Allerstorfer, Rene and Bluhm, Andreas and Buhrman, Harry and Christandl, Matthias and Escol{\`a}-Farr{\`a}s, Lloren{\c{c}} and Speelman, Florian and Verduyn Lunel, Philip},
  journal={Physical Review Letters},
  volume={135},
  number={26},
  pages={260801},
  year={2025},
  publisher={APS}
}

@article{escola2025quantum,
  title={Quantum position verification in one shot: parallel repetition of the $ f $-{BB84} and $ f $-routing protocols},
  author={Escol{\`a}-Farr{\`a}s, Lloren{\c{c}} and Speelman, Florian},
  journal={arXiv preprint arXiv:2503.09544},
  year={2025}
}

@article{browne2016one,
  title={One-Way Quantum Computation},
  author={Browne, Dan and Briegel, Hans},
  journal={Quantum information: From foundations to quantum technology applications},
  pages={449--473},
  year={2016},
  publisher={Wiley Online Library}
}

@inproceedings{unruh2014quantum,
  title={Quantum position verification in the random oracle model},
  author={Unruh, Dominique},
  booktitle={Annual Cryptology Conference},
  pages={1--18},
  year={2014},
  organization={Springer}
}

@article{lau2011insecurity,
  title={Insecurity of position-based quantum-cryptography protocols against entanglement attacks},
  author={Lau, Hoi-Kwan and Lo, Hoi-Kwong},
  journal={Physical Review A—Atomic, Molecular, and Optical Physics},
  volume={83},
  number={1},
  pages={012322},
  year={2011},
  publisher={APS}
}

@article{chakraborty2015practical,
  title={Practical position-based quantum cryptography},
  author={Chakraborty, Kaushik and Leverrier, Anthony},
  journal={Physical Review A},
  volume={92},
  number={5},
  pages={052304},
  year={2015},
  publisher={APS}
}

@article{may2026entanglement,
  title={Entanglement cost in non-local quantum computation},
  author={May, Alex},
  journal={arXiv preprint arXiv:2605.02840},
  year={2026}
}

@article{fan2026relativistic,
  title={Relativistic Position Verification with Coherent States},
  author={Fan-Yuan, Guan-Jie and Shan, Yang-Guang and Zhang, Cong and Wang, Yu-Long and Fan, Yu-Xuan and Xie, Wei-Xin and He, De-Yong and Wang, Shuang and Yin, Zhen-Qiang and Chen, Wei and others},
  journal={arXiv preprint arXiv:2602.01787},
  year={2026}
}

@article{kanneworff2025towards,
  title={Towards experimental demonstration of quantum position verification using single photons},
  author={Kanneworff, Kirsten and Poortvliet, Mio and Bouwmeester, Dirk and Allerstorfer, Rene and Lunel, Philip Verduyn and Speelman, Florian and Buhrman, Harry and Steindl, Petr and L{\"o}ffler, Wolfgang},
  journal={Quantum Science and Technology},
  volume={10},
  number={4},
  pages={045004},
  year={2025},
  publisher={IOP Publishing}
}

@inproceedings{ananth2024unclonable,
  title={Unclonable secret sharing},
  author={Ananth, Prabhanjan and Goyal, Vipul and Liu, Jiahui and Liu, Qipeng},
  booktitle={International Conference on the Theory and Application of Cryptology and Information Security},
  pages={129--157},
  year={2024},
  organization={Springer}
}

@article{yu2012fast,
  title={Fast protocols for local implementation of bipartite nonlocal unitaries},
  author={Yu, Li and Griffiths, Robert B and Cohen, Scott M},
  journal={Physical Review A—Atomic, Molecular, and Optical Physics},
  volume={85},
  number={1},
  pages={012304},
  year={2012},
  publisher={APS}
}

@inproceedings{girish2026magic,
  title={Magic and communication complexity},
  author={Girish, Uma and May, Alex and Parham, Natalie and Yuen, Henry},
  booktitle={Proceedings of the 58th Annual ACM Symposium on Theory of Computing},
  pages={845--856},
  year={2026}
}

@article{may2022connected,
  title={The connected wedge theorem and its consequences},
  author={May, Alex and Sorce, Jonathan and Yoshida, Beni},
  journal={Journal of High Energy Physics},
  volume={2022},
  number={11},
  pages={1--65},
  year={2022},
  publisher={Springer}
}

@article{may2022complexity,
  title={Complexity and entanglement in non-local computation and holography},
  author={May, Alex},
  journal={Quantum},
  volume={6},
  pages={864},
  year={2022},
  publisher={Verein zur F{\"o}rderung des Open Access Publizierens in den Quantenwissenschaften}
}

@article{may2019quantum,
  title={Quantum tasks in holography},
  author={May, Alex},
  journal={Journal of High Energy Physics},
  volume={2019},
  number={10},
  pages={1--39},
  year={2019},
  publisher={Springer}
}

@article{may2020holographic,
  title={Holographic scattering requires a connected entanglement wedge},
  author={May, Alex and Penington, Geoff and Sorce, Jonathan},
  journal={Journal of High Energy Physics},
  volume={2020},
  number={8},
  pages={1--34},
  year={2020},
  publisher={Springer}
}

@article{apel2024security,
  title={Security of quantum position-verification limits Hamiltonian simulation via holography},
  author={Apel, Harriet and Cubitt, Toby and Hayden, Patrick and Kohler, Tamara and P{\'e}rez-Garc{\'\i}a, David},
  journal={Journal of High Energy Physics},
  volume={2024},
  number={8},
  pages={1--40},
  year={2024},
  publisher={Springer}
}

@article{asadi2025conditional,
  title={Conditional disclosure of secrets with quantum resources},
  author={Asadi, Vahid R and Kuroiwa, Kohdai and Leung, Debbie and May, Alex and Pasterski, Sabrina and Waddell, Chris},
  journal={Quantum},
  volume={9},
  pages={1885},
  year={2025},
  publisher={Verein zur F{\"o}rderung des Open Access Publizierens in den Quantenwissenschaften}
}

@article{girish2026new,
  title={New bounds on private simultaneous quantum message passing},
  author={Girish, Uma and May, Alex and Parham, Natalie and Yuen, Henry},
  journal={arXiv preprint arXiv:2606.12557},
  year={2026}
}

@article{girish2026comparing,
  title={Comparing classical and quantum conditional disclosure of secrets},
  author={Girish, Uma and May, Alex and Orshansky, Leo and Waddell, Chris},
  journal={Quantum},
  volume={10},
  pages={2049},
  year={2026},
  publisher={Verein zur F{\"o}rderung des Open Access Publizierens in den Quantenwissenschaften}
}

@article{speelman2015instantaneous,
  title={Instantaneous non-local computation of low T-depth quantum circuits},
  author={Speelman, Florian},
  journal={arXiv preprint arXiv:1511.02839},
  year={2015}
}

\end{document}